\documentclass[11pt]{amsart}
\usepackage{amsfonts,amssymb,mathtools,stmaryrd,graphicx,a4,bbm,mdframed,mathrsfs,hyperref}
\graphicspath{{./figures/}}
\usepackage{float}
\usepackage{tikz, tkz-euclide}
\usepackage{tikz-cd}
\usepackage[title]{appendix}
\usetikzlibrary{calc}
\usetikzlibrary{arrows}
\usetikzlibrary{positioning}
\usetikzlibrary{decorations.markings}
\usetikzlibrary{tikzmark}
\usetikzlibrary{tikzmark}
\usepackage[normalem]{ulem}
\usepackage{euscript} %Fontes EuScript 
\usepackage{color}\definecolor{Red}{cmyk}{0,1,1,0}
\usepackage{enumitem}

\newcommand{\Var}[0]{\text{Var}}
\newcommand\redsout{\bgroup\markoverwith{\textcolor{red}{\rule[0.5ex]{2pt}{0.4pt}}}\ULon}
\renewcommand{\restriction}{\mathord{\upharpoonright}}
\newtheorem{theorem}           {Theorem}

\theoremstyle{definition}
\newtheorem{remark}[theorem]{Remark}
\newtheorem*{theorem*}{Theorem}
\newtheorem*{conjecture*}   {Conjecture}
\newtheorem*{corollary*}   {Corollary}
  \newtheorem{lemma}[theorem]              {Lemma}
  \newtheorem*{lemma*}          {Lemma}
    \newtheorem*{claim*}          {Claim}
  \newtheorem{definition}[theorem]         {Definition}
  \newtheorem{corollary}[theorem]          {Corollary}
  \newtheorem{proposition}[theorem]      {Proposition}

  \theoremstyle{definition}
  \newtheorem{example}[theorem]          {Example}
  
\newcommand{\g}{\mathbb{G}}
\newcommand{\zd}{\mathbb{Z}^d}

\newcommand{\gr}{\mathfrak{T}}

\newcommand*{\QEDB}{\hfill\ensuremath{\square}}
\DeclareMathOperator{\diag}{diag}

\begin{document}
%\pagestyle{headings}

%%%%%%%%%%%%%%%%%       FORMATO

%\magnification=\magstep1\hoffset=0.cm
\voffset=-1.5truecm\hsize=16.5truecm    \vsize=24.truecm
\baselineskip=14pt plus0.1pt minus0.1pt \parindent=12pt
\lineskip=4pt\lineskiplimit=0.1pt      \parskip=0.1pt plus1pt

\def\ds{\displaystyle}\def\st{\scriptstyle}\def\sst{\scriptscriptstyle}

%%%%%%%%%%%%%%%% GRECO

%\let\a=\alpha \let\b=\beta
%\let\c=\chi
%\let\d=\delta \let\e=\varepsilon
%\let\f=\varphi \let\g=\gamma \let\h=\eta    \let\k=\kappa \let\l=\lambda
%\let\m=\mu \let\n=\nu \let\o=\omega    \let\p=\pi \let\ph=\varphi
%\let\r=\rho \let\s=\sigma \let\t=\tau \let\th=\vartheta
%\let\y=\upsilon \let\x=\xi \let\z=\zeta
%\let\D=\Delta \let\F=\Phi \let\G=\Gamma \let\L=\Lambda \let\Th=\Theta
%\let\O=\Omega \let\P=\Pi \let\Ps=\Psi \let\Si=\Sigma \let\X=\Xi
%\let\Y=\Upsilon
%%%%%%%%%%%%%%%% EQUAZIONI CON NOMI SIMBOLICI
%%% per assegnare un nome simbolico ad una equazione basta
%%% scrivere \Eq(...) o, in \eqalignno, \eq(...) o,
%%% nelle appendici, \Eqa(...) o \eqa(...):
%%% dentro le parentesi e al posto dei ...
%%% si puo' scrivere qualsiasi commento; per avere i nomi
%%% simbolici segnati a sinistra delle formule si deve
%%% dichiarare il documento come bozza, iniziando il testo con
%%% \BOZZA. Sinonimi \Eq,\EQ.
%%% All' inizio di ogni paragrafo si devono definire il
%%% numero del paragrafo e della prima formula dichiarando
%%% \numsec=... \numfor=... (brevetto Eckmannn).

\global\newcount\numsec\global\newcount\numfor
\gdef\profonditastruttura{\dp\strutbox}
\def\senondefinito#1{\expandafter\ifx\csname#1\endcsname\relax}
\def\SIA #1,#2,#3 {\senondefinito{#1#2}
\expandafter\xdef\csname #1#2\endcsname{#3} \else
\write16{???? il simbolo #2 e' gia' stato definito !!!!} \fi}
\def\etichetta(#1){(\veroparagrafo.\veraformula)
\SIA e,#1,(\veroparagrafo.\veraformula)
 \global\advance\numfor by 1
% \write15{@def@equ(#1){\equ(#1)} \%:: ha simbolo= #1 }
 \write16{ EQ \equ(#1) ha simbolo #1 }}
\def\etichettaa(#1){(A\veroparagrafo.\veraformula)
 \SIA e,#1,(A\veroparagrafo.\veraformula)
 \global\advance\numfor by 1\write16{ EQ \equ(#1) ha simbolo #1 }}
\def\BOZZA{\def\alato(##1){
 {\vtop to \profonditastruttura{\baselineskip
 \profonditastruttura\vss
 \rlap{\kern-\hsize\kern-1.2truecm{$\scriptstyle##1$}}}}}}
\def\alato(#1){}
\def\veroparagrafo{\number\numsec}\def\veraformula{\number\numfor}
\def\Eq(#1){\eqno{\etichetta(#1)\alato(#1)}}
\def\eq(#1){\etichetta(#1)\alato(#1)}
\def\Eqa(#1){\eqno{\etichettaa(#1)\alato(#1)}}
\def\eqa(#1){\etichettaa(#1)\alato(#1)}
\def\equ(#1){\senondefinito{e#1}$\clubsuit$#1\else\csname e#1\endcsname\fi}
\let\EQ=\Eq

%%%%%%%%%%%%%%% DEFINIZIONI LOCALI

\def\\{\noindent}
\let\io=\infty

\def\VU{{\mathbb{V}}}
\def\EE{{\mathbb{E}}}
\def\GI{{\mathbb{G}}}
\def\TT{{\mathbb{T}}}
\def\C{\mathbb{C}}
\def\LL{{\cal L}}
\def\RR{{\cal R}}
\def\SS{{\cal S}}
\def\NN{{\cal N}}
\def\HH{{\cal H}}
\def\GG{{\cal G}}
\def\PP{{\cal P}}
\def\AA{{\cal A}}
\def\BB{{\cal B}}
\def\FF{{\cal F}}
\def\vv{\vskip.2cm}
\def\gt{{\tilde\g}}
\def\E{{\mathcal E} }
\def\I{{\rm I}}

\def\cal{\mathcal}

\def\tende#1{\vtop{\ialign{##\crcr\rightarrowfill\crcr
              \noalign{\kern-1pt\nointerlineskip}
              \hskip3.pt${\scriptstyle #1}$\hskip3.pt\crcr}}}
\def\otto{{\kern-1.truept\leftarrow\kern-5.truept\to\kern-1.truept}}
\def\arm{{}}
\font\bigfnt=cmbx10 scaled\magstep1

%%%%%%%%%%%%%%% DEFINIZIONI ROBERTO
\newcommand{\card}[1]{\left|#1\right|}
\newcommand{\und}[1]{\underline{#1}}
\newcommand{\Dom}[0]{\text{Dom}}
\def\1{\rlap{\mbox{\small\rm 1}}\kern.15em 1}
\def\ind#1{\1_{\{#1\}}}
\def\bydef{:=}
\def\defby{=:}
\def\buildd#1#2{\mathrel{\mathop{\kern 0pt#1}\limits_{#2}}}
\def\card#1{\left|#1\right|}
\def\proofof#1{\noindent{\bf Proof of #1. }}
\def\trp{\mathbb{T}}
\def\trt{\mathcal{T}}

\def\bfz{\boldsymbol z}
\def\bfa{\boldsymbol a}
\def\bfalpha{\boldsymbol\alpha}
\def\bfmu{\boldsymbol \mu}
\def\bfmust{\bfT^\infty(\bfmu)}
\def\bfmupr{\boldsymbol {\widetilde\mu}}
\def\bfrho{\boldsymbol \rho}
\def\bfrhost{\boldsymbol \rho^*}
\def\bfrhopr{\boldsymbol {\widetilde\rho}}
\def\bfT{{\boldsymbol T}_{\!\!\bfrho}}
\def\bfR{\boldsymbol R}
\def\bfvarphi{\boldsymbol \varphi}
\def\bfvarphist{\boldsymbol \varphi^*}
\def\bfPi{\boldsymbol \Pi}
\def\bfzero{\boldsymbol 0}
\def\bfW{\boldsymbol W}
\def\formal{\stackrel{\rm F}{=}{}}
\def\eee{{\rm e}}
\def\nnn{\mathcal N}
\def\nst{\nnn^*}
\def\Var{\text{Var}}

%\BOZZA
\thispagestyle{empty}

\begin{center}
{\LARGE Gibbs Measures on Multidimensional Spaces.\\ Equivalences and a Groupoid Approach.}
\vskip.5cm
Rodrigo Bissacot$^{1,2}$, Bruno Hideki Fukushima-Kimura$^{1,3}$,  Rafael Pereira Lima$^{4}$ and Thiago Raszeja$^{5}$
\vskip.3cm
\begin{footnotesize}
$^{1}$Institute of Mathematics and Statistics (IME-USP), University of S\~{a}o Paulo, Brazil\\
$^{2}$Faculty of Mathematics and Computer Science, Nicolaus Copernicus University, Poland\\
$^{3}$ Faculty of Science, Hokkaido University, Japan \\   
$^{4}$ School of Mathematics and Statistics, Victoria University of Wellington, New Zealand \\
$^{5}$ Faculty of Applied Mathematics, AGH University of Science and Technology, Poland
\end{footnotesize}
\vskip.1cm
\begin{scriptsize}
emails: rodrigo.bissacot@gmail.com; bruno.hfkimura@gmail.com; rplima\textunderscore 31@yahoo.com.br; tcraszeja@gmail.com
\end{scriptsize}

\end{center}

\def\be{\begin{equation}}
\def\ee{\end{equation}}

\vskip1.0cm
%\begin{abstract}
\begin{quote}
{\small

\textbf{Abstract.} \begin{footnotesize} 
We consider some of the main notions of Gibbs measures on subshifts introduced by different communities, such as dynamical systems, probability, operator algebras, and mathematical physics. For potentials with $d$-summable variation, we prove that several of the definitions considered by these communities are equivalent. In particular, when the subshift is of finite type (SFT), we show that all definitions coincide. In addition, we introduced a groupoid approach to describe some Gibbs measures, allowing us to show the equivalence between Gibbs measures and KMS states (the quantum analogous to the Gibbs measures).   
\end{footnotesize}

}
\end{quote}
%\end{abstract}
\tableofcontents
\numsec=2\numfor=1
\section*{Introduction}
\vv
\noindent

The theory of Gibbs measures is one among several successfully developed branches of mathematics motivated by ideas from physics. An expressive number of specialists in rigorous statistical mechanics, ergodic theory, symbolic dynamics, and operator algebras are or were engaged in topics related to some notion of Gibbsianess.

Historically, the first paper with a rigorous treatment on this subject dates back to \cite{Bogo:49} (see also \cite{Bogo:69} for a more up-to-date exposition) by N. N. Bogolyubov and B. I. Khatset. Following the same ideas presented in this paper, R. L. Dobrushin \cite{dob:68} and,  independently, O. E. Lanford III with D. Ruelle \cite{lanfordruelle:69} introduced the notion of Gibbsianess in the context of statistical mechanics by means of a system of prescribed conditional probabilities. Due to its physical and probabilistic interpretations, this approach is still widely adopted in mathematical physics and probability theory; moreover, such measures are often referred to as DLR measures in honor of them. 

On the other hand, these papers, together with one by R. A. Minlos \cite{minlos:67}, motivated the study of Gibbs measures in the context of (differentiable) dynamical systems started by Ya. G.  Sinai \cite{sinai:72}. Sinai introduced Markov partitions and symbolic dynamics for Anosov diffeomorphisms, subjects for which R. Bowen made several contributions, such as his book \cite{bowen:08}, one of the well-known references for Gibbs measures in symbolic dynamics. Due to the influence of Ruelle, a notion of Gibbs measure was also introduced for $\mathbb{Z}^d$-actions by an expansive group of homeomorphisms on compact metrizable spaces by D. Capoccaccia in \cite{cap:76}. Ruelle also wrote one of the classical books \cite{ruelle:04} towards Gibbs measures focusing on dynamical systems; the ergodic theory community regards this book (together with Bowen's) as the main reference of a subfield that is known as {\it thermodynamic formalism}.

%In $1976$, Capocaccia \cite{cap:76} gave a definition of a Gibbs measure in the general context of compact metrizable spaces where $\zd$ acts by an expansive group of homeomorphisms.

So, the notions of Gibbsianess developed by researchers concerned with the boundary between mathematical physics and dynamical systems, together with several papers and books published in the 1970s, provided a link between these areas. However, the approach adopted by these communities to deal with Gibbs measures split in two different ways: while the majority of probabilists and mathematical physicists follow Dobrushin's ideas, relying on conditional expectations and thermodynamic limits (see one of the classical modern references by Georgii \cite{georgii:11}), the dynamicists follow the approach introduced by Bowen, Ruelle, and Capocaccia.

Furthermore, certain works such as Aaronson and Nakada \cite{aaronson:07}, K. Schmidt and K. Petersen \cite{PetSch:97}, and Schmidt \cite{schmidt:97} studied an abstract notion of Gibbs measures for subshifts of finite type (SFTs) over finite alphabets whose connection with the previous definitions is not obvious, although they adopted the same nomenclature. In 2013, T. Meyerovitch \cite{meyerovitch:13} generalized this definition for multidimensional subshifts (over finite alphabets) and proved that for SFTs it is recovered the previous definition. He proved that equilibrium measures for a potential belonging to a suitable class of functions (functions with $d$-summable variation) are Gibbs measures, including shifts that are not of STF. %This result generalizes another one presented in \cite{ruelle:04} since it expands the class of potentials this result holds. 

The communities mentioned above know that all these definitions do not always coincide; see \cite{Fer} and \cite{sarig:15} for some examples from the probabilistic and dynamic points of view, respectively. Nevertheless, there exist classes of shifts and potentials for which the equivalence of some of these notions of Gibbsianess holds. 
One classical reference for positive results, showing the equivalence of some definitions, is provided by Keller's book \cite{keller:98}. Inspired by Capocaccia's definition of 
Gibbs measures he showed that a definition used by the dynamical systems community for the full shift over a finite alphabet in $\mathbb{Z}^d$ coincides with the notion of 
DLR measures for the class of potentials with $d$-summable variation. In particular, for the one-dimensional case, he proved that these measures are 
Gibbs in Bowen's sense. S. Muir \cite{muir, Muir_paper} extended the results obtained by Keller, showing that a natural extension of Capocaccia's and the DLR 
definitions coincide when the configuration space is $\mathbb{N}^{\mathbb{Z}^d}$. More recently, by using the fact that, in the one-dimensional case, the existence of the Ruelle operator (a standard tool in one-dimensional thermodynamic formalism) is ensured, L. Cioletti and A. O. Lopes \cite{Cio:14} 
showed the equivalence of some notions of Gibbs measures for Walters potentials defined on the full shift over a finite alphabet, such as DLR measures, 
measures constructed with the Ruelle operator, and thermodynamic limit measures.

Now, under the C$^*$-algebras perspective (or quantum setting),  the analogous thermodynamical object to a Gibbs measure is a Kubo-Martin-Schwinger (KMS) state \cite{Kubo1957,MartinSchwinger1959}, which is a linear functional of norm one that satisfies the so-called KMS condition. This condition was introduced by R. Kubo in 1957 \cite{Kubo1957}, and later in 1959, by P. C. Martin and J. Schwinger \cite{MartinSchwinger1959}, for the context of finite-volume Gibbs states. In 1967, R. Haag, N. M. Hugenholtz, and M. Winnink \cite{HaagHugWin1967} used this condition as a criterion to describe equilibrium systems that cover the case of infinite degrees of freedom, i.e., infinite-dimensional Hilbert spaces. In other words, the KMS states give the thermodynamical probabilistic description for equilibrium systems in the context of quantum mechanics and quantum statistical physics \cite{Bratteli1987vol1,Bratteli1981vol2}. Just like Gibbs measures, each KMS state depends on the inverse temperature $\beta$.

%Despite their physical origin, KMS states also have been extensively studied in mathematics due to their interesting mathematical properties. In the theory of groupoid C$^*$-algebras (see \cite{Renault:80} and \cite[Part II]{SSW2020}), J. Renault \cite[Proposition II.5.4]{Renault:80} shows that, for a C$^*$-algebra of a principal groupoid equipped with a particular dynamical system, there is a correspondence between the KMS states and the probability measures on the unit space of the groupoid $\mathcal{G}$ satisfying certain properties. Such measures are referred to as quasi-invariant in his book. Later, Neshveyev \cite[Theorem 1.3]{Neshveyev:2013} generalized Renault's characterization of KMS states to non-principal groupoids, where each KMS state on $C^*(\mathcal{G})$ corresponds to a pair consisting of a quasi-invariant probability measure and a family of states on $C^*(\mathcal{G}_x^x)$ satisfying some conditions. J. Christensen \cite[Theorem 6.3]{christensen2023structure} generalized Neshveyev's theorem to KMS weights, where the corresponding quasi-invariant measures are not required to be probabilities.

Despite their physical origin, KMS states also have been extensively studied in mathematics due to their interesting mathematical properties. In the theory of groupoid C$^*$-algebras (see \cite{Renault:80} and \cite[Part II]{SSW2020}), J. Renault \cite[Proposition II.5.4]{Renault:80} shows that, for a C$^*$-algebra of a principal groupoid equipped with a particular dynamical system, there is a correspondence between the KMS states and the probability measures on the unit space of the groupoid $\mathcal{G}$ satisfying certain properties. Such measures are referred to as quasi-invariant in his book. Later, Neshveyev \cite[Theorem 1.3]{Neshveyev:2013} generalized Renault's characterization of KMS states to non-principal groupoids, where each KMS state on $C^*(\mathcal{G})$ corresponds to a pair consisting of a quasi-invariant probability measure and a family of states on $C^*(\mathcal{G}_x^x)$ satisfying some conditions. J. Christensen \cite[Theorem 6.3]{christensen2023structure} generalized Neshveyev's theorem to KMS weights, where the corresponding quasi-invariant measures are not required to be probabilities. The class of C*-algebras of Renault-Deaconu groupoids \cite{ADelaroche:97,Deaconu:95,Renault:80,Renault1999} covers several interesting examples, including most graph algebras \cite[Proposition 4.1]{KPRR1997}, higher-rank graph algebras \cite{KumjianPask:00}, crossed products of commutative C*-algebras by the integers \cite{Deaconu:95} and Exel-Laca algebras \cite{Cuntz1977,CK1980,EL1999,Renault1999}, and there is a vast literature concerning the study of KMS states for these examples \cite{aHLRS:13,aHLRS:14,aHLRS:15, aHKR:15,aHKR:17,BEFR2022,ChristensenThomsen:21,ChristensenThomsen:22,ChristensenVaes:22,EFW:1984,Evans:80,Exel:04,EL2003,FGLP:21,FHKP:22,LiYang:19,OlesenPedersen:1978,Thomsen:12,Thomsen2016}.

For countable Markov shift spaces, both in standard \cite{sarig:09} and generalized \cite{EL1999} contexts, for a fixed potential, each conformal probability measure \cite{DU91,sarig:09} is a fixed point for the Ruelle's transformation \cite{BEFR2022}, and it generates a KMS state on their respective Renault-Deaconu groupoid via 1-cocycle dynamics. Under some conditions on the potential \cite{BEFR2022,KessStadStrat2007,renault2009} (for instance, when the potential is either strictly positive or strictly negative), every KMS state indeed comes from a conformal measure, because the kernel of the cocycle is a principal groupoid. For Renault-Deaconu groupoids (non-generalized), K. Thomsen \cite[Theorem 2.2]{Thomsen2016} gives a general description of all extremal KMS states on the C*-algebras of Renault-Deaconu groupoids, including KMS states that do not come from conformal probabilities.

The C$^*$-community also have studied the DLR measures from the point of view of approximately proper (AP) equivalence relations and using the groupoid theory \cite{BEFR2018,renault:05}, which is used as a natural algebraic structure to approach dynamical systems for operator algebras \cite{BEFR2022,Renault:80}. In particular, for Markov shift spaces, the equivalence relation groupoid associated to the DLR measures is a subgroupoid of the Renault-Deaconu groupoid.

%Beyond that, these ideas from physics were also extended to mathematical problems, where we can mention, for instance, the works on Cuntz-Krieger algebras \cite{Cuntz1977,CK1980} and Exel-Laca algebras \cite{EL1999} and their KMS states under a potential that generates a dynamics \cite{BEFR2022,Exel:04,EL2003,Evans:80,OlesenPedersen:1978,Thomsen2016}. The aforementioned algebras are groupoid C$^*$-algebras \cite{Renault:80,Renault1999,SSW2020}, whose connection with dynamical systems arises from their associated Markov shift space structure, called Renault-Deaconu groupoid \cite{BEFR2022,CK1980,EL1999}, and their C$^*$-dynamics are given by 1-cocycles.

%In many cases, these KMS states correspond to conformal measures \cite{DU91,sarig:09} on Markov shift spaces \cite{BEFR2022,Thomsen2016}, and in this case, these measures are fixed point measures for the Ruelle's operator \cite{BEFR2022}.

%In particular, the C$^*$-community also have been studied the DLR measures from the point of view of approximately proper (AP) equivalence relations and using the groupoid theory \cite{BEFR2018,renault:05}, which is used as a natural algebraic structure to approach dynamical systems for operator algebras \cite{BEFR2022,Renault:80}. For the Cuntz-Krieger and Exel-Laca algebras, the DLR measures are conformal measures on a subgroupoid of the Renault-Deaconu groupoid. 

This paper connects the different notions of Gibbs measures on subshifts of $\mathcal{A}^{\zd}$, where $\mathcal{A}$ is a finite alphabet and the potentials have $d$-summable variation. 
We prove the equivalences between them and connect these results to other equivalences proven in the literature, which can be summarized in the theorem below.

\begin{theorem} \label{main_theorem}
Let $X \subseteq \mathcal{A}^{\zd}$ be a subshift, and let $f:X\to \mathbb{R}$ be a potential with $d$-summable variation. In addition, let $\gr \subseteq X \times X$ and $\gr^{0} \subseteq \gr$ be respectively the Gibbs and the topological Gibbs equivalence relations (see Definitions \ref{gibbsrel_shift} and \ref{topgibbsrel_shift}), both endowed with the groupoid structure, and let $c_f: \gr \to \mathbb{R}$ be the 1-cocycle given by
\begin{equation*}
    c_f(x,y) := \sum\limits_{k \in \zd}f(\sigma^k y) - f(\sigma^k x).
\end{equation*}

For a Borel probability measure $\mu$ on $X$, the following conditions are equivalent:
\begin{itemize}
    \item[$(a)$] $\mu$ is a topological Gibbs measure, in the sense that $\mu$ is $(c_f,\gr^{0})$-conformal. More precisely, the Radon-Nikodym derivative $D_{\mu,\gr^{0}} = \frac{d\nu_r}{d\nu_s}$ is well-defined (see Theorem \ref{thm:nur_nus} and Definition \ref{def:radon_nikodym_nu_r_nu_s}) and
    \begin{equation*}
        D_{\mu,\gr^{0}} = e^{-c_f} \quad \mu\text{-a.e.};
    \end{equation*}

    \item[$(b)$] The state $\varphi_\mu:C^*(\gr^{0}) \to \mathbb{C}$ given by
        \begin{equation*}
            \varphi_\mu(h) = \int_{X} h d\mu, \quad h \in C_c(\gr^{0}),
        \end{equation*}
    is a KMS state for the C$^*$-dynamical system $(C^*(\gr^{0}),\tau)$, where $\tau = (\tau_t)_{t\in \mathbb{R}}$ is the one-parameter group of automorphisms given by $\tau_t(h) = e^{ic_ft}h$, $h \in C_c(\gr^0)$;
\end{itemize}

    In addition, the following are equivalent:
\begin{itemize}
  \item[$(c)$] $\mu$ is a Gibbs measure, in the sense that $\mu$ is $(c_f,\gr)$-conformal, i.e., the Radon-Nikodym derivative $D_{\mu,\gr} = \frac{d\nu_r}{d\nu_s}$ is well-defined and
    \begin{equation*}
        D_{\mu,\gr} = e^{-c_f} \quad \mu\text{-a.e.};
    \end{equation*}
    
    \item[$(d)$] The state $\varphi_\mu:C^*(\gr) \to \mathbb{C}$ given by
        \begin{equation*}
            \varphi_\mu(h) = \int_{X} h d\mu, \quad h \in C_c(\gr),
        \end{equation*}
    is a KMS state for the C$^*$-dynamical system $(C^*(\gr),\tau)$, where $\tau$ is the one-parameter group of automorphisms given in $(c)$, but now defined on $C_c(\gr)$;
    \newline

     \item[$(e)$] $\mu$ is a DLR measure, that is, let  $\mathcal{B}$ and $\mathcal{B}_{\Lambda^c}$ be the Borel $\sigma$-algebra of $X$ and the $\sigma$-algebra generated by the cylinders supported in $\Lambda^c$, respectively, then $\mu$ satisfies the DLR equations 
        \begin{equation*}
            \mu(A|\mathcal{B}_{\Lambda^c}) = \gamma_{\Lambda}(A|\,\cdot\,) \qquad \mu\text{-a.e.},
        \end{equation*}
    where $\Lambda \subseteq \zd$ is finite, and $\gamma_{\Lambda} : \mathcal{B}\times X\rightarrow {[}0,1]$ are defined by
    \begin{equation*}
        \gamma_{\Lambda}(A|\,x) = \lim\limits_{n \to \infty}\frac{\sum\limits_{\omega \in \mathcal{A}^{\Lambda}}e^{f_n(\omega x_{\Lambda^c})}\mathbbm{1}_{\{\omega x_{\Lambda^c}\in A\}}}{\sum\limits_{\eta\, \in \mathcal{A}^{\Lambda}}e^{f_n(\eta x_{\Lambda^c})}\mathbbm{1}_{\{\eta x_{\Lambda^c}\in X\}}}  
    \end{equation*}
    for each $A \in \mathcal{B}$ and each point $x \in X$; 
  
    \item[$(f)$] $\mu$ is a fixed point for the kernel maps $Q_{\Lambda_{n}}^{*}$, for every $n \in \mathbb{N}$, where $\Lambda_{n}:=(-n,n)^d \cap {\zd}$, $\mathcal{M}^+(X)$ is the set of all non-negative Borel measurable functions on $X$, $Q_{\Lambda_{n}}:\mathcal{M}^+(X) \to \mathcal{M}^+(X)$ is given by
    \begin{equation*}
        Q_{\Lambda_{n}}(g)(x) = \lim\limits_{n \to \infty}\frac{\sum\limits_{\omega \in \mathcal{A}^{\Lambda_{n}}}e^{f_n(\omega x_{\Lambda_{n}^c})}\mathbbm{1}_{\{\omega x_{\Lambda_{n}^c}\in X\}} g(\omega x_{\Lambda_{n}^c})}{\sum\limits_{\eta\, \in \mathcal{A}^{\Lambda_{n}}}e^{f_n(\eta x_{\Lambda_{n}^c})}\mathbbm{1}_{\{\eta x_{\Lambda_{n}^c}\in X\}}} 
    \end{equation*} 
    for each $x \in X$, and $Q_{\Lambda_{n}}^{*}(\mu)$ is given by
    \begin{equation*}
    \int g d Q_{\Lambda_{n}}^*(\mu) = \int Q_{\Lambda_{n}}(g) d \mu, \quad g \in \mathcal{M}^+(X).
    \end{equation*}
    \end{itemize}

Furthermore, if $X$ is a SFT, then all the items above are equivalent. In this case, the maps $Q_{\Lambda_{n}}$ can be defined as bounded operators on $C(X)$ and $Q_{\Lambda_{n}}^*$ is the dual operator of $Q_{\Lambda_{n}}$. In addition, we have one more definition that is equivalent to all the others, the Capocaccia's definition of Gibbs states:
    \begin{itemize}
    \item[$(g)$] $\mu$ is a Gibbs state for the family of multipliers given by $R_{(O,\varphi)}(x) = e^{c_f(x,\varphi(x))}$ at each point $x \in O$, for all pairs $(O,\varphi)$, where $O$ is an open set and $\varphi$ is a conjugating homeomorphism defined on $O$, see Definition \ref{def:Gibbs_state_Capocaccia}:
    \begin{equation*}
        \varphi_{\ast}\big(R_{(O,\varphi)}\,d\mu_{O}\big) = \mu_{\varphi(O)}
    \end{equation*}
    for every pair $(O,\varphi)$.
    \end{itemize}
\end{theorem}

\begin{remark}
    We remark that the groupoids $\gr$ and $\gr^{0}$ both have $X$ as the unit space. For that reason, in the statements of the items $(b)$ and $(d)$, the integrals are defined over $X$. In addition, when $X$ is a SFT, we have $\gr = \gr^{0}$, which explains why we have all the equivalences under this assumption.
\end{remark}

The paper is organized as follows.

%In section \ref{sec:Preliminaries}, we briefly introduce groupoids and their respective C$^*$-algebras. We also introduce the notion of approximately proper equivalence relation and we endow these relations with their natural groupoid structure. We compare different topologies on the equivalence relation groupoid that are used for different notions of Gibbsianess, namely the subspace topology and the inductive limit topology, and we prove that, although the first one is contained in the second one, these topologies generate the same $\sigma$-algebra. This result ensures all the different notions of measures studied in this paper are Borel measures on shift spaces endowed with its usual topology. In the same section, we present the notion of a shift space and the Gibbs equivalence relation, one of the central objects of the present work. We also show that the Gibbs relation is an approximately proper equivalence relation. In addition, as an example, we prove that the aforementioned inclusion relation between the topologies above is proper for the full shift. At the end of the section, we also show that the Gibbs relation is an amenable groupoid; therefore, the correspondent full and reduced groupoid C$^*$-algebras coincide.

In section \ref{sec:Preliminaries}, we briefly introduce groupoids and their respective C$^*$-algebras. We also introduce the notion of approximately proper equivalence relation and we endow these relations with their natural groupoid structure. We compare different topologies on the equivalence relation groupoid that are used for different notions of Gibbsianess, namely the subspace topology and the inductive limit topology, and we prove that, although the first one is contained in the second one, these topologies generate the same $\sigma$-algebra. This result ensures all the different notions of measures studied in this paper are Borel measures on shift spaces endowed with its usual topology. In the same section, we present the notion of a shift space and both Gibbs and topological Gibbs equivalence relations, which are central objects of the present work. We show that the Gibbs relation is approximately proper. In addition, as an example, we prove that the aforementioned inclusion relation between the topologies above is proper for the full shift. At the end of the section, we also show that both relations are amenable groupoids; therefore, their corresponding full and reduced groupoid C$^*$-algebras coincide. The proof for topological Gibbs relations is more delicate and needs more technical notions since is not evident that this groupoid is \'etale at first sight. Notions such as inductive limits of groupoids, $\mathcal{G}$-spaces, and other results are used to prove its etalicity and amenability.

%In particular, the case of topological Gibbs relation is more delicate, since in order to prove its amenability we needed to use extra results concerning groupoid theory, such as the notions of the inductive limit of sequences of groupoids, $\mathcal{G}$-spaces, and extra results concerning etalicity and amenability using these notions. 

%The proof for topological Gibbs relations is more delicate and needs more technical notions, such as inductive limits of groupoids, $\mathcal{G}$-spaces, and other results on \'etale and amenable groupoids.

Section \ref{rn} is dedicated to the notion of conformal measures on countable Borel equivalence relations on Polish spaces. This notion appears in dynamical systems and groupoid C$^*$-algebras sides of the theory. We proved that two different notions for quasi-invariant measures defined in each area are actually equivalent. Since these equivalence relations are generated by a countable subgroup of the group of Borel automorphisms on $X$, we prove the equivalence between the conformal measures and the Borel measures whose Radon-Nikodym derivatives of the pushforward of the elements of the generator group with respect to the measure satisfy the identity \eqref{xuxa}.

The thermodynamic formalism is presented in Section \ref{sec:therm_form}, where we begin the study of Gibbs measures for a specific class of functions, the so-called functions with $d$-summable variation (\cite{meyerovitch:13}, \cite{PetSch:97}, \cite{schmidt:97}) or
regular local energy functions (\cite{keller:98}, \cite{muir}). Adopting Meyerovitch's approach, we provide the definitions of a Gibbs measure and of a topological Gibbs measure for such a function $f$. It was shown in \cite{meyerovitch:13} that although we consider a Gibbs measure in two different ways, the second definition is a relaxed notion that coincides with the first one for subshifts of finite type (SFTs); moreover, every equilibrium measure for a function with $d$-summable variation is a topological Gibbs measure. In particular, if we suppose that we are dealing with a SFT, then every equilibrium measure is also a Gibbs measure. We connect the notion of a Gibbs measure provided by Meyerovitch with more familiar definitions adopted in the literature (e.g. \cite{cap:76}, \cite{velenik2017}, \cite{georgii:11}, \cite{ny:08}, \cite{sarig:09}, \cite{ruelle:04}). We proved that when dealing with subshifts of finite type, Meyerovitch's definition coincides with Capocaccia's notion of Gibbs states \cite{cap:76}. In order to connect our approach with the adopted by Georgii \cite{georgii:11}, for each subshift $X$ and each potential $f$ with $d$-summable variation, we defined a corresponding family $\gamma = (\gamma_{\Lambda})_{\Lambda \in \mathscr{S}}$ of proper probability kernels
satisfying the compatibility relation $\gamma_{\Delta}\gamma_{\Lambda} = \gamma_{\Delta}$ whenever $\Lambda \subseteq \Delta \subseteq \zd$, and proved 
that every Gibbs measure for $f$ satisfies the equation
\begin{equation}\label{dlrdlrdlr}
\mu(A|\mathcal{B}_{\Lambda^c}) = \gamma_{\Lambda}(A|\,\cdot\,){}
\end{equation}
for each Borel set $A$ of $X$ and each $\Lambda$ in $\mathscr{S}$. Conversely, if we suppose that $\mu$ is a probability measure that satisfies (\ref{dlrdlrdlr}),
then $\mu$ is a topological Gibbs measure for $f$. In particular, if $X$ is a SFT, then a probability measure $\mu$ is a Gibbs measure if and only if
$\mu$ satisfies (\ref{dlrdlrdlr}). The set of equations above is often referred to as DLR equations, named for Dobrushin, Lanford and Ruelle. Later in this section, we present the notions of a C$^*$-dynamical system, KMS state, and the 1-cocycle dynamics for groupoid C$^*$-algebras. We discuss, for the 1-cocycle dynamics generated by a $d$-summable potential, the equivalence between its KMS-states on the groupoid Gibbs relation C$^*$-algebra and Gibbs measures on the shift space, which is a consequence of Theorem 3.3.12 of \cite{renault2009}. Further, we also approach the specifications as $\sigma$-additive maps on the set of non-negative measurable functions on the subshift, similarly as done in \cite{BEFR2018}, with the difference that the AP relation now is given by projections on the lattice instead of being explicitly given by the dynamics as in \cite{BEFR2018}. By using the DLR measures theory from \cite{georgii:11}, it is possible to construct a family of specification maps such that a Borel measure is a Gibbs measure if and only if it is a fixed point of the family of the dual specification maps (in the integral sense). We compare this notion of DLR measure between the Markov shift space case and the subshift $\mathcal{A}^{\zd}$ case, and we show that the STF condition is sufficient for these maps to be considered operators on the continuous functions on the subshift, and therefore the dual maps are dual operators in the functional analysis sense, and the DLR measures are eigenmeasures of these dual operators with associated eigenvalue $1$.

A substantial amount of results presented in this paper are contained in the Master's thesis of Bruno Hideki Fukushima-Kimura \cite{Kimura}.

\numsec=2\numfor=1
\section{Preliminaries} \label{sec:Preliminaries}

\subsection{Groupoids and approximately proper equivalence relations} \label{subsec:groupoids_AP_relations}\\

In this subsection, we introduce the basic definitions and results regarding the groupoid theory for C$^\ast$-algebras. We refer to \cite{Renault:80,SSW2020} as the standard references on this subject. In this work, in particular, we are concerned with the groupoid structure of equivalence relations and notions such as approximately proper equivalence relations \cite{renault:05}, which comprise a general setting for exploring the so-called Gibbs relation and topological Gibbs relation.

\begin{definition} A \emph{groupoid} consists of a $4$-tuple $(\mathcal{G},\mathcal{G}^{(2)},\cdot,^{-1})$, where $\mathcal{G}$ is a set, $\mathcal{G}^{(2)} \subseteq \mathcal{G} \times \mathcal{G}$ is the composable parts set, $\cdot:\mathcal{G}^{(2)} \to \mathcal{G}$ is the product (or composition) operation, and $^{-1}:\mathcal{G} \to \mathcal{G}$ is the inverse operation, satisfying the following conditions:
\begin{itemize}
    \item[$(G1)$] $(g^{-1})^{-1} = g$ for every $g \in \mathcal{G}$;
    \item[$(G2)$] if $(g_1,g_2),(g_2,g_3) \in \mathcal{G}^{(2)}$, then $(g_1g_2,g_3),(g_1,g_2g_3) \in \mathcal{G}^{(2)}$, and in this case $g_1g_2g_3 := (g_1g_2)g_3 = g_1(g_2g_3)$;
    \item[$(G3)$] $(g,g^{-1}) \in \mathcal{G}^{(2)}$ for all $g \in \mathcal{G}$, and given $(g_1,g_2) \in \mathcal{G}^{(2)}$ it is true that
    \begin{equation*}
        (g_1^{-1}g_1)g_2 = g_1^{-1}(g_1g_2) = g_2 \quad \text{and} \quad g_1(g_2g_2^{-1}) = (g_1g_2)g_2^{-1} = g_1.
    \end{equation*}
\end{itemize}
\end{definition}

Unless we explicitly state the opposite, we always refer to a groupoid $(\mathcal{G},\mathcal{G}^{(2)},\cdot,^{-1})$ by its set $\mathcal{G}$. We define the unit space of a groupoid $\mathcal{G}$ as the set $\mathcal{G}^{(0)}:= \{g^{-1}g:g \in \mathcal{G}\} \subseteq \mathcal{G}$.
In addition, we introduce the range and source maps as the surjective maps $r:\mathcal{G} \to \mathcal{G}^{(0)}$ and $s:\mathcal{G} \to \mathcal{G}^{(0)}$ respectively defined by
\begin{equation*}
    r(g):= gg^{-1} \quad \text{and} \quad s(g):= g^{-1}g.
\end{equation*}
For every $x \in \mathcal{G}^{(0)}$, we define the \emph{$r$-fiber} and the \emph{$s$-fiber} of $x$ by letting $\mathcal{G}^x:=r^{-1}(\{x\})$ and $\mathcal{G}_x:=s^{-1}(\{x\})$, respectively.
The \emph{isotropy subgroupoid} of $\mathcal{G}$ is defined by Iso$(\mathcal{G}) := \{g \in \mathcal{G}: r(g) = s(g)\}$, and we say that $\mathcal{G}$ is \emph{principal} when Iso$(\mathcal{G}) = \mathcal{G}^{(0)}$.  
A groupoid $\mathcal{G}$ is said to be a \emph{topological groupoid} when it is endowed with a topology so that the composition and inverse operations are continuous, and consequently, the range and source maps are also continuous. Furthermore, a topological groupoid is said to be \emph{\'etale} when $r$ and $s$ are local homeomorphisms, i.e., every element $g \in \mathcal{G}$ has an open neighborhood $U$, called an \emph{open bisection}, such that $r(U), s(U)$ are open and the restrictions $r\vert_U: U \to r(U)$ and $s\vert_U: U \to s(U)$ are homeomorphisms. In this case, the family of open bisections forms a basis for its topology, and the subspace topologies of the $r$-fiber and $s$-fiber of $x$ coincide with the discrete topology, and these fibers are countable if we assume that $\mathcal{G}$ is also second-countable. We introduce the notion of (topologically) amenable groupoids as in \cite[Definition 2.2.8, Proposition 2.2.13]{renaultamenable} (see also \cite[Lemma 10.1.3]{SSW2020} and \cite[Definition 4.1.1]{renault2009}).

%{\color{red} 
%\begin{definition} Let $\mathcal{G}$ be a groupoid, a $1$-cocycle is a homomorphism $c:\mathcal{G} \to %\mathbb{R}$, where $\mathbb{R}$ is regarded as the additive group, that is, $c(gh) = c(g) + c(h)$, for every %$(g,h) \in \mathcal{G}^{(2)}$.
%\end{definition}
%}

\begin{definition}
A locally compact \'etale groupoid $G$ is \emph{topologically amenable} if there exists a sequence $(f_i)$ in $C_c(G)$ such that
\begin{itemize}
\item[$(a)$] the maps $x \mapsto \sum_{g \in G^x} \vert f_i(g) \vert^2$ (indexed by $i$) converge uniformly to $1$ on every compact subset of $G^{(0)}$, and

\item[$(b)$] the maps $h \mapsto \sum_{g: r(g) = r(h)} \vert f_i(h^{-1} g) - f_i(g) \vert$ (indexed by $i$) converge uniformly to $0$ on every compact subset of $G$.
\end{itemize}
\end{definition}

For the purposes of this paper, our main example of a groupoid is constructed from an equivalence relation, the so-called Gibbs relation, and in particular, we are interested in the subshift case \eqref{gibbsrel_shift}. The general notion of a groupoid generated by an equivalence relation is presented next.

\begin{example}\label{exa:equiv_relation_groupoid} Let $X$ be an arbitrary set and $R \subseteq X \times X$ an equivalence relation. We endow $R$ with a groupoid structure by setting $R^{(2)} := \{((x,y),(w,z)) \in R \times R : y = w \},$ with the product and inverse operations respectively given by
\begin{equation*}
    (x,y) \cdot (y,z) := (x,z) \quad \text{and} \quad (x,y)^{-1} := (y,x).
\end{equation*}
In this case, the set of units is $R^{(0)} = \{ (x, x): x \in X \}$, and the range and source maps are given by $r(x,y) = (x,x)$ and $s(x,y) = (y,y)$. Through this paper, we make an abuse of notation and identify $R^{(0)}$ with $X$ by the bijection $(x,x) \mapsto x$.
\end{example}

In the example above, $R^{(2)}$ and the product operation encode the transitivity of $R$, the inverse operation encodes the symmetry property, and $R^{(0)}$ encodes the reflexivity.    

In order to study the topological Gibbs relation $\gr^0$, we use the notions of inductive limit groupoid and $\mathcal{G}$-spaces \cite{renaultamenable,SSW2020}, and in special the Proposition \ref{prop:amenability_via_G_spaces}. In addition, we prove some technical auxiliary results in order to prove that $\gr^0$ is an \'etale groupoid.

\begin{definition}[Inductive limit groupoid] Let $\mathcal{G}$ be a topological groupoid  and $\mathcal{G}_n$, $n \in \mathbb{N}$, be an increasing sequence of subgroupoids of $\mathcal{G}$. We say that $\mathcal{G}$ is the \emph{inductive limit} of the sequence $\{\mathcal{G}_n\}_{n \in \mathbb{N}}$ when
\begin{enumerate}
    \item $\mathcal{G} = \bigcup_{n \in \mathbb{N}}\mathcal{G}_n$;
    \item $\mathcal{G}^{(0)} = \mathcal{G}_n^{(0)}$, for every $n \in \mathbb{N}$;
    \item the topology of $\mathcal{G}$ is the inductive limit topology with respect to the sequence $\{\mathcal{G}_n\}_{n \in \mathbb{N}}$, that is,
    \begin{equation*}
        U \subseteq \mathcal{G} \text{ is open } \iff \text{ } U \cap \mathcal{G}_n \text{ is open for every $n \in \mathbb{N}$ (subspace topology)}.  
    \end{equation*}
\end{enumerate}
\end{definition}

\begin{lemma} \label{lemma:inductive_limit_groupoid_etalicity} Let be an inductive limit $\mathcal{G} = \bigcup_{n \in \mathbb{N}}\mathcal{G}_n$ with respect to the increasing sequence $\{\mathcal{G}_n\}_{n \in \mathbb{N}}$ of locally compact \'etale subgroupoids. Then, $\mathcal{G}$ is \'etale.    
\end{lemma}

\begin{proof} Let $g \in \mathcal{G}$ and $V$ be an open neighborhood of $g$ with compact closure. Since the family $\{\mathcal{G}_n\}_{n \in \mathbb{N}}$ covers $V$, there exists an $m \in \mathbb{N}$ such that $V \subseteq \mathcal{G}_m$. Let $\mathcal{U}$ be an open bisection in $\mathcal{G}_m$ such that $g \in \mathcal{U} \subseteq V$. Then, $r(\mathcal{U})$ and $s(\mathcal{U})$ are open sets in $\mathcal{G}^{(0)}$, and $r\vert_\mathcal{U}: \to r(\mathcal{U})$ and $s\vert_\mathcal{U}: \to s(\mathcal{U})$ are homeomorphisms. Therefore $\mathcal{U}$ is an open bisection of $\mathcal{G}$, and since $h$ is arbitrary, we have that $\mathcal{G}$ is \'etale.    
\end{proof}

\begin{definition}[$\mathcal{G}$-spaces]\label{def:G_spaces} Let $\mathcal{G}$ be an \'etale groupoid. A \emph{left $\mathcal{G}$-space} is a locally compact Hausdorff space $X$ endowed with a continuous map $r:X \to \mathcal{G}^{(0)}$ and a continuous map $\cdot: \mathcal{G} \ast X \to X$, where
\begin{equation*}
    \mathcal{G} \ast X := \{(g,x) \in \mathcal{G} \times X: s(g) = r(x)\},
\end{equation*}
and such that the following properties hold:
\begin{itemize}
    \item[$(a)$] $r(g\cdot x) = r(g)$, for every $(g,x) \in \mathcal{G} \ast X$;
    \item[$(b)$] $r(x) \cdot x = x$, for every $x \in X$;
    \item[$(c)$] $g \cdot (h \cdot x) = (gh) \cdot x$, for every $(h,x) \in  \mathcal{G} \ast X$ and every $g \in \mathcal{G}_{r(h)}$.
\end{itemize}
In addition, we say that $X$ is a proper $\mathcal{G}$-space when the map $\Theta:  \mathcal{G} \ast X \to X \times X$, given by $\Theta(g,x):= (g \cdot x,x)$ is proper, i.e., if the preimages of compact sets under $\Theta$ are compact.
\end{definition}

\begin{proposition}\label{prop:amenability_via_G_spaces} Suppose $\mathcal{G}$ is an \'etale groupoid, and $\{\mathcal{G}_n\}_{n \in \mathbb{N}}$ is a family of closed subgroupoids of $\mathcal{G}$ that satisfies the following properties:
\begin{itemize}
    \item[$(i)$] $\mathcal{G}^{(0)} \subseteq \mathcal{G}_n \subseteq \mathcal{G}_{n+1}$, for each $n \in \mathbb{N}$;
    \item[$(ii)$] $\mathcal{G}_{n+1}$ is a proper $ \mathcal{G}_n$-space, for each $n \in \mathbb{N}$;
    \item[$(iii)$] $\mathcal{G} = \bigcup_{n \in \mathbb{N}} \mathcal{G}_n$.
\end{itemize}
If each $\mathcal{G}_n$ is amenable, then $\mathcal{G}$ is amenable.
\end{proposition}

\begin{proposition}\label{prop:amenability_closed_or_open_subgroupoid} Let $\mathcal{G}$ be an amenable \'etale groupoid and $\mathcal{H} \subseteq \mathcal{G}$ be subgroupoid. If $\mathcal{H}$ is open or closed, then $\mathcal{H}$ is amenable.
\end{proposition}

Now, we briefly describe the groupoid full C$^*$-algebra. Let us restrict ourselves to the level of generality of locally compact Hausdorff (LCH) second-countable \'{e}tale groupoids. We omit any proof concerning the well-definedness of these algebras, which can be found, for instance, in \cite{Lima2019,SSW2020,Renault:80}. Note that there are C$^*$-algebras for more general classes of groupoids. Renault \cite{Renault:80} defines C$^*$-algebras for non-\'etale groupoids, and Exel and Pitts \cite{exelpitts:22} study non-Hausdorff groupoids.

%%%%%%%%%% Rafael e Thiago discutiram correções até aqui NÃO APAGAR ESTA LINHA

In the following, we denote by $C_c(\mathcal{G})$ the set of all complex-valued continuous functions on $G$ with compact support.

\begin{definition}\label{def:operations_C_star_groupoid} Let $\mathcal{G}$ be a locally compact Hausdorff second-countable \'{e}tale groupoid. We define the involution $*:C_c(\mathcal{G}) \to C_c(\mathcal{G})$ and convolution product $\cdot :C_c(\mathcal{G}) \times C_c(\mathcal{G}) \to C_c(\mathcal{G})$ by
\begin{equation*}
    f^*(g):= \overline{f(g^{-1})} \quad \text{and} \quad (f_1 \cdot f_2 )(g) = \sum_{g_1g_2 = g} f_1(g_1)f_2(g_2),
\end{equation*}
for every $f, f_1, f_2 \in C_c(\mathcal{G})$ and $g \in \mathcal{G}$.
\end{definition}

Let us regard $C_c(\mathcal{G})$ as a vector space equipped with the usual operations of addition and product by a scalar. This space, endowed with the convolution product and involution as above, is a $*$-algebra over $\mathbb{C}$. In addition, the function $\|\cdot\|:C_c(\mathcal{G}) \to \mathbb{R}$ defined by letting
\begin{align*}
\Vert f \Vert = \sup \lbrace \Vert \pi(f) \Vert \text{$:$ } \pi: C_c(\mathcal{G}) \rightarrow B(H_\pi) \text{ is a $\ast$-representation of $C_c(\mathcal{G})$} \rbrace
\end{align*}
for each $f \in C_c(\mathcal{G})$, turns $C_c(\mathcal{G})$ into a normed $*$-algebra whose completion is a C$^*$-algebra.

\begin{definition}[Full groupoid C$^*$-algebra] Let $\mathcal{G}$ be a locally compact Hausdorff second-countable \'{e}tale groupoid. The full groupoid C$^*$-algebra of $\mathcal{G}$, denoted by $C^*(\mathcal{G})$, is defined as the norm completion of $C_c(\mathcal{G})$.
\end{definition}

\begin{remark} The full groupoid C$^*$-algebra, defined as above, is separable.    
\end{remark}

At this point, we should address an important comment about the definition above concerning the word \emph{full}: in contrast to the full C$^*$-algebra, there is also a notion called groupoid reduced C$^*$-algebra (see \cite{SSW2020}). In general, the reduced one is a C$^*$-subalgebra of the full one. However, in subsection \ref{subsec:Gibbs_relation_properties}, we show that the Gibbs relation is amenable, and therefore both full and reduced C$^*$-algebras coincide in our context, see \cite[Theorem 10.1.4]{SSW2020} \cite{renaultamenable}. In order to prove that, we make use of the theory of approximately finite groupoids (see \cite{Farsi:19,Renault:80}). Some of the essential definitions and remarks used in this paper are presented as follows.

\begin{definition}[\cite{Farsi:19}] A groupoid $\mathcal{G}$ is said to be \emph{ample} when it is \'etale and $\mathcal{G}^{(0)}$ is zero-dimensional. 
\end{definition}

%\begin{definition}[\cite{Farsi:19}] Let $\mathcal{G}$ be a Hausdorff \'etale groupoid, $X$ a locally compact Hausdorff space, and $\psi:X \to \mathcal{G}^{(0)}$ a local homeomorphism. The \emph{ampliation} $\mathcal{G}^\psi$ of $\mathcal{G}$ corresponding to $\psi$ is the groupoid
%\begin{equation*}
 %   \mathcal{G}^\psi := \{(x,\gamma,y) \in X \times \mathcal{G} \times X : \psi(x)= r(\gamma) \text{ and } \psi(y) = s(\gamma)\},
%\end{equation*}
%where the set of composable parts is
%\begin{equation*}
 %   (\mathcal{G}^\psi)^{(2)} = \{((x,\gamma_1,y),(z,\gamma_2, w)) \in \mathcal{G}^\psi \times \mathcal{G}^\psi : y=z\},
%\end{equation*}
%and the composition and inverse operations are given respectively by $(x,\gamma_1,y)(y,\gamma_2, w) = (y,\gamma_1\gamma_2, w)$ and $(x,\gamma,y)^{-1} = (y, \gamma^{-1}, x)$.
%\end{definition}

%\begin{remark} The ampliation $\mathcal{G7}^\psi$ endowed with the subspace topology of $X \times \mathcal{G} \times X$, is a Hausdorff \'{e}tale groupoid.    
%\end{remark}

The main groupoid studied in this paper is an example of approximately finite groupoids, which we define next. These groupoids are approximated by groupoids of the form $R(\psi)$, described as follows:  given $X$ and $Y$ locally compact Hausdorff spaces, and a local homeomorphism $\psi:Y \to X$, we define the set
\begin{equation*}
    R(\psi) := \{(y,z) \in Y \times Y : \psi(y) = \psi(z)\}.
\end{equation*}
Note that $R(\psi)$ is an equivalence relation. We equip the set with the operations of Example \ref{exa:equiv_relation_groupoid}, which makes this set a groupoid. We also endow the groupoid with the subspace topology from $X \times X$, such that this groupoid is a locally compact, Hausdorff, \'etale groupoid. Moreover, if $X$ is zero-dimensional, then $R(\psi)$ is ample. Since the map $R(\psi)^{(0)} \to X$, given by $(x,x) \mapsto x$, is a homeomorphism, we identify $X$ with $R(\psi)^{(0)}$.

\begin{lemma}\label{lemma:clopen_subgroupoid_of_etale_groupoid_is_etale_same_unity_set} Let $\mathcal{G}$ be a Hausdorff ample groupoid, and $\mathcal{H} \subseteq \mathcal{G}$ be a clopen subgroupoid such that $X:= \mathcal{H}^{(0)} = \mathcal{G}^{(0)}$, and $X$ is separable metrizable. Then, $\mathcal{H}$ is \'etale (and hence ample).    
\end{lemma}

\begin{proof} Let $g \in \mathcal{H}$. Since $\mathcal{G}$ is \'etale, there exists an open bisection $\mathcal{U} \subseteq \mathcal{G}$ that contains $g$. Let $\{B_\lambda\}_{\lambda \in \Lambda}$ be a basis of clopen sets of $X$. Observe that these sets are clopen in $\mathcal{G}$ because $X$ is clopen due to Hausdorff property and the etalicity of $\mathcal{G}$ (respectively Lemmas 8.3.2 and 8.4.2 of \cite{SSW2020}). We have that
\begin{equation*}
    r(\mathcal{U}) = \bigcup_{\lambda \in \Lambda'} B_\lambda,
\end{equation*}
for some non-empty $\Lambda' \subseteq \Lambda$. Hence, there exists $\lambda \in \Lambda'$ such that $r(g) \in B_\lambda$. Moreover, given the homeomorphisms $r\vert_\mathcal{U}:\mathcal{U} \to r(\mathcal{U})$ and $s\vert_\mathcal{U}:\mathcal{U} \to s(\mathcal{U})$, we have that $\mathcal{V}:= r\vert_\mathcal{U}^{-1}(B_\lambda) \subseteq \mathcal{U}$ is a clopen bisection that is a neighborhood of $g$. Since $\mathcal{H}$ is clopen, we have that $\mathcal{V} \cap \mathcal{H}$ is a clopen bisection in $\mathcal{G}$ that is contains $g$
and hence $r(\mathcal{V} \cap \mathcal{H})$ and $s(\mathcal{V} \cap \mathcal{H})$ are clopen in $\mathcal{G}$. Therefore, $r\vert_{\mathcal{V} \cap \mathcal{H}}$ and $s\vert_{\mathcal{V} \cap \mathcal{H}}$ are continuous open invertible maps, i.e., they are homeomorphisms. Therefore, since $g$ is arbitrary, we have that $r$ and $s$ are local homeomorphisms on $\mathcal{H}$, i.e., $\mathcal{H}$ is \'etale.
\end{proof}

\begin{remark} In the lemma above, we only used that $X$ is separable and metrizable to ensure that the distinct definitions of zero-dimensionality be equivalent. If we use the definition with respect to the small inductive dimension (i.e., the existence of a basis of clopen sets), one can remove these hypotheses.% In addition, we use the Hausdorff property to ensure that $X$ is clopen, then the closed (resp. open) sets in $X$ are also closed (resp. open) sets in $\mathcal{G}$ (and then $\mathcal{H}$). 
\end{remark}

\begin{definition}[\cite{Farsi:19}] An ample groupoid $\mathcal{G}$ is said to be \emph{elementary} when it is isomorphic to $R(\psi)$ for some local homeomorphism $\psi: Y \to X$ between two locally compact Hausdorff zero-dimensional spaces. A groupoid $\mathcal{G}$ is said to be \emph{approximately finite (AF)} when it can be expressed as an increasing union of open elementary subgroupoids with the same unit space.
\end{definition}

\begin{remark} The definition of AF-groupoid above is equivalent to the notion of the same name as in Definition 1.1 of \cite{Renault:80}. A proof of this equivalence can be found in \cite{Rafa:23}. Then from this point, we use results on both definitions without distinction.
\end{remark}

As we mentioned earlier in this section, the main groupoid of this paper originates from the Gibbs equivalence relation, and, as we prove in Subsection \ref{subsec:Gibbs_relation_properties}, this equivalence relation is approximately proper. This class of relations was studied by J. Renault \cite{renault:05} in the context of the so-called Radon-Nikodym problem: for a groupoid $\mathcal{G}$ on a space $X$, and a cocycle $D: \mathcal{G} \to \mathbb{R}_+^*$, what are the quasi-invariant measures on $X$ with respect to $\mathcal{G}$ such that the correspondent Radon-Nikodym derivative is the cocycle $D$? A generalization of this problem can be found in \cite{BEFR2018} and it covers dynamical systems where the dynamic map is not defined in the whole space. As we present in section \ref{rn}, we are interested in cocycles in the form $D = e^c$, where $c$ is a homomorphism from a groupoid to the additive group of the real numbers. In this paper, we refer to $c$ as the cocycle. Next, we present the notions of proper and approximately proper equivalence relations.

In the following, we assume that $X$ is a locally compact Hausdorff second-countable space. An equivalence relation $R \subseteq X \times X$ is said to be \emph{proper and \'etale} when the topology of its quotient space is Hausdorff and its quotient map $q:X \to X/R$ is a local homeomorphism. Since we will mostly be dealing with \'etale groupoids derived from proper and \'etale equivalence relations, we may use the nomenclature ``proper'' and ``proper and \'etale'' interchangeably.

\begin{definition}\label{def:AP_relation} 
Let $R \subseteq X \times X$ be an equivalence relation. We say that $R$ is \emph{approximately proper} (AP) when there exists a sequence

\begin{equation*}
    \begin{tikzcd}
    X_0 \arrow[r, "q_{1,0}"] 
    & X_1 \arrow[r, "q_{2,1}"] 
    & X_2 \arrow[r, "q_{3,2}"] 
    &\cdots \arrow[r, "q_{n-1,n-2}"] 
    & X_{n-1} \arrow[r, "q_{n,n-1}"] 
    & X_n \arrow[r, "q_{n+1,n}"]
    & \cdots
    %
    %X_0 \stackrel{q_{1,0}}{\longrightarrow} X_1 \stackrel{q_{2,1}}{\longrightarrow} X_2 \stackrel{q_{3,2}}{\longrightarrow} \cdots \stackrel{q_{n-1,n-2}}{\longrightarrow} X_{n-1} \stackrel{q_{n,n-1}}{\longrightarrow} X_n \stackrel{q_{n+1,n}}{\longrightarrow} \cdots,
    \end{tikzcd}
\end{equation*}
where $X_0 = X$, and for each $n \in \mathbb{N}$, $X_n$ is a Hausdorff space, $q_{n,n-1}:X_{n-1} \to X_n$ is a surjective local homeomorphism, and
\begin{equation*}
    R = \{(x,y) \in X \times X: \exists n \in \mathbb{N} \text{ s.t. } q_n(x) = q_n(y)\},
\end{equation*}
where $q_n:= q_{n,n-1} \circ q_{n-1,n-2} \circ \cdots \circ q_{1,0}$.
\end{definition}

For an AP relation $R$ as above and for each $n \in \mathbb{N}$, the map $q_n$ defines a proper equivalence relation $R_n \subset R$, where $R_n:= \{(x,y) \in X \times X: q_n(x) = q_n(y)\}$, and we endow it with the subspace topology of the product topology in $X \times X$. The subspace topology of the product topology in $X \times X$ induced on the equivalence relation $R$ will be denoted by $\tau_{R}$, and its respective Borel $\sigma$-algebra will be denoted by $\mathcal{B}_{R}$. Note that, if $X$ is zero-dimensional, then every $R_n$ is an elementary groupoid, and therefore every AP relation is an AF groupoid.

The next proposition was proved by J. Renault and it will be useful for studying the properties of the Gibbs relation.

\begin{proposition}[Proposition 2.1 of \cite{renault:05}]\label{prop:properties_AP_relation} Let $R$ be the equivalence relation defined by the sequence $(q_n)_{n \in \mathbb{N}}$ as in Definition \ref{def:AP_relation}. Then, the following properties hold:
\begin{itemize}
    \item[$(i)$] $(R_n)_{n \in \mathbb{N}}$ is an increasing sequence and $R = \bigcup_{n \in \mathbb{N}}R_n$;
    \item[$(ii)$] for $m < n$, $R_m$ is a clopen subset of $R_n$;
    \item[$(iii)$] $R$ endowed with the inductive limit topology is an \'etale locally compact groupoid.
\end{itemize}
\end{proposition}

\begin{remark} \label{rmk:inductive_limit_topology} The inductive limit topology on an AP equivalence relation $R \subseteq X \times X$ is described as follows: $U \subseteq R$ is open if, and only if, $U \cap R_n$ is open for every $n \in \mathbb{N}$. We denote this topology by $\tau_{\lim}$ and its Borel $\sigma$-algebra by $\mathcal{B}_{\lim}$. 
\end{remark}

Because of the different topologies used to define the different notions of Gibbs measures, one important question concerns the compatibility of these theories, in the sense that these measures are defined on the same $\sigma$-algebra. Two distinct topologies for the same equivalence relation are used in this paper, namely, $\tau_{\lim}$ explained in Remark \ref{rmk:inductive_limit_topology}, and the subspace topology $\tau_R$ of the prodiscrete topology on $X\times X$. %conectar essa parte com onde usa e justificar que algumas noções são construídas mesmo em mesma definição, porém vindo de diferentes papers em diferentes áreas

\begin{proposition}\label{prop:topologies_comparison_R} Let $\tau_{R}$ and $\tau_{\lim}$ be, respectively the subspace topology of the product topology and the inductive limit topology both on $R$. Then $\tau_R \subseteq \tau_{\lim}$, and their respective Borel $\sigma$-algebras coincide.   
\end{proposition}

\begin{proof} Let $U \in \tau_{R}$. Then, there exists $V$ open subset of $X \times X$ such that $U = V\cap R$, and hence $R_n \cap U = R_n \cap V$ is an open subset of $R_n$ for every $n \in \mathbb{N}$, and then $\tau_R \subseteq \tau_{\lim}$, and also $\mathcal{B}_{R} \subseteq \mathcal{B}_{\lim}$. Now, we claim that $R_n$ is closed in $\tau_{R}$. In fact, in the product topologies of $X \times X$ and $X_n \times X_n$, the map $H:X \times X \to X_n \times X_n$ given by $H(x,y):= (q_n(x),q_n(y))$ is continuous, and since $X_n \times X_n$ is Hausdorff, $\diag(X_n \times X_n)$ is closed, and so is $R_n = H^{-1}(\diag(X_n \times X_n))$, and the claim is proved. For $U \in \tau_{\lim}$, we have that, for each $n \in \mathbb{N}$, that $U \cap R_n$ is open. Then there exists $V_n \subseteq X \times X$ open such that $V_n \cap R_n = U \cap R_n$, the intersection of an open set with a closed set, and therefore it belongs to $\mathcal{B}_{R}$, and so is $U$. Then, $\tau_{\lim} \subseteq \mathcal{B}_{R}$, and therefore $\mathcal{B}_{R} = \mathcal{B}_{\lim}$.
\end{proof}

%%%%%%%%%%%%%%%%%%%%%% Provisório %%%%%%%%%%%%%%%%%%%%%%%%%%
%\begin{definition}
%Given an \'etale groupoid $\mathcal{G}$, an open subset $W$ of $\mathcal{G}$ is called an \textit{open bisection} $if the maps $r$ and $s$, when restricted to $W$, are homeomorphisms onto their images.
%\end{definition}
%%%%%%%%%%%%%%%%%%%%%% Provisório %%%%%%%%%%%%%%%%%%%%%%%%%%

\subsection{Shift spaces}

In this section, we recall some of the basics about shift spaces, more specifically, subshifts and their essential properties. These objects are of central importance in the study of lattice spin systems in classical equilibrium statistical mechanics and dynamical systems. In this work, we restrict ourselves to multidimensional shifts over finite alphabets, but if the reader is interested in results concerning the case of countably infinite alphabets, see \cite{muir, Muir_paper,sarig:09, BEFR2022}.

In the following, we recall the definition of a full shift and its topological aspects. 
Given a positive integer $d$, let us consider the $d$-dimensional integer lattice $\zd$ endowed with its usual additive group structure.
Furthermore, we let $\|i\| := \max\limits_{1 \leq l \leq d}|i_l|$ for each point $i$ in $\zd$, and 
consider an exhausting sequence $(\Lambda_n)_{n \geq 0}$ of sets in $\zd$, where we call each $\Lambda_n$ a \emph{box} and it is defined by 
\begin{equation}
    \Lambda_n := \{i \in \zd : \|i\| < n\}
\end{equation}
for every non-negative integer $n$.

The \emph{$\zd$-full shift} over a finite alphabet $\mathcal{A}$ is defined as the set $\mathcal{A}^{\zd}$ of all functions from $\zd$ into $\mathcal{A}$. It is straightforward to show that the $\zd$-full shift  endowed with the prodiscrete topology is a compact metrizable space with the metric
\begin{equation}\label{metric}
\rho (x,y) = 
\begin{cases}
2^{-n(x,y)}  & \text{if} \; x \neq y ,\; \text{where} \; n(x,y) := \max\left\{n \geq 0 : x_{\Lambda_n} = y_{\Lambda_n}\right\}, \\
0  & \text{if} \; x=y;
\end{cases}
\end{equation}
where $x_\Lambda$ denotes the restriction of $x$ to a set $\Lambda \subseteq \zd$. 

\begin{remark}\label{box}
Note that, for every positive integer $n$ and for every elements $x$ and $y$ of $\mathcal{A}^{\zd}$, we have
\begin{equation}
\rho(x,y) \leq 2^{-n}\;\, \text{if and only if}\;\, x_{\Lambda_n} = y_{\Lambda_n}.
\end{equation}
\end{remark}

A typical element $x$ of $\mathcal{A}^{\zd}$ whose $i$-th coordinate is $x_i$ will often be expressed as $x = (x_i)_{i \in \zd}$.  Let us consider the action $j \mapsto \sigma^j$ on $\mathcal{A}^{\zd}$ that associates to each $j \in \zd$ a homeomorphism $\sigma^j : \mathcal{A}^{\zd} \rightarrow \mathcal{A}^{\zd}$ defined by $\sigma^j (x) = \left(x_{i+j}\right)_{i \in \zd}$. We refer to each $\sigma^j$ as the \emph{shift or translation by j}.

%%%%%%%%%%%%%%%%%%%%%%%%%%%%%%%%%%%%%%%%%%%%%

% NÃO ESQUECER DE BOTAR A CONCATENACAO NO LUGAR QUE VAI USAR

%%%%%%%%%%%%%%%%%%%%%%%%%%%%%%%%%%%%%%%%%%%%%%%%%%%%%%%

Now, we present the notion of subshift. A subset $X$ of $\mathcal{A}^{\zd}$ is said to be a \emph{subshift} if it is topologically closed and invariant under translations (i.e., the inclusion
$\sigma^{j}(X) \subseteq X$ holds for each $j$ in $\zd$).

%In this section, we will study some particular subsets of full shifts called shift spaces. These objects are most commonly referred to as subshifts and play an important role in the study of dynamical systems. For the one who is interested in studying one-dimensional subshifts in symbolic dynamics, we invite the reader to check the book by Lind and Marcus \cite{lind:95}. And, for the reader who is interested in how the study of shifts connects with statistical mechanics, we strongly recommend the books by Georgii \cite{georgii:11} and Keller \cite{keller:98} which are two masterpieces on this subject. 

%In the following, we will present the definition and basic properties of a subshift and turn to few examples. 

For $\Lambda \subseteq \zd$ finite, we will often refer to an element of $\mathcal{A}^{\Lambda}$ as a \emph{pattern} on $\Lambda$.
Given an arbitrary collection $\mathcal{F}$ of patterns, let us define the set $\mathsf{X}_{\mathcal{F}}$ as the set of all configurations in $\mathcal{A}^{\zd}$ in which the patterns in $\mathcal{F}$ do not appear, that is, 
\begin{equation}\label{subshift}
\mathsf{X}_{\mathcal{F}} := \left\{x \in \mathcal{A}^{\zd}: \sigma^{j}(x)_{\Lambda} \notin \mathcal{F}\; \text{for all}\; j \in \zd\;\text{and for every}\; \Lambda \subseteq \zd\; \textnormal{finite}\right\}.
\end{equation}
It is well-known that given an arbitrary $X \subseteq \mathcal{A}^{\zd}$, the set $X$ is a subshift if and only if it can be written in the form $X = \mathsf{X}_{\mathcal{F}}$ for some
collection $\mathcal{F}$ of patterns (see \cite{Kimura}). In particular, when $\mathcal{F}$ is finite, we call $X$ a \emph{Subshift of Finite Type} (SFT).

In the following, we introduce the \emph{Gibbs relation} or \emph{homoclinic equivalence relation} \cite{schmidt:97} defined as the set of all pairs of elements of the subshift that differ up to a finite number of sites, that is,
\begin{equation}\label{gibbsrel_shift}
\gr := \big\{(x,y) \in X \times X : x_{\Lambda^c} = y_{\Lambda^c}\;\text{for some}\;\Lambda \subseteq \zd\; \text{finite}\big\}.
\end{equation}
In addition, we also introduce the \emph{topological Gibbs relation} which is a subset of $\gr$ defined as follows
\begin{equation}\label{topgibbsrel_shift}
    \gr^{0} := \left\{(x,y) \in X \times X : y=\varphi(x)\;\text{for some}\;\varphi \in \mathcal{F}(X) \right\},
\end{equation}
where $\mathcal{F}(X)$ is the family of homeomorphisms on $X$ that does not change the configurations out of a box $\Lambda_n$ for some $n\in \mathbb{N}$, i.e.,
\begin{eqnarray*}
    \mathcal{F}(X) = \left\{\varphi \in \mathsf{Homeo(}X{\mathsf{)}}: \exists n \in \mathbb{N} \;\text{such that}\; \varphi(x)_{{\Lambda}_{n}^{c}}=x_{{\Lambda}_{n}^{c}}\; \text{for all}\; x \in X\right\}.	
\end{eqnarray*}

As we will see in the following sections, the Gibbs relation $\gr$ fits into the category of countable Borel equivalence relations and plays a central role in the definition of conformal/Gibbs measures. Analogously, the topological Gibbs relation $\gr^0$ is used to define the topological Gibbs measures.

\subsection{Groupoid approach for the Gibbs relations} \label{subsec:Gibbs_relation_properties}

%\begin{remark}
%   The open sets $W(n,m,V_1,V_2)$ such that $\sigma^n\vert_{V_1}$ and $\sigma^m\vert_{V_2}$ are injective are open bisections.
%\end{remark}

% We use some facts about \'etale groupoids, we refer to \cite{Deaconu1995, putnam, Renault1980, SimsSzaboWilliams2020} for this topic.\\%%% 
%\begin{remark}

We dedicate this subsection to approach the Gibbs relations $\gr$ and $\gr^0$ for a fixed subshift $X \subseteq \mathcal{A}^{\zd}$ from the point-of-view of the groupoid and AP relations theory. Here, we prove that $\gr$ is, in fact, an AP equivalence relation, and both $\gr$ and $\gr^0$ are amenable groupoids. The latter property ensures that the full groupoid C$^*$-algebra coincides with its reduced one for each groupoid. Moreover, we discuss the differences and similarities between the groupoid inductive limit topology and the subspace topology of the product topology of $X \times X$. Although these topologies do not coincide in general, they generate the same $\sigma$-algebra. 

First, we study the groupoid $\gr$.

\begin{proposition} The Gibbs relation $\gr \subseteq X \times X$ is a approximately proper relation.
\end{proposition}

\begin{proof} Set $X_0 = X$, and for each $n \in \mathbb{N}$, we take $X_n$ being the projections on each box complement, that is, $X_n:=\left \{x_{\Lambda_n^c}: x \in X\right\}$ endowed with the subspace topology of the prodiscrete topology on $\mathcal{A}^{\Lambda_n^c}$ and observe that each $X_n$ is a metric space, an hence they are Hausdorff spaces. Define $q_{n,n-1}:X_{n-1}\to X_n$ being the projection map to $\Lambda_n^c$, that is, $q_{n,n-1}(x_{\Lambda_{n-1}^c}):=x_{\Lambda_n^c}$. These maps are surjective, since the definition of $X_n$ gives that, for every $y \in X_n$, there exists $x \in X$ such that $x_{\Lambda_n^c} = y$ and $z = x_{\Lambda_{n-1}^c} \in X_{n-1}$, and then we have $q_{n,n-1}(z) = y$. Now, we show that $q_{n,n-1}$ is a local homeomorphism. Let $x \in X_{n-1}$ and the set $O = \bigcap_{i\in \Lambda_n\setminus\Lambda_{n-1}} \pi_{i,n-1}^{-1}(x_i)$, where $\pi_{i,n-1}:X_{n-1}\to \mathcal{A}$ is the canonical projection on the coordinate $i \in \Lambda_{n-1}^c$. We have that $O$ is an open set of $X_{n-1}$ containing $x$. We show that the restriction $q_{n,n-1}\vert_O:O \to q_{n,n-1}(O)$ is a homeomorphism. In fact, given $y,z \in O$, we have that $y_{\Lambda_{n-1}^c} = z_{\Lambda_{n-1}^c} = x_{\Lambda_{n-1}^c}$, and, if $q_{n,n-1}(y) = q_{n,n-1}(z)$, then $y_{\Lambda_n^c} = z_{\Lambda_n^c}$, and so we obtain $y = y_{\Lambda_{n-1}^c} = z_{\Lambda_{n-1}^c} = z$, and we conclude that $q_{n,n-1}\vert_O$ is injective. It is straightforward that $q_{n,n-1}$ is continuous, so is $q_{n,n-1}\vert_O$. Let a convergent sequence $(x^m)_{m \in \mathbb{N}}$ in $q_{n,n-1}(O)$ such that $x^m \to x \in q_{n,n-1}(O)$ as $m \to \infty$. Then, for each $j \in \Lambda_n^c$ $x^m_j \to x_j$ and then $q_{n,n-1}\vert_O^{-1}(x^m)_j \to q_{n,n-1}\vert_O(x)_j$, and since $q_{n,n-1}\vert_O^{-1}(x^m),q_{n,n-1}\vert_O^{-1}(x) \in O$, we have $q_{n,n-1}\vert_O^{-1}(x^m)_i = q_{n,n-1}\vert_O^{-1}(x)_i$ for every $i \in \Lambda_{n}\setminus\Lambda_{n-1}$, and so $q_{n,n-1}\vert_O^{-1}(x^m) \to q_{n,n-1}\vert_O^{-1}(x)$. Therefore, $q_{n-1,n}$ is a local homeomorphism for every $n \in \mathbb{N}$. It is straightforward to check that
\begin{equation*}
    \gr = \left\{(x,y) \in X \times X: \exists n \in \mathbb{N} \text{ s.t. } q_n(x) = q_n(y)\right\},
\end{equation*}
where $q_n = q_{n,n-1}\circ q_{n-1,n-2}\circ \cdots \circ q_{1,0}$.
\end{proof}

\begin{remark} \label{rmk:maps_q_n_and_compactness_of_I_n}In the proof above, $q_n:X \to X_n$ is the restriction of $X$ on the lattice $\Lambda_n^c$. In addition, we have that
\begin{equation*}
    \gr_n = \left\{(x,y) \in X \times X: q_n(x) = q_n(y)\right\} = \left\{(x,y) \in X \times X: x_{\Lambda_n^c} = y_{\Lambda_n^c}\right\},
\end{equation*}
for each $n \in \mathbb{N}$. Moreover, each $\gr_n$ is compact in both topologies $\tau_\gr$ and $\tau_{\lim}$, because the topology $\tau_{\lim}$ restricted to $\gr_n$ coincides with the topology $\tau_\gr$ and $\gr_n$ closed in $X\times X$, a compact space.
\end{remark}

%The next lemma is used only to avoid trivialities on $X$.

%\begin{lemma} Let $X$ be a subshift such that there is $w \in X$ satisfying $w_i \neq w_j$ for some $i,j \in \mathcal{A}^{\zd}$. Then, for every $n \in \mathbb{N}$, there exist two sequences $(x^m)_{m\in \mathbb{N}}$ and $(y^m)_{m\in \mathbb{N}}$ in $X$ such that $(x^m,y^m)_{m\in \mathbb{N}}$ is a sequence in $R_n^c$ that converges to $(w,w)$.   
%\end{lemma}

%comment about the sigma algebras

%In general, the topologies $\tau_{X\times X}$ and $\tau_{\lim}$ do not coincide, as it is shown in the next example.

By Proposition \ref{prop:topologies_comparison_R}, we have that $\tau_{\gr} \subseteq \tau_{\lim}$. In general, these topologies do not coincide, as we show in the next example.

\begin{example} Let $X = \mathcal{A}^{\zd}$, where $|\mathcal{A}| \geq 2$. For fixed $n \in \mathbb{N}$, we have by Proposition \ref{prop:properties_AP_relation} $(i)$ that, for every $m \leq n$, $\gr_n \cap \gr_m = \gr_m$ and therefore it is open in $\gr_m$. Similarly, for $m > n$ we have $\gr_n \cap \gr_m = \gr_n$, and by Proposition \ref{prop:properties_AP_relation} $(iii)$, $\gr_n$ is an open subset of $\gr_m$, and therefore $\gr_n \in \tau_{\lim}$.\footnote{Until this part the result is general.} We claim that $\gr_n \notin \tau_{X\times X}$. In fact, for fixed $n \in \mathbb{N}$, let $(x^m, y^m)_{m \in \mathbb{N}}$ be a sequence in $\gr_n^c$ constructed as follows: choose $x_i^m \neq y_i^m$ for every $i \in \Lambda^c_m$, and $x_i^m = y_i^m = x_i$ otherwise. Then $x^m$ and $y^m$ converge to $x$, and hence $(x^m,y^m)$ converges to $(x,x) \in \gr_n$, and then $\gr_n^c$ is not closed in $X \times X$, that is, $\gr_n \notin \tau_{X\times X}$.
\end{example}

Now we prove that the groupoid $\gr$ is amenable, and therefore its full and the reduced C$^*$-algebras coincide.

\begin{theorem} $\gr$ is an amenable groupoid. Hence, the full and the reduced C$^*$-algebras coincide.
\end{theorem}

\begin{proof} For every $n \in \mathbb{N}$, the groupoid $\gr_n$ is an elementary groupoid. In fact, $X$ and $X_n$ are compact Hausdorff topological spaces, and $q_n: X \to X_n$ is a local homeomorphism, and it is straightforward that $\gr_n = R(q_n)$, and therefore $\gr_n$ is elementary. Proposition \ref{prop:properties_AP_relation} $(iii)$ gives that $\gr_n$ is open for all $n \in \mathbb{N}$, and hence $R_n$ is an AF groupoid. By Proposition 1.15 of \cite{Renault:80}, $C^*(\gr)$ is an AF C$^*$-algebra, and hence it is an inductive limit of an increasing sequence of finite-dimensional C$^*$-algebras $\{U_n\}_{n \in \mathbb{N}}$. By Theorem 6.3.9 of \cite{Murphy:90}, each $U_i$ is nuclear, and so is $C^*(\gr)$ because of Theorem 6.3.10 of \cite{Murphy:90}. By Theorem 10.1.5 of \cite{SSW2020}, we conclude that $\gr$ is amenable and, therefore $C^*(\gr) \simeq C^*_r(\gr)$.
\end{proof}

Now, we study the groupoid $\gr^0$, which is a subgroupoid of $\gr$. For each $n \in \mathbb{N}$, we define
\begin{eqnarray*}
    \mathcal{F}_n(X) := \left\{\varphi \in \mathsf{Homeo(}X{\mathsf{)}}: \varphi(x)_{{\Lambda}_{n}^{c}}=x_{{\Lambda}_{n}^{c}}\; \text{for all}\; x \in X\right\},	
\end{eqnarray*}
and
\begin{equation*}
    \gr^{0}_n := \left\{(x,y) \in X \times X : y=\varphi(x)\;\text{for some}\;\varphi \in \mathcal{F}_n(X) \right\}.
\end{equation*}
Observe that $\mathcal{F}_n(X) \subseteq \mathcal{F}_{n+1}(X)$ for every $n \in \mathbb{N}$, and then $\gr^{0}_n \subseteq \gr^{0}_{n+1}$. In addition, each $\gr^{0}_n$ is a subgroupoid of $\gr^{0}$, and $(\gr^{0}_n)^{(0)} = (\gr^{0})^{(0)} = X$. Furthermore, the following identities hold:
\begin{equation*}
    \mathcal{F}(X) = \bigcup_{n \in \mathbb{N}} \mathcal{F}_n(X) \quad \text{and} \quad \gr^{0} = \bigcup_{n \in \mathbb{N}} \gr^{0}_n.
\end{equation*}
We endow $\gr^{0}$ with the inductive limit topology, where each $\gr^{0}_n$ is endowed with the subspace topology of $\gr$. Equivalently, the topology of $\gr^{0}$ is the subspace topology of $\gr$ when this is endowed with the topology $\tau_{\lim}$.

\begin{lemma}\label{lemma:gr0_n_is_clopen_in_gr_n} $\gr^0_n$ is a clopen subgroupoid of $\gr_n$.
\end{lemma}

\begin{proof} For each $n \in \mathbb{N}$, $\gr^0_n \subseteq \gr_n$, and we have
\begin{equation}\label{eq:I_0_n_union_of_graphs}
    \gr^0_n = \bigcup_{\varphi \in \mathcal{F}_n(X)} \mathsf{gr}(\varphi),
\end{equation}
then $\gr^0_n$ is a finite union of closed sets; hence, it is closed in the subspace of $X \times X$, and then $\gr^0_n$ is a closed subgroupoid of $\gr_n$. Now we prove that $\gr^0_n$ is an open subset of $\gr_n$. Define $B_n := \gr_n \setminus \gr^0_n$. In other words,
\begin{equation*}
    B_n = \left\{(x,y) \in X \times X: x_{\Lambda_n^c} = y_{\Lambda_n^c} \text{ and } \varphi(x) \neq y \text{ for every } \varphi \in \mathcal{F}_n(X) \right\}.
\end{equation*}
Let $\left\{(x^m,y^m)\right\}_{m \in \mathbb{N}}$ be a sequence in $B_n$ that converges to some point $(x,y) \in X \times X$. Since $\gr_n$ is closed in $X\times X$, we have that $(x,y) \in \gr_n$. Since $\mathcal{A}$, $\Lambda_n$, and $\mathcal{F}_n(X)$ are finite sets, we have that there exists $M > 0$ such that, for every $m > M$, we have $x^m_{\Lambda_n} = x_{\Lambda_n}$, $y^m_{\Lambda_n} = y_{\Lambda_n}$, and $\varphi(x^m)_{\Lambda_n} = \varphi(x)_{\Lambda_n}$. Suppose that there exists $\varphi \in \mathcal{F}_n(X)$ satisfying $\varphi(x) = y$. Then, for $m > M$, we have
\begin{equation*}
    \varphi(x^m)_{\Lambda_n} = \varphi(x)_{\Lambda_n} = y_{\Lambda_n} = y^m_{\Lambda_n}.
\end{equation*}
Since $(x^m,y^m) \in \gr_n$ and $\varphi \in \mathcal{F}_n(X)$, we have
\begin{equation*}
    \varphi(x^m)_{\Lambda_n^c} = x^m_{\Lambda_n^c} = y^m_{\Lambda_n^c},
\end{equation*}
and hence $\varphi(x^m) = y^m$, a contradiction since $(x^m,y^m) \notin \gr^0_n$. We conclude that $B_n$ is closed, and therefore $\gr^0_n$ is open.
\end{proof}

\begin{proposition} \label{prop:I_0_etale} $\gr^0$ is an \'etale groupoid.    
\end{proposition}

\begin{proof} By Lemma \ref{lemma:gr0_n_is_clopen_in_gr_n}, for each $n \in \mathbb{N}$, the subgroupoid $\gr^0_n$ is clopen in $\gr_n$. Since each $\gr_n$ is \'etale and $(\gr^0_n)^{(0)} = \gr_n^{(0)} = X$, we have by Lemma \ref{lemma:clopen_subgroupoid_of_etale_groupoid_is_etale_same_unity_set} that $\gr^{0}_n$ is \'etale. By definition, the topologies $\tau_{\lim}$ and the prodiscrete topology coincide as subspace topologies on $\gr^0_n$, and hence each $\gr^0_n$ is a compact subset of $\gr^0$. Then, by Lemma \ref{lemma:inductive_limit_groupoid_etalicity} we conclude that $\gr^0$ is \'etale.
\end{proof}

\begin{remark} \label{rmk:I_0_n_compact} Observe that \eqref{eq:I_0_n_union_of_graphs} implies on the compactness of $\gr^0_n$, since a finite union of closed subsets of $X \times X$. By similar argument in Remark \ref{rmk:maps_q_n_and_compactness_of_I_n}, the compactness also holds in $\tau_{\lim}$.    
\end{remark}

\begin{lemma}\label{lemma:I_0_n+1_is_I_0_n_space} $\gr^{0}_{n+1}$ is a proper $\gr^{0}_n$-space.
\end{lemma}

\begin{proof} Let the continuous map $r:\gr^{0}_{n+1} \to (\gr^{0}_{n})^{(0)} = X$ from Definition \ref{def:G_spaces} be the range map on $\gr^{0}_{n+1}$, and $s$ be the source map on $\gr^{0}_n$. We have that
\begin{align*}
    \gr^{0}_{n} \ast \gr^{0}_{n+1} &= \{(g,h) \in \gr^{0}_{n} \times \gr^{0}_{n+1}: s(g) = r(h)\} \\
    &= \{((x,y),(y,z)) \in (X \times X) \times (X \times X): (x,y) \in \gr^{0}_{n}, (y,z) \in \gr^{0}_{n+1}\}.
\end{align*}
Now, choose the map $\cdot: \gr^{0}_{n} \ast \gr^{0}_{n+1} \to \gr^{0}_{n+1}$ being the groupoid product of $\gr^0$ restricted to $\gr^{0}_{n} \ast \gr^{0}_{n+1}$. We claim that the $\cdot$ is well-defined, i.e., for $(g,h) \in \gr^{0}_{n} \ast \gr^{0}_{n+1}$, we have $g \cdot h \in \gr^{0}_{n+1}$. Indeed, let $(x,y) \in \gr^{0}_{n}$ and $(y,z) \in \gr^{0}_{n+1}$. Then, there exist $\varphi_1 \in \mathcal{F}_n(X)$ and $\varphi_2 \in \mathcal{F}_{n+1}(X)$ such that $y = \varphi_1(x)$ and $z = \varphi_2(y)$. Then, $z = \varphi_2 \circ \varphi_1(x)$, $(x,y) \cdot (y,z) = (x,z)$, and $\varphi_2 \circ \varphi_1 \in \mathcal{F}_{n+1}(X)$, and hence $(X,z) \in \gr^{0}_{n+1}$. Therefore $\cdot$ is well-defined. It is straightforward that $\cdot$ is continuous, and the properties $(a)-(c)$ in Definition \ref{def:G_spaces} are direct consequences of the properties of the source and range maps, and hence $\gr^0_{n+1}$ is a $\gr^0_n$-space. Now we prove that $\gr^0_n$ is proper. In fact, $\gr^{0}_{n} \ast \gr^{0}_{n+1}$ is closed because every sequence $((x^m,y^m),(y^m,z^m))$ in $\gr^{0}_{n} \ast \gr^{0}_{n+1}$ that converges to some point $((x,y)(w,z))$ in $(X \times X) \times (X \times X)$ implies that $w=y$, $(x,y) \in \gr^{0}_{n}$, and $\gr^{0}_{n+1}$, because $\gr^{0}_n$ is closed for every $n \in \mathbb{N}$. Hence $\gr^{0}_{n} \ast \gr^{0}_{n+1}$ is compact. Since $I^0_{n+1}$ is Hausdorff and compact, its compact subsets are precisely its closed subsets, by continuity of $\cdot$, the preimages of every compact set of $\cdot$ are closed in $\gr^{0}_{n} \ast \gr^{0}_{n+1}$, and therefore compact. We conclude that $\cdot$ is a proper map and therefore $\gr^{0}_{n+1}$ is a proper $\gr^{0}_n$-space.
\end{proof}
%\begin{corollary} $C^*(\gr) \simeq C^*_r(\gr)$. 
%\end{corollary}

\begin{theorem} $\gr^0$ is an amenable groupoid. Hence, the full and the reduced C$^*$-algebras coincide.
\end{theorem}

\begin{proof} This is an application of Proposition \ref{prop:amenability_via_G_spaces}. $\gr^0$ is an \'etale groupoid due to Proposition \ref{prop:I_0_etale}. On the other hand, $\{\gr^0_n\}_{n \in \mathbb{N}}$ is an increasing family of closed subgroupoids of $\gr^0$ such that $(\gr^0)^{(0)} \subseteq \gr^0_n$ for every $n \in \mathbb{N}$, where $\gr^0 = \bigcup_{n \in \mathbb{N}} \gr^0_n$. Also, Lemma \ref{lemma:I_0_n+1_is_I_0_n_space} gives that $\gr^0_{n+1}$ is a proper $\gr^0_n$-space. It remains to prove that each $\gr^0_n$ is amenable, but this comes from the fact that $\gr^0_n$ is a closed subgroupoid of $\gr$, which is amenable, and then the amenability of $\gr^0_n$ is a consequence of Proposition \ref{prop:amenability_closed_or_open_subgroupoid}.  
\end{proof}

\section{Conformal measures}\label{rn}

In this section, we start by recalling some concepts present in \cite{feldman:77} concerning countable Borel equivalence relations on Polish spaces in order to introduce the notions of quasi-invariant measures, Radon-Nikodym derivatives, and conformal measures. We also establish a connection between various terminologies and objects that commonly appear in the descriptive set theory and the groupoid theory; in particular, we concentrate on the case addressed in Example \ref{exa:equiv_relation_groupoid}. The main purpose of connecting these two theories is to link the so-called DLR measures to the KMS states, in such a way that, once we establish such a correspondence between these objects, the DLR measures can be regarded as thermodynamic objects for a quantum system. Therefore, we can transport a classical statistical mechanical problem to the non-commutative C$^*$-algebra setting. If the reader wishes to look at this approach only from the Borel sets/descriptive set theory point-of-view, see \cite{Merda,Kimura}.

Given an equivalence relation $R \subseteq X \times X$ on a Polish space $X$, we denote the equivalence class of an element $x$ of $X$  by $R(x) := \{y \in X : (x,y) \in R\}$, and, given a subset $A$ of $X$, we define its $R$-saturation by $R(A) := \bigcup\{R(x) : x \in A\}$. In the case where each of its equivalence classes is countable, $R$ is said to be a countable equivalence relation. 

\begin{remark} On the groupoid from Example \ref{exa:equiv_relation_groupoid}, we have $R(x) \equiv s(r^{-1}(x))$ and $R(A) \equiv s(r^{-1}(A))$.    
\end{remark}

In this whole section, we assume that $R$ is a countable equivalence relation on a Polish space  $X$ which is Borel measurable with respect to the product topology on $X \times X$. As before, $\mathcal{B}_R$ denotes the Borel $\sigma$-algebra of $R$, which coincides with the restriction of the Borel $\sigma$-algebra of $X \times X$ to $R$.
Let us define the functions $r, s: R \rightarrow X$ by letting $r(x,y) = x$ and $s(x,y) = y$.
The maps $r$ and $s$ defined above are called the left projection and the right projection of $R$, respectively, and their notations were derived from the range and source maps from groupoid theory. Let us remark that there is no confusion about using such notation because in our case $X$ is homeomorphic to the unit space of the groupoid $\mathcal{G}(R)$. It is also useful to consider the flip map $\theta: R \rightarrow R$ defined by
$\theta(x,y) = (y,x)$, that corresponds to the inverse map in Example \ref{exa:equiv_relation_groupoid}. 
\begin{remark}
We claim that $\theta$ is a Borel isomorphism, thus it maps
sets in $\mathcal{B}_R$ onto sets in $\mathcal{B}_R$; and the projections $r$ and $s$ map sets in $\mathcal{B}_R$ onto Borel sets of $X$.
The proof of the first statement is straightforward.
The second one follows by using Theorem $4.12.3$ from \cite{sri:98} and the fact that $s = r \circ \theta$.	
\end{remark}

Next, we provide some important examples of countable Borel equivalence relations.

\begin{example}[Gibbs relation]\label{gr}
Let $X$ be a compact metrizable space, and let $T$ be a continuous $\zd$-action on $X$, that is, $T$ is a homomorphism of the additive group $\zd$ into the group of homeomorphisms on $X$. If we assume that $T$ is expansive, then,
the Gibbs (or homoclinic) relation of $(X,T)$ defined by
\[\gr(X,T) := \left\{(x,y) \in X \times X : \lim\limits_{\|i\| \to \infty} \rho(T^ix,T^iy) = 0 \right\},\]
where $\rho$ is a metric on $X$ which induces its topology, is a countable Borel equivalence relation, see \cite{meyerovitch:13}.
Note that the definition of $\gr(X,T)$ does not depend on the choice of the metric $\rho$. In particular, if we let $X$ be a subshift, $T:i \mapsto \sigma^i$ the shift action of $\zd$, and $\rho$ the metric defined in equation (\ref{metric}), it follows from Remark \ref{box} that $\gr(X,T)$ coincides with the Gibbs relation defined in (\ref{gibbsrel_shift}). 
\end{example}

\begin{example}[Orbit equivalence relation]\label{group}
Let $X$ be a Polish space and let
\[\mathsf{Aut(}X\mathsf{)} := \left\{f \in X^X : f\; \textnormal{is invertible, and both}\; f\; \text{and}\;f^{-1} \; \textnormal{are Borel measurable}\right\}\]
be the group of all Borel automorphisms of $X$. If we
consider a countable group $G \subseteq \mathsf{Aut(}X\mathsf{)}$ and define the orbit equivalence relation by
\[R_G := \big\{(x,y) \in X \times X:y=g(x)\;\text{for some}\;g\in G\big\},\] 
then, $R_G$ is a countable Borel equivalence relation on $X$. In fact,
since $R_G$ is a countable union of graphs of Borel measurable functions, it follows that $R_G$ is a Borel subset of $X \times X$; moreover, each $R_G(x) = \{g(x) : g \in G\}$ is countable. 
\end{example}

An important result due to J. Feldman and C. C. Moore \cite{feldman:77} proves that the reciprocal of Example \ref{group} is true, that is, every countable Borel equivalence relation can be written in the form  $R = R_G$ for some countable group of Borel automorphisms.

\begin{theorem}[Feldman and Moore]\label{fmteo}
Let $R$ be a countable Borel equivalence relation on a Polish space $X$. Then there exists a countable group
$G \subseteq \mathsf{Aut(}X\mathsf{)}$ such that $R = R_G$.
\end{theorem}
\begin{proof}
For a modern proof, see Theorem 5.8.13 from \cite{sri:98}.
\end{proof}

Let us present the necessary tools that will allow us to introduce the notion of a Radon-Nikodym derivative of a $\sigma$-finite Borel measure on $X$ with respect to $R$. First, let us note that for every Borel set $A$ of $X$, its $R$-saturation $R(A)$ is also a Borel set of $X$. Indeed,
according to Theorem \ref{fmteo}, there is a countable group $G \subseteq \mathsf{Aut(}X\mathsf{)}$ such that
$R(A) = \bigcup\limits_{\varphi \in G} \varphi^{-1}(A)$.

\begin{definition}\label{def:quasi_invariant_FM}
Let $\mu$ be a $\sigma$-finite Borel measure on $X$. We say that $\mu$ is quasi-invariant under $R$ (or $R$ is non-singular with respect to $\mu$) if the condition
\[\mu(A)=0 \;\,\text{implies}\;\, \mu\left(R(A)\right) = 0\] 
is satisfied for every Borel set $A$ of X.
\end{definition}

The next theorem will provide us with two measures on the measurable space $(R,\mathcal{B}_R)$ which will allow us to define the Radon-Nikodym derivative of
$\mu$ with respect to $R$.

\begin{theorem}\label{thm:nur_nus}
Let $\mu$ be a $\sigma$-finite Borel measure on $X$. Then, the following properties hold.
\begin{itemize}
\item[(a)] For each $C \in \mathcal{B}_R$, the map $x \mapsto \left|r^{-1}(\{x\})\cap C\right|$ defined on $X$ is Borel measurable, and the formula 
\begin{equation}\label{eq:nu_r}
\nu_r(C) = \int_X{\left|r^{-1}(\{x\}) \cap C\right|\, d\mu (x)} 
\end{equation}
defines a $\sigma$-finite measure on $\mathcal{B}_R$. This measure will be referred to as the left counting measure of $\mu$.

\item[(b)] The null sets of $\nu_r$ are exactly the elements of $\left\{C \in \mathcal{B}_R: \mu\left(r(C)\right) = 0\right\}$.

\item[(c)] The right counting measure of $\mu$, defined as the pushforward measure $\nu_s = \theta_{\ast}\nu_r$, is a $\sigma$-finite measure on $\mathcal{B}_R$. Moreover,
\begin{equation}\label{eq:nu_s}
\nu_s(C) = \int_{X}{\left|s^{-1}(\{x\}) \cap C\right|\, d\mu (x)} 
\end{equation}
for every $C \in \mathcal{B}_R$.

\item[(d)] The null sets of $\nu_s$ are exactly the elements of $\left\{C \in \mathcal{B}_R: \mu\left(s(C)\right) = 0\right\}$.
\end{itemize}
\end{theorem}
\begin{proof} 
See \cite{feldman:77}.    
\end{proof}

The following result shows that the notion of quasi-invariance can be connected with the equivalence between the left and right counting measures.

\begin{corollary}\label{xala}
A $\sigma$-finite Borel measure $\mu$ is quasi-invariant under $R$  if and only if $\nu_r \ll \nu_s$ and $\nu_s \ll \nu_r$.
\end{corollary} 
\begin{proof}
It is straightforward to show that, given an arbitrary subset $A$ of $X$, its $R$-saturation coincides with the sets $s\left(r^{-1}(A)\right)$ and $ r\left(s^{-1}(A)\right)$.

Assume that $\mu$ is quasi-invariant under $R$. Let us show that $\nu_s \ll \nu_r$.
Let $C$ be a set of $\mathcal{B}_R$ satisfying $\mu\left(r(C)\right) = 0$. It follows from our hypothesis that $\mu\left(R(r(C))\right) =0$. Since $C \subseteq r^{-1}\left(r(C)\right)$, it follows that
$s(C) \subseteq s\left(r^{-1}(r(C))\right) = R(r(C))$, thus $\mu\left(s(C)\right) = 0$.
Analogously, one can easily prove that $\nu_r \ll \nu_s$.

On the other hand, let us assume that $\nu_r$ and $\nu_s$ are equivalent measures and $A$ is a Borel set satisfying $\mu(A) = 0$. Since $r(r^{-1}(A)) \subseteq A$, then, we have $\mu(r(r^{-1}(A))) = 0$, in other words, $\nu_{r}(r^{-1}(A)) = 0$. Therefore,  it follows that $\nu_{s}(r^{-1}(A)) = 0$, and we conclude that $\mu(R(A)) = \mu(s(r^{-1}(A))) = 0$.
\end{proof}

Some comments should be made about the Corollary \ref{xala}. This result proves the equivalence of two different notions of quasi-invariant measures. The notion of a quasi-invariant measure presented in Definition \ref{def:quasi_invariant_FM} is the one introduced in \cite{feldman:77}. On the other hand, there is another notion adopted in the groupoid theory that corresponds to the equivalence between the measures $\nu_r$ and $\nu_s$. In the context of non-\'etale groupoids, the latter definition requires the existence of a Haar system, which consists of a set of measures $\{\lambda^x\}_{x \in G^{(0)}}$ where each $\lambda^x$ is supported on the $r$-fiber of $x$ and satisfies some properties of continuity with respect to $x$ and left-invariance (see Definition 2.2 of \cite{Renault:80}). It is well-known that Haar systems generalize the concept of Haar measures studied in the context of group theory (see Lemma 4.1.2 of \cite{Lima2019}). For locally compact \'{e}tale groupoids, the existence of a Haar system is ensured and the unique choice possible corresponds to $\lambda^x$ as the counting measure on $G^x$ up to a multiplicative constant (see Lemma I.2.7 of \cite{Renault:80} and also Proposition 4.1.3 of \cite{Lima2019}), which is precisely the term $|r^{-1}(\{x\}) \cap C|$ in equation \eqref{eq:nu_r}.

From now on, whenever $R$ is non-singular with respect to $\mu$, instead of saying that a property of points of $R$ holds $\nu_r$-a.e. (equivalently, $\nu_s$-a.e.), we will simply say that this property holds almost everywhere (or a.e.). %In this case, the Radon-Nikodym derivatives $\frac{d\nu_r}{d\nu_s}$ and $\frac{d\nu_s}{d\nu_r}$ satisfy the identity $\frac{d\nu_r}{d\nu_s}\cdot \frac{d\nu_s}{d\nu_r} = 1$ a.e., in particular, such derivatives are positive a.e. 

\begin{definition}\label{def:radon_nikodym_nu_r_nu_s}
Let $\mu$ be a $\sigma$-finite Borel measure on $X$ quasi-invariant under $R$. Then, the Radon-Nikodym derivative of $\mu$ with respect to $R$ is the measurable function $D_{\mu,R}$ on $R$ defined by 
\begin{equation}\label{defrn}
D_{\mu,R} = \frac{d\nu_r}{d\nu_s}.
\end{equation}
The function $D_{\mu,R}$ is unique up to almost everywhere equality.
\end{definition}

In the following proposition, given a Borel subset $A$ of $X$ we will denote by $\mu_A$ the measure $\mu$ restricted to the $\sigma$-algebra of Borel subsets of $A$.

\begin{proposition}\label{boraut}
Let $A, B \subseteq X$ be Borel sets, let $\varphi: A \rightarrow B$ be an Borel isomorphism with $\mathsf{gr}(\varphi) \subseteq R$, and let $\mu$ be a $\sigma$-finite Borel measure on $X$ quasi-invariant under $R$.
Then, $\varphi_{\ast}\mu_A$ is absolutely continuous with respect
to $\mu_B$ and the equation
\begin{equation}\label{xena}
\frac{d \varphi_{\ast}\mu_A}{d\mu_B}(y) = D_{\mu,R}\left(\varphi^{-1}(y),y\right)
\end{equation}
holds for $\mu$-almost every $y \in B$.
\end{proposition}

\begin{proof}
See \cite{feldman:77}.    
\end{proof}

%\begin{corollary}\label{mot}
%For $\mu$-almost every $z \in X$, we have 
%\begin{equation}
%D_{\mu,R}(x,z) = D_{\mu,R}(x,y) \cdot D_{\mu,R}(y,z)
%\end{equation} 	
%for all $x,y \in R(z)$.
%\end{corollary}

%\begin{remark}
%From Corollary \ref{mot} and from the fact that $D_{\mu,R} > 0$ a.e., one can %easily verify that for $\mu$-almost every $z \in X$, we have
%\begin{equation}\label{rncocycle}
%\log D_{\mu,R}(x,z) = \log D_{\mu,R}(x,y) + \log D_{\mu,R}(y,z)	
%\end{equation}
%for all $x,y \in R(z)$.
%\end{remark}

In order to define the concept of a conformal measure, first we need to present the definition of a cocycle. For further references, see \cite{aaronson:07,DU91, Renault:80,schmidt:97}.

\begin{definition}
An $R$-cocycle (also called an 1-cocycle of $R$) is a measurable function $c : R \rightarrow \mathbb{R}$ such that
\begin{equation}\label{cocycle}
c(x,z) = c(x,y) + c(y,z){}
\end{equation}
holds for all $x,y, z \in X$ satisfying $(x,y), (y,z) \in R$.
\end{definition}

%\begin{remark}
%It is easy to check that $c(x,x) = 0 $ for every $x \in X$. We also have %$c(x,y) = -c(y,x)$ for each pair $(x,y) \in R$.	
%\end{remark}

%Finally, equations (\ref{rncocycle}) and (\ref{cocycle}) motivate the following definition.

\begin{definition}
Let $c : R \rightarrow \mathbb{R}$ be an $R$-cocycle. A Borel probability measure $\mu$ on $X$ is
called $(c, R)$-conformal if $\mu$ is quasi-invariant under $R$ and the formula
\begin{equation}\label{xeena} 
D_{\mu,R} = e^{-c} 
\end{equation}
holds almost everywhere.
\end{definition}

The following proposition characterizes a conformal measure in terms of a group $G$ (provided by Theorem \ref{fmteo}) which generates the relation $R$. 
This result will be useful in later sections.

\begin{proposition}\label{haaa}
Let $G \subseteq \mathsf{Aut(}X\mathsf{)}$ be a countable group which generates $R$. Then, a Borel probability measure $\mu$ on $X$ is 
$(c, R)$-conformal if and only if for each $\varphi \in G$ the measure $\varphi_{\ast}\mu$ is absolutely continuous with respect to $\mu$ and 
the equation
\begin{equation}\label{xuxa}
\frac{d \varphi_{\ast}\mu}{d \mu}(x) = e^{c(x,\varphi^{-1}(x))} 	
\end{equation}
holds for $\mu$-almost every $x \in X$. 
\end{proposition}

\begin{proof}
If we assume $\mu$ is $(c,R)$-conformal, then equation (\ref{xuxa}) immediately follows by combining Proposition \ref{boraut} with equation (\ref{xeena}).

On the other hand, let us show that $\mu$ is quasi-invariant under $R$ and satisfies (\ref{xeena}). Given a  Borel subset $A$ of $X$ such that 
$\mu(A) = 0 $, we have $\varphi_{\ast}\mu(A) = 0$ for each $\varphi \in G$. It follows that $\mu\left(R(A)\right) = \mu\left(\bigcup\limits_{\varphi \in G}\varphi^{-1}(A)\right) = 0$,
thus $\mu$ is quasi-invariant under $R$.
By combining our assumption with Proposition \ref{boraut}, we obtain that for each $\varphi \in G$ we have
\begin{equation}\label{psico}
D_{\mu,R}\left(\varphi^{-1}(y),y\right) = e^{-c(\varphi^{-1}(y),y)}
\end{equation}
for $\mu$-almost every $y \in X$. Let $N$ be a $\mu$-null set such that (\ref{psico}) is satisfied for every $\varphi \in G$ at each point $y$ of $X\backslash N$.
If we define $C_0 = s^{-1}(N)$, then it follows that $C_0$ is a $\nu_s$-null set and for all $(x,y) \in R\backslash C_0$ we have
\[D_{\mu,R}(x,y) = e^{-c(x,y)}.\]
\end{proof}

\section{Gibbs measures and thermodynamic formalism}\label{sec:therm_form} \\
Let us start this section by introducing the notion of Gibbs measures on subshifts, whose description is given in terms of conformal measures and is frequently adopted in the approach to symbolic dynamical systems, for instance, \cite{aaronson:07,cap:76,keller:98,meyerovitch:13,PetSch:97,schmidt:97}.
We restrict ourselves to the study of Gibbs measures for a specific class of potentials, the so-called
functions with $d$-summable variation \cite{meyerovitch:13} or regular local energy functions \cite{keller:98,muir}. By relying on the theory of conformal measures, we introduce two different definitions of Gibbs measures provided by \cite{meyerovitch:13} that, in particular, coincide for SFTs.  
Differently from the description that is widely employed in statistical physics and probability theory, the above-mentioned definitions do not characterize the Gibbs measures in terms of prescribed conditional probabilities with respect to configurations outside of finite regions, so, for this reason,
we dedicate part of this section to connecting the notions provided by \cite{meyerovitch:13} with other definitions frequently considered in the literature (e.g. \cite{cap:76,georgii:11,ny:08,sarig:09}). For the sake of conciseness, the proofs of some technical statements were omitted in this section, but they can be found explicitly in \cite{Kimura}. Later, we recall the notions of C$^*$-dynamical system and KMS state, in particular, we recall the 1-cocycle dynamics via one-parameter group of $*$-automorphisms. A result by J. Renault \cite{renault2009} (Theorem \ref{thm:Renault_KMS}) allows us to explicit the connection between the KMS states from this dynamics and the Gibbs measures, including the case of the topological Gibbs measures.

Let us begin by introducing some notation. From now on, we will always let $X$ denote a subshift of $\mathcal{A}^{\zd}$ and consider the shift action $i \mapsto \sigma^i$ of $\zd$ on $X$.
Given an arbitrary subset $\Delta$  of $\zd$, then we define the set of all $\Delta$-configurations permited on $X$ by $X_\Delta := \{x_\Delta : x \in X\}$.
And, given a finite subset $\Lambda$ of $\zd$ and a pattern $\omega \in \mathcal{A}^{\Lambda}$, we define the \emph{cylinder} with configuration $\omega$ as the subset of $X$ given by $[\omega] := \{x \in X : x_\Lambda = \omega\}$. 

\subsection{Gibbs measures}\label{sec:medgibbs}

In this subsection, we introduce the Gibbs measures for subshifts, considered by T. Meyerovitch \cite{meyerovitch:13}, J. Aaronson and H. Nakada \cite{aaronson:07}, K. Schmidt \cite{schmidt:97}, and K. Petersen and K. Schmidt \cite{PetSch:97}. In Meyerovitch's paper 
were given two different definitions for Gibbs measures, where for the first notion they were referred to as `Gibbs measures' and for the second one as `topologically Gibbs'. Although these definitions coincide for SFTs, they differ from the usual presented in the literature (cf. \cite{velenik2017,georgii:11,ny:08,ruelle:04,sarig:09}). We will show later that they can be formulated, as usual, in terms of the so-called DLR equations. In addition, these definitions are closely related to another one provided by
Capocaccia \cite{cap:76}. In particular, we will show that all these definitions coincide for subshifts of finite type.

The class of potentials that we will be dealing with in this paper is defined as follows. Given an arbitrary real-valued function $f$ defined on $X$ and a positive integer $n$, we define the $n$-th variation of $f$ as the non-negative extended real number given by   
\begin{equation}
\delta_n (f) := \sup\left\{\left|f(x)-f(y)\right| : x, y \in X \;\text{satisfy}\; x_{\Lambda_{n}} = y_{\Lambda_{n}}\right\}.
\end{equation}
Then, we define the set of all functions with $d$-summable variation on $X$ by

\begin{equation}
SV_{d}(X) := \left\{f \in \mathbb{R}^{X} : \sum\limits_{n =1}^{\infty} n^{d-1}\delta_{n}(f) < +\infty\right\}.	
\end{equation}

\begin{remark}
For each $f \in SV_{d}(X)$ we have $\lim\limits_{n \to \infty} \delta_{n}(f) = 0$. It follows that every 
function with $d$-summable variation is uniformly continuous. 
%Hence $SV_{d}(X) \subseteq C(X)$, where $C(X)$ denotes the set of all real-valued continuous functions on $X$.
\end{remark}

%Note that $SV_{d}(X)$ is a real separable Banach space, endowed with the norm 
%\begin{equation}
%\|f\|_{SV_{d}} = \|f\|_{\infty} + \sum\limits_{n =1}^{\infty} n^{d-1}\delta_{n}(f)\quad \;\text{for each}\; f \in SV_{d}(X),
%\end{equation}
%and the usual function space operations.

  As we previously mentioned, a Gibbs measure for a function $f$ with $d$-summable variation on $X$ will be defined as a special kind of conformal measure. Recall that it follows from Example \ref{gr} that the Gibbs relation $\gr$ defined in equation (\ref{gibbsrel_shift}) is a countable Borel equivalence relation on $X$. Then, given a potential $f$ in $SV_{d}(X)$ let us consider the $\gr$-cocycle $c_{f}$ associated with it defined as follows. 

\begin{definition}\label{def:cocycle_potential}
Given a function $f$ in $SV_d(X)$, we define the cocycle $c_{f} : \gr \rightarrow \mathbb{R}$ by
\begin{equation}
c_{f}(x,y) = \sum\limits_{k \in \zd}f(\sigma^k y) - f(\sigma^k x),
\end{equation}
where the sum above is an unordered sum.
\end{definition}

\begin{definition}[Gibbs measure]\label{gibbsm}
A Borel probability measure $\mu$ on $X$ is called a Gibbs measure for a function $f \in SV_d(X)$ 
if it is $(c_{f}, \gr)$-conformal.
\end{definition}

As proved in Proposition \ref{haaa}, the Gibbs measures can be characterized in terms of a countable group $G$ of Borel automorphisms of $X$ that generates $\gr$, whose existence is assured by Theorem \ref{fmteo}. Such a group $G$ can be determined for the Gibbs relation as follows. For each pair of finite configurations $\omega, \eta$ in $\mathcal{A}^{\Lambda}$,
let us define the map $\varphi_{\omega, \eta} : X \rightarrow X$ by letting
\begin{equation}\label{eq:geradores}
\varphi_{\omega, \eta}(x) = 
\begin{cases}
\omega x_{\Lambda^c} & \text{if}\;x\in [\eta]\;\text{and}\;\omega x_{\Lambda^c}\in X, \\
\eta x_{\Lambda^c} & \text{if}\;x\in [\omega]\;\text{and}\;\eta x_{\Lambda^c}\in X, \\
x &	\text{otherwise}.
\end{cases}
\end{equation}	
It is straightforward to check that each $\varphi_{\omega, \eta}$ is a Borel automorphism of $X$ such that $\mathsf{gr}(\varphi) \subseteq \gr$, and the set $G$ of all such functions written in the form (\ref{eq:geradores}) is a countable group Borel automorphisms of $X$ that generates $\gr$.

Due to certain technical difficulties, it can be challenging to deal with proofs where Gibbs measures are involved unless one can guarantee that the generators of $\gr$ are homeomorphisms. In \cite{meyerovitch:13}, the author provided a generalization of Lanford-Ruelle theorem \cite{lanfordruelle:69} on general subshifts by considering a weaker notion of Gibbs measures which is introduced as follows. Let us consider the countable group $\mathcal{F}(X)$ of  Borel automorphisms of $X$ defined by
\begin{eqnarray*}
\mathcal{F}(X) := \left\{\varphi \in \mathsf{Homeo(}X{\mathsf{)}}: \exists n \in \mathbb{N} \;\text{such that}\; \varphi(x)_{{\Lambda}_{n}^{c}}=x_{{\Lambda}_{n}^{c}}\; \text{for all}\; x \in X\right\}.	
\end{eqnarray*}

We define the topological Gibbs relation $\gr^{0}$ as the equivalence relation generated by $\mathcal{F}(X)$, i.e.,
\begin{equation}
\gr^{0} := R_{\mathcal{F}(X)} = \left\{(x,y) \in X \times X : y=\varphi(x)\;\text{for some}\;\varphi \in \mathcal{F}(X) \right\}.
\end{equation}
Observe that $\gr^{0}$ is a countable Borel equivalence relation (see Example \ref{group}) included in $\gr$; moreover, it can be proven that the equality holds if $X$ is a SFT, see \cite{Kimura}. Given a function $f$ in $SV_{d}(X)$ the restriction of $c_{f}$ to $\gr^{0}$ is a $\gr^{0}$-cocycle. Now, we are able to introduce the concept of a topological Gibbs measure.

\begin{definition}[Topological Gibbs measure]\label{topgibbsm}
A Borel probability measure $\mu$ on $X$ is called a topological Gibbs measure for a function $f \in SV_d(X)$ 
if it is $\left(c_{f}\vert_{\gr^{0}}, \gr^0\right)$-conformal.
\end{definition}

\begin{remark}\label{capeeeta}
It follows from Proposition \ref{haaa} that a Borel probability measure $\mu$ on $X$ is a topological Gibbs measure for a function $f$ in $SV_{d}(X)$ if and only if for each $\varphi$ in $\mathcal{F}(X)$ 
the measure $\varphi_{\ast}\mu$ is absolutely continuous with respect to $\mu$ and 
the equation
\begin{equation}
\frac{d \varphi_{\ast}\mu}{d \mu}(x) = e^{c_{f}(x,\varphi^{-1}(x))} 		
\end{equation}	
holds for $\mu$-almost every $x$ in $X$.
\end{remark}

Using Proposition \ref{boraut} and Remark \ref{capeeeta}, one can prove that every Gibbs measure for a function $f$ in $SV_{d}(X)$ is also a topological Gibbs measure for $f$, but the converse is not necessarily true \cite{meyerovitch:13}. 
%%%%%%%%%%%%%%%%%%%%%%%%%%%%%%%%%%%%%%%%%%%%%%%%%%%%%%%%%%%%%%%%%%%%%%%%%%%%

\subsection{DLR equations}\label{sec:DLReq}
The results in this section connecting the Gibbs measure definition from \cite{meyerovitch:13} with the DLR equations were all based on \cite{Kimura}.
The main result of this section is Theorem \ref{DLR}, which proves that item (c) implies item (e) in the Theorem \ref{main_theorem}. On the other hand, with the improved version of Theorem \ref{DLR1} in \cite{Merda}, we guarantee that the reciprocal statement holds. 
Therefore, with the equivalence between items (c) and (e) from Theorem \ref{main_theorem}, we provide an alternative characterization for Gibbs measures on subshifts that is formulated in terms of a set of conditional expectations, the so-called DLR equations. 
%Due to its physical interpretation and relevance in probability theory, this framework is widely adopted in the approach to classical equilibrium statistical mechanics, for instance, see \cite{dob:68,velenik2017,georgii:11,lanfordruelle:69,ruelle:04,ny:08}.  

We start this section by introducing some notation. 
Let us denote the collection of all nonempty finite subsets of $\zd$
by $\mathscr{S}$. If $X$ is a subshift of $\mathcal{A}^{\zd}$, we denote by $\mathcal{B}$ its Borel $\sigma$-algebra; moreover,
for each $j$ in $\zd$ the projection of the subshift $X$ onto the $j$-th coordinate is the map $\pi_j : X \rightarrow \mathcal{A}$ 
defined by letting $\pi_{j}(x) = x_{j}$ for each $x = (x_i)_{i \in \zd}$.
Given an arbitrary subset $\Delta$ of $\zd$, let $\mathcal{B}_{\Delta}$ be the smallest $\sigma$-algebra on $X$ that contains the collection  
\begin{equation}
\left\{\pi_i^{-1}(A): i \in \Delta, A \subseteq \mathcal{A}\right\}.
\end{equation}
Note that $\mathcal{B}_{\Delta}$ is a sub-$\sigma$-algebra of $\mathcal{B}$.
In the following, given a real-valued function $f$ defined on $X$ and a positive integer $n$, for the sake of conciseness, we will write $f_n$ instead of $\sum\limits_{i \in \Lambda_n}f\circ \sigma^i$.

\begin{lemma}\label{xenaaprincesaguerreira}
Let $X \subseteq \mathcal{A}^{\zd}$ be a subshift, let $f$ be a function in $SV_d(X)$, and let $\Lambda \in \mathscr{S}$. Then, the limit	
\begin{equation}\label{limlimlim}
\lim\limits_{n \to \infty}\frac{e^{f_n(\omega x_{\Lambda^c})}\mathbbm{1}_{\{\omega x_{\Lambda^c}\in X\}}}{\sum\limits_{\eta\, \in \mathcal{A}^{\Lambda}}e^{f_n(\eta x_{\Lambda^c})}\mathbbm{1}_{\{\eta x_{\Lambda^c}\in X\}}} 
\end{equation}
exists for each $\omega \in \mathcal{A}^{\Lambda}$ and each $x \in X$, moreover, it is a non-negative real number.
\end{lemma}

\begin{proof}
First, note that for each positive integer $n$, we have 
\[\sum\limits_{\eta\, \in \mathcal{A}^{\Lambda}}e^{f_n(\eta x_{\Lambda^c})}\mathbbm{1}_{\{\eta x_{\Lambda^c}\in X\}} > 0.\]
In the case where $\omega x_{\Lambda^c}$ does not belong to $X$, the limit given by equation (\ref{limlimlim}) is equal to $0$. 
Otherwise, if $\omega x_{\Lambda^c}$ belongs to $X$, then
\begin{eqnarray*}
\frac{e^{f_n(\omega x_{\Lambda^c})}\mathbbm{1}_{\{\omega x_{\Lambda^c}\in X\}}}{\sum\limits_{\eta\, \in \mathcal{A}^{\Lambda}}e^{f_n(\eta x_{\Lambda^c})}\mathbbm{1}_{\{\eta x_{\Lambda^c} \in X\}}} 
&=& \frac{e^{f_n(\omega x_{\Lambda^c})}}{\sum\limits_{\eta\, \in \mathcal{A}^{\Lambda}}e^{f_n(\eta x_{\Lambda^c})}\mathbbm{1}_{\{\eta x_{\Lambda^c} \in X\}}} \\ 	
&=& \frac{1}{\sum\limits_{\eta\, \in \mathcal{A}^{\Lambda}}\exp{\left(\sum\limits_{i \in \Lambda_n}f\circ \sigma^{i}(\eta x_{\Lambda^c})-f\circ \sigma^{i}(\omega x_{\Lambda^c})\right)}\mathbbm{1}_{\{\eta x_{\Lambda^c} \in X\}}} \\	
\end{eqnarray*}
holds for every positive integer $n$. It is straightforward  to prove that 
\[\lim\limits_{n \to \infty}\sum\limits_{\eta\, \in \mathcal{A}^{\Lambda}}\exp{\left(\sum\limits_{i \in \Lambda_n}f\circ \sigma^{i}(\eta x_{\Lambda^c})-f\circ \sigma^{i}(\omega x_{\Lambda^c})\right)}\mathbbm{1}_{\{\eta x_{\Lambda^c} \in X\}} = \sum\limits_{\eta\, \in \mathcal{A}^{\Lambda}} e^{c_f(\omega x_{\Lambda^c},\eta x_{\Lambda^c})}\mathbbm{1}_{\{\eta x_{\Lambda^c} \in X\}} > 0.\]
Therefore, it follows that
\begin{equation}\label{oiadlr}
\lim\limits_{n \to \infty}\frac{e^{f_n(\omega x_{\Lambda^c})}\mathbbm{1}_{\{\omega x_{\Lambda^c}\in X\}}}{\sum\limits_{\eta\, \in \mathcal{A}^{\Lambda}}e^{f_n(\eta x_{\Lambda^c})}\mathbbm{1}_{\{\eta x_{\Lambda^c} \in X\}}} 
= \frac{1}{\sum\limits_{\eta\, \in \mathcal{A}^{\Lambda}} e^{c_f(\omega x_{\Lambda^c},\eta x_{\Lambda^c})}\mathbbm{1}_{\{\eta x_{\Lambda^c} \in X\}}} > 0. 
\end{equation}
\end{proof}

\begin{definition}\label{especificacao15}
Let $X$ be a subshift of $\mathcal{A}^{\zd}$, and let $f$ be a function in $SV_d(X)$. Let us define 
the family $\gamma = (\gamma_{\Lambda})_{\Lambda \in \mathscr{S}}$, where each $\gamma_{\Lambda} : \mathcal{B}\times X\rightarrow {[}0,1]$ is defined by letting
\begin{equation}\label{gamma}
\gamma_{\Lambda}(A|\,x) = \lim\limits_{n \to \infty}\frac{\sum\limits_{\omega \in \mathcal{A}^{\Lambda}}e^{f_n(\omega x_{\Lambda^c})}\mathbbm{1}_{\{\omega x_{\Lambda^c}\in A\}}}{\sum\limits_{\eta\, \in \mathcal{A}^{\Lambda}}e^{f_n(\eta x_{\Lambda^c})}\mathbbm{1}_{\{\eta x_{\Lambda^c}\in X\}}}  
\end{equation}
for each $A \in \mathcal{B}$ and each point $x \in X$.
\end{definition}
\begin{remark}
 Using Lemma \ref{xenaaprincesaguerreira}, the reader can verify that equation (\ref{gamma}) 
is well-defined, and the relation
\begin{equation}\label{lalalalallalalala}
\gamma_{\Lambda}(A|x) =  \sum\limits_{\omega \in \mathcal{A}^{\Lambda}} \gamma_{\Lambda}([\omega]|x) \mathbbm{1}_{\{\omega x_{\Lambda^c} \in A\}}	
\end{equation}
holds for each $A \in \mathcal{B}$ and each $x \in X$.
\end{remark}

In the following lemma, we show that $\gamma$ given in Definition \ref{especificacao15} is a family of probability kernels satisfying special properties that characterize it as a \emph{specification} \cite{velenik2017,georgii:11}.

\begin{lemma}\label{gammaprop}
The family $\gamma = (\gamma_{\Lambda})_{\Lambda \in \mathscr{S}}$ is a
specification with parameter set $\zd$ and state space $(\mathcal{A}, \mathcal{P}(\mathcal{A}))$. 
In other words, given a nonempty finite subset $\Lambda$ of $\zd$, the following properties hold:
\begin{enumerate}[label=(\alph*),ref=\alph*]
\item $\gamma_{\Lambda}(\cdot\,|x)$ is a Borel probability measure on $X$ for each $x \in X$.

\item $\gamma_{\Lambda}(A|\,\cdot)$ is a $\mathcal{B}_{\Lambda^c}$-measurable function for each $A \in \mathcal{B}$.

\item $\gamma_{\Lambda}(B|\,\cdot) = \mathbbm{1}_{B}$ for every $B \in \mathcal{B}_{\Lambda^c}$. 

\item The consistency condition 
\begin{equation}\label{eetalele}
\gamma_{\Delta}\gamma_{\Lambda} = \gamma_{\Delta}
\end{equation}
holds whenever $\Lambda$ and $\Delta$ are elements of $\mathscr{S}$ satisfying $\Lambda \subseteq \Delta$. It means that the equation
\begin{equation}
\int_{X}\gamma_{\Delta}(dy|x)\gamma_{\Lambda}(A|y)  = \gamma_{\Delta}(A|x)
\end{equation}
holds for each set $A \in \mathcal{B}$ and each point $x \in X$.
\end{enumerate} 
\end{lemma}

\begin{proof}
The proofs of conditions (a), (b), and (c) are straightforward. So, let us prove that given a Borel measurable function $f: X \rightarrow [0,+\infty]$, the equation
\begin{equation}\label{intgamma}
\int_{X}\gamma_{\Delta}(dy|x)f(y)  = \sum\limits_{\omega \in \mathcal{A}^{\Delta}}\gamma_{\Delta}([\omega]|x)\cdot \left(f(\omega x_{\Delta^c}) \mathbbm{1}_{\{\omega x_{\Delta^c} \in X\}}\right)
\end{equation}
holds for every point $x$ in $X$. Indeed, the equation above is easily verified if $f$ is a characteristic function. Using 
the linearity of the integral, it is straightforward to show that the equation above also holds for simple functions. Now, in the case where
$f$ is a non-negative extended real-valued function, the result follows by using the fact that 
there is an increasing sequence of non-negative (measurable) simple functions converging pointwise to $f$, and then applying
the monotone convergence theorem. 

Therefore,
\begin{align*}
\int_{X}&\gamma_{\Delta}(dy|x)\gamma_{\Lambda}(A|y) = \\
&= \sum\limits_{\omega \in \mathcal{A}^{\Delta}}\gamma_{\Delta}([\omega]|x) \cdot \left(\gamma_{\Lambda}(A|\omega x_{\Delta^c})\mathbbm{1}_{\{\omega x_{\Delta^c} \in X\}}\right) \\
&= \lim\limits_{n \to \infty}\sum\limits_{\omega \in \mathcal{A}^{\Delta}}\frac{\left(e^{f_n(\omega x_{\Delta^c})}\mathbbm{1}_{\{\omega x_{\Delta^c}\in X\}}\right)\left(\sum\limits_{\omega' \in \mathcal{A}^{\Lambda}}e^{f_n(\omega'\omega_{\Delta \backslash \Lambda}x_{\Delta^c})}\mathbbm{1}_{\{\omega'\omega_{\Delta \backslash \Lambda}x_{\Delta^c}\in A\}}\right)}{\left(\sum\limits_{\eta\, \in \mathcal{A}^{\Delta}}e^{f_n(\eta x_{\Delta^c})}\mathbbm{1}_{\{\eta x_{\Delta^c}\in X\}}\right)\left(\sum\limits_{\eta'\, \in \mathcal{A}^{\Lambda}}e^{f_n(\eta'\omega_{\Delta \backslash \Lambda}x_{\Delta^c})}\mathbbm{1}_{\{\eta'\omega_{\Delta \backslash \Lambda}x_{\Delta^c}\in X\}}\right)}\,\mathbbm{1}_{\{\omega x_{\Delta^c}\in X\}} \\
&= \lim\limits_{n \to \infty}\sum\limits_{\omega'' \in \mathcal{A}^{\Delta \backslash \Lambda}}\sum\limits_{\omega \in \mathcal{A}^{\Lambda}}\frac{\left(e^{f_n(\omega \omega''x_{\Delta^c})}\mathbbm{1}_{\{\omega \omega''x_{\Delta^c}\in X\}}\right) \left(\sum\limits_{\omega' \in \mathcal{A}^{\Lambda}}e^{f_n(\omega'\omega''x_{\Delta^c})}\mathbbm{1}_{\{\omega'\omega''x_{\Delta^c}\in A\}}\right)}{\left(\sum\limits_{\eta\, \in \mathcal{A}^{\Delta}}e^{f_n(\eta x_{\Delta^c})}\mathbbm{1}_{\{\eta x_{\Delta^c}\in X\}}\right)\left(\sum\limits_{\eta'\, \in \mathcal{A}^{\Lambda}}e^{f_n(\eta'\omega''x_{\Delta^c})}\mathbbm{1}_{\{\eta'\omega''x_{\Delta^c}\in X\}}\right)}\,\mathbbm{1}_{\{\omega \omega''x_{\Delta^c}\in X\}} \\
&= \lim\limits_{n \to \infty}\frac{\sum\limits_{\omega'' \in \mathcal{A}^{\Delta \backslash \Lambda}}\;\sum\limits_{\omega' \in \mathcal{A}^{\Lambda}}e^{f_n(\omega'\omega''x_{\Delta^c})}\mathbbm{1}_{\{\omega'\omega''x_{\Delta^c}\in A\}}}{\sum\limits_{\eta\, \in \mathcal{A}^{\Delta}}e^{f_n(\eta x_{\Delta^c})}\mathbbm{1}_{\{\eta x_{\Delta^c}\in X\}}} \\
&= \gamma_{\Delta}(A|x)
\end{align*}
holds for each $A \in \mathcal{B}$ and each $x \in X$.
\end{proof}

The main results of this section discuss the relationship between the notions of Gibbs measures established in section \ref{sec:medgibbs} with probability measures specified (or admitted) by a specification. In the framework of classical statistical physics, in particular, we are mainly concerned with the description of equilibrium states by means of probability measures specified by a \emph{Gibssian specification} in the sense of \cite{georgii:11}. In that context, such a measure is referred to as a Gibbs measure or DLR state, and its corresponding system of equations of prescribed conditional expectations is referred to as DLR equations (e.g. \cite{velenik2017,georgii:11,muir,sarig:09,ny:08}), named in honor of R. L. Dobrushin, O. E. Lanford, and D. Ruelle. Although such terminology should be reserved only for the Gibssian case, its extension to the general case is often observed in common practice.

\begin{theorem}\label{DLR}
Let $X \subseteq \mathcal{A}^{\zd}$ be a subshift, let $f$ be a function in $SV_d(X)$, and let $\gamma = (\gamma_{\Lambda})_{\Lambda \in \mathscr{S}}$ be the corresponding specification given in Definition \ref{especificacao15}. 
If $\mu$ is a Gibbs measure for $f$, then $\mu$ satisfies
\begin{equation}\label{DLReq}
\mu(A|\mathcal{B}_{\Lambda^c}) = \gamma_{\Lambda}(A|\,\cdot\,) \quad \text{$\mu$-a.e.}
\end{equation}
for each $\Lambda \in \mathscr{S}$ and each $A \in \mathcal{B}$.
\end{theorem}

\begin{proof}
Step 1. Let us show that for all $\Lambda \in \mathscr{S}$ and $\omega \in \mathcal{A}^{\Lambda}$, the equation
\begin{equation}
\mu\left([\omega]|{\mathcal{B}_{\Lambda^c}}\right)(x) = \gamma_{\Lambda}([\omega]|x)  
\end{equation}
holds for $\mu$-almost every point $x$ in $X$. Given $\eta$ in $\mathcal{A}^{\Lambda}$,
let us consider the map $\varphi = \varphi_{\omega,\eta}$  as in equation (\ref{eq:geradores}).
Recall that $\varphi$ is a Borel automorphism of $X$ such that $\mathsf{gr}(\varphi) \subseteq \gr$. 
For each $F \in \mathcal{B}_{\Lambda^c}$, we have
\begin{eqnarray*}
\int_{F}\mathbbm{1}_{[\eta]}(x)\mathbbm{1}_{\{\omega x_{\Lambda^c} \in X\}}\,d\mu(x) 
&=& \mu\big([\eta]\cap \{y \in X : \omega y_{\Lambda^c} \in X\} \cap F\big)\\
&=& \varphi_{\ast}\mu\big(\varphi\big([\eta]\cap \{y \in X : \omega y_{\Lambda^c} \in X\}\big) \cap \varphi(F)\big)\\
&=& \varphi_{\ast}\mu\big([\omega]\cap \{y \in X : \eta y_{\Lambda^c} \in X\} \cap \varphi(F)\big).
\end{eqnarray*}
Observe that the $\sigma$-algebra $\mathcal{B}_{\Lambda^c}$ is generated by the collection $\mathscr{C}$ given by
$\mathscr{C} = \{\pi_{i}^{-1}(A) : i \in \Lambda^c, A\subseteq \mathcal{A}\}$.
Moreover, since the collection $\{F \in \mathcal{B}_{\Lambda^c} : \varphi(F) = F\}$ is a $\sigma$-algebra of subsets of $X$ which contains
$\mathscr{C}$, then it coincides with $\mathcal{B}_{\Lambda^c}$. It follows that 

\begin{equation}
\int_{F}\mathbbm{1}_{[\eta]}(x)\mathbbm{1}_{\{\omega x_{\Lambda^c} \in X\}}\,d\mu(x) 
= \varphi_{\ast}\mu\big([\omega]\cap \{y \in X : \eta y_{\Lambda^c} \in X\} \cap F\big).
\end{equation}
In view of Proposition \ref{boraut} and the fact that $\mu$ is $(c_{f},\gr)$-conformal, we obtain
\begin{eqnarray*}
\int_{F}\mathbbm{1}_{[\eta]}(x)\mathbbm{1}_{\{\omega x_{\Lambda^c} \in X\}}\,d\mu(x) 
&=& \int_{[\omega]\cap \{y \in X : \eta y_{\Lambda^c} \in X\} \cap F}e^{c_f(x,\varphi^{-1}(x))}d\mu(x) \\
&=& \int_{F}e^{c_f(x,\varphi(x))}\mathbbm{1}_{\{\eta x_{\Lambda^c} \in X\}}\mathbbm{1}_{[\omega]}(x)\,d\mu(x) \\
&=& \int_{F}\underbrace{\left(e^{c_f(\omega x_{\Lambda^c},\varphi(\omega x_{\Lambda^c}))}\mathbbm{1}_{\{\omega x_{\Lambda^c} \in X,\; \eta x_{\Lambda^c} \in X\}}\right)}_{\mathcal{B}_{\Lambda^c}\;\text{-measurable function on}\;X}\cdot\,\mathbbm{1}_{[\omega]}(x)\,d\mu(x) \\
&=& \int_{F}\left(e^{c_f(\omega x_{\Lambda^c},\eta x_{\Lambda^c})}\mathbbm{1}_{\{\omega x_{\Lambda^c} \in X,\; \eta x_{\Lambda^c} \in X\}}\right)\cdot\mu([\omega]|\mathcal{B}_{\Lambda^c})(x)\,d\mu(x) 
\end{eqnarray*}
for every $F \in \mathcal{B}_{\Lambda^c}$. It follows that the equation
\begin{equation*}\label{dlreita}
\mu([\eta]|\mathcal{B}_{\Lambda^c})(x)\mathbbm{1}_{\{\omega x_{\Lambda^c} \in X\}} = \left(e^{c_f(\omega x_{\Lambda^c},\eta x_{\Lambda^c})}\mathbbm{1}_{\{\omega x_{\Lambda^c} \in X,\; \eta x_{\Lambda^c} \in X\}}\right)\cdot\mu([\omega]|\mathcal{B}_{\Lambda^c})(x)
\end{equation*}
holds for $\mu$-almost every $x$ in $X$.
Thus, if we sum the equation above over all elements $\eta$ of $\mathcal{A}^{\Lambda}$, we conclude that
\begin{equation}\label{dlr1}
\mathbbm{1}_{\{\omega x_{\Lambda^c} \in X\}} = \left(\,\sum\limits_{\eta\, \in \mathcal{A}^{\Lambda}}e^{c_f(\omega x_{\Lambda^c},\eta x_{\Lambda^c})}\mathbbm{1}_{\{\omega x_{\Lambda^c} \in X,\; \eta x_{\Lambda^c} \in X\}}\right)\cdot\mu([\omega]|\mathcal{B}_{\Lambda^c})(x)
\end{equation}
holds for $\mu$-almost every point $x$ in $X$. 

\begin{claim*}\label{lemadlr}
The equality
\begin{equation}\label{eqlemmadlr}
\mu([\omega]|\mathcal{B}_{\Lambda^c})(x) =\mu([\omega]|\mathcal{B}_{\Lambda^c})(x) \mathbbm{1}_{\{\omega x_{\Lambda^c} \in X\}}
\end{equation}
holds for $\mu$-almost every $x$ in $X$.
\end{claim*}
\begin{proof}[Proof of the Claim]\renewcommand{\qedsymbol}{\QEDB}
Given a set $F$ in $\mathcal{B}_{\Lambda^c}$, we have
\begin{align*} 
\int_{F}\mathbbm{1}_{[\omega]}(x)\,d\mu = \int_{F}\mathbbm{1}_{[\omega]}(x)\mathbbm{1}_{\{\omega x_{\Lambda^c} \in X\}}\,d\mu(x) =  \int_{F}\underbrace{\mu([\omega]|\mathcal{B}_{\Lambda^c})(x)\mathbbm{1}_{\{\omega x_{\Lambda^c} \in X\}}}_{\mathcal{B}_{\Lambda^c}\text{-measurable function on}\;X}\,d\mu(x).
\end{align*}
Thus, equality holds.
\end{proof}

Now, let $N \subseteq X$ be a set of measure zero such that equations (\ref{dlr1}) and (\ref{eqlemmadlr}) hold at each point of $X \backslash N$. For every $x$ in $X \backslash N$, if $\omega x_{\Lambda^c}$ belongs to  $X$, then
\[\underbrace{\left(\sum\limits_{\eta\, \in \mathcal{A}^{\Lambda}}e^{c_f(\omega x_{\Lambda^c},\eta x_{\Lambda^c})}\mathbbm{1}_{\{\eta x_{\Lambda^c} \in X\}}\right)}_{>\;0}\cdot\,\mu([\omega]|\mathcal{B}_{\Lambda^c})(x) = 1,\]
and using equation (\ref{oiadlr}), we obtain
\begin{equation}\label{baratabola15}
\mu([\omega]|\mathcal{B}_{\Lambda^c})(x) = \frac{1}{\sum\limits_{\eta\, \in \mathcal{A}^{\Lambda}}e^{c_f(\omega x_{\Lambda^c},\eta x_{\Lambda^c})}\mathbbm{1}_{\{\eta x_{\Lambda^c} \in X\}}} = \lim\limits_{n \to \infty}\frac{e^{f_n(\omega x_{\Lambda^c})}\mathbbm{1}_{\{\omega x_{\Lambda^c} \in X\}}}{\sum\limits_{\eta\, \in \mathcal{A}^{\Lambda}}e^{f_n(\eta x_{\Lambda^c})}\mathbbm{1}_{\{\eta x_{\Lambda^c} \in X\}}}. 	
\end{equation}
Otherwise, if $\omega x_{\Lambda^c}$ does not belong to $X$, we have $\mu([\omega]|\mathcal{B}_{\Lambda^c})(x) = 0$.
Therefore, we conclude that 
\begin{equation}
\mu([\omega]|\mathcal{B}_{\Lambda^c})(x) 
= \lim\limits_{n \to \infty}\frac{e^{f_n(\omega x_{\Lambda^c})}\mathbbm{1}_{\{\omega x_{\Lambda^c} \in X\}}}{\sum\limits_{\eta \in \mathcal{A}^{\Lambda}}e^{f_n(\eta x_{\Lambda^c})}\mathbbm{1}_{\{\eta x_{\Lambda^c} \in X\}}} 		
= \gamma_{\Lambda}([\omega]|x)
\end{equation}
holds at each point $x$ in $X \backslash N$.

Step 2. For all $\Lambda \in \mathscr{S}$ and $A \in \mathcal{B}$,  
it is straightforward to prove that we have
\[\mu(A|\mathcal{B}_{\Lambda^c})(x) = \sum\limits_{\omega \in \mathcal{A}^{\Lambda}}\mu([\omega]|\mathcal{B}_{\Lambda^c})(x) \mathbbm{1}_{\{\omega x_{\Lambda^c}\in A\}}\]
for $\mu$-almost every $x$ in $X$.
Then, in view of the previous step and equation (\ref{lalalalallalalala}), we conclude that
\begin{equation}
\mu(A|\mathcal{B}_{\Lambda^c}) (x) = \sum\limits_{\omega \in \mathcal{A}^{\Lambda}}\gamma_{\Lambda}([\omega]|x)\mathbbm{1}_{\{\omega x_{\Lambda^c}\in A\}} = \gamma_{\Lambda}(A|x)
\end{equation}
holds for $\mu$-almost every point $x$ in $X$. 
\end{proof} 

On the other hand, we have the following result.

\begin{theorem}\label{DLR1}
Let $X \subseteq \mathcal{A}^{\zd}$ be a subshift, let $f$ be a function in $SV_d(X)$, and let $\gamma = (\gamma_{\Lambda})_{\Lambda \in \mathscr{S}}$ be the corresponding specification given 
in Definition \ref{especificacao15}. If $\mu$ is a Borel probability measure on $X$ that 
satisfies 
\begin{equation}\label{spec}
\mu(A|\mathcal{B}_{\Lambda^c}) = \gamma_{\Lambda}(A|\,\cdot\,) \qquad \mu\text{-a.e.}
\end{equation}
for each $\Lambda \in \mathscr{S}$ and each $A \in \mathcal{B}$, then $\mu$ is a topological Gibbs measure for $f$.
\end{theorem}

Before we prove our result, let us prove the following lemma.

\begin{lemma}\label{satanasjr}
Given an arbitrary element $\varphi$ of $\mathcal{F}(X)$, there is a positive integer $n$ such that
\begin{itemize}
\item[(a)] the equality $\varphi (x)_{\Lambda_{n}^c}=x_{\Lambda_{n}^c}$ holds for every $x$ in $X$, and 
\item[(b)] for each pair $x,y$ of points in $X$ we have
$x_{\Lambda_{n}}=y_{\Lambda_{n}}$ if and only if $\varphi (x)_{\Lambda_{n}}=\varphi (y)_{\Lambda_{n}}$. 
\end{itemize}	
\end{lemma}

\begin{proof}[Proof of Lemma \ref{satanasjr}]
Let $m$ be a positive integer such that $\varphi (x)_{\Lambda_{m}^c}= x_{\Lambda_{m}^c}$ holds for every $x$ in $X$. 
Since $\varphi$ and $\varphi^{-1}$ are continuous functions on $X$, by compactness, it follows that both functions are uniformly continuous. Then, 
there is an integer $n \geq m$ such that for all points $x$ and $y$ in $X$ satisfying
$x_{\Lambda_{n}}=y_{\Lambda_{n}}$ we have $\varphi (x)_{\Lambda_{m}}=\varphi (y)_{\Lambda_{m}}$
and $\varphi^{-1} (x)_{\Lambda_{m}}=\varphi^{-1} (y)_{\Lambda_{m}}$.

Note that part (a) follows from the fact that $\Lambda_{n}^{c} \subseteq \Lambda_{m}^{c}$.
On the other hand, for any two elements $x$ and $y$ of $X$ such that
$x_{\Lambda_{n}}=y_{\Lambda_{n}}$ we have  
$\varphi (x)_{\Lambda_{n} \backslash \Lambda_{m}}= x_{\Lambda_{n} \backslash \Lambda_{m}} = y_{\Lambda_{n} \backslash \Lambda_{m}}=\varphi (y)_{\Lambda_{n} \backslash \Lambda_{m}}$.
Thus, the equality $\varphi (x)_{\Lambda_{n}}=\varphi (y)_{\Lambda_{n}}$ holds whenever $x$ and $y$ are of elements of $X$ that satisfy $x_{\Lambda_{n}}=y_{\Lambda_{n}}$.
One can easily prove an analogous result for $\varphi^{-1}$.
Therefore, part (b) follows.
\end{proof}

\begin{proof}[Proof of Theorem \ref{DLR1}]
Step 1.  Let us write 
\[\mathcal{F}(X) = \bigcup\limits_{n \in \mathbb{N}}\mathcal{F}_n(X),\]
where $\mathcal{F}_n(X)$ is the set of all homeomorphisms $\varphi$ from $X$ onto itself satisfying items (a) and (b) from Lemma \ref{satanasjr}.
Given a positive integer $n$, let $\varphi$ be an arbitrary element of $\mathcal{F}_n(X)$. Let us show that for every $\omega \in \mathcal{A}^{\Lambda_n}$ we have
\begin{equation}\label{recdlr1}
\varphi_{\ast}\mu([\omega]) = \int_{[\omega]}e^{c_{f}(x,\varphi^{-1}(x))}\,d\mu(x).
\end{equation}
In the following, let us denote $\Lambda_n$ simply by $\Lambda$ just for convenience. Observe that $\varphi^{-1}([\omega]) = [\zeta]$ for some $\zeta \in \mathcal{A}^{\Lambda}$. 
Furthermore, it is straightforward  to check that for every $x$ in $X$, $\omega x_{\Lambda^{c}}$ belongs to $X$ if and only if $\zeta x_{\Lambda^{c}}$ belongs to $X$.
Then 
\begin{eqnarray*}
\gamma_{\Lambda}([\zeta]|x)  	
&=& \lim\limits_{n \to \infty}\frac{e^{f_n(\zeta x_{\Lambda^c})}\mathbbm{1}_{\{\omega x_{\Lambda^c} \in X\}}}{\sum\limits_{\eta\, \in \mathcal{A}^{\Lambda}}e^{f_n(\eta x_{\Lambda^c})}\mathbbm{1}_{\{\eta x_{\Lambda^c} \in X\}}} \\ 	
&=& \lim\limits_{n \to \infty}\left(e^{f_n(\zeta x_{\Lambda^c})-f_n(\omega x_{\Lambda^c})}\mathbbm{1}_{\{\omega x_{\Lambda^c} \in X\}}\right)\cdot\frac{e^{f_n(\omega x_{\Lambda^c})}\mathbbm{1}_{\{\omega x_{\Lambda^c} \in X\}}}{\sum\limits_{\eta\, \in \mathcal{A}^{\Lambda}}e^{f_n(\eta x_{\Lambda^c})}\mathbbm{1}_{\{\eta x_{\Lambda^c} \in X\}}} \\ 	
&=& \left(e^{c_f(\omega x_{\Lambda^c},\zeta x_{\Lambda^c})}\mathbbm{1}_{\{\omega x_{\Lambda^c} \in X\}}\right)\cdot \gamma_{\Lambda}([\omega]|x) \\
&=& \left(e^{c_f(\omega x_{\Lambda^c},\varphi^{-1}(\omega x_{\Lambda^c}))}\mathbbm{1}_{\{\omega x_{\Lambda^c} \in X\}}\right)\cdot \gamma_{\Lambda}([\omega]|x)
\end{eqnarray*}
holds at each point $x$ in $X$. Therefore, we have
\begin{eqnarray*}
\varphi_{\ast}\mu([\omega])
&=& \int_{X}\mathbbm{1}_{[\zeta]}(x)\,d\mu(x) \\
&=& \int_{X}\mu([\zeta]|\mathcal{B}_{\Lambda^c})(x)\,d\mu(x) \\
&=& \int_{X}\underbrace{\left(e^{c_f(\omega x_{\Lambda^c},\varphi^{-1}(\omega x_{\Lambda^c}))}\mathbbm{1}_{\{\omega x_{\Lambda^c} \in X\}}\right)}_{\mathcal{B}_{\Lambda^c}\text{-measurable function on}\; X}\,\cdot\, \mu([\omega]|\mathcal{B}_{\Lambda^c})(x)\,d\mu(x) \\
&=& \int_{X}\left(e^{c_f(\omega x_{\Lambda^c},\varphi^{-1}(\omega x_{\Lambda^c}))}\mathbbm{1}_{\{\omega x_{\Lambda^c} \in X\}}\right)\cdot\mathbbm{1}_{[\omega]}(x)\,d\mu(x) \\
&=& \int_{X}e^{c_f(x,\varphi^{-1}(x))}\mathbbm{1}_{[\omega]}(x)\,d\mu(x), 
\end{eqnarray*}
and equation (\ref{recdlr1}) follows.

Step 2. Let $\varphi$ be an arbitrary element of $\mathcal{F}(X)$. Let us consider a collection $\mathscr{C}$ of subsets of $X$ defined by
\[\mathscr{C} = \left\{[\zeta] : \zeta \in \mathcal{A}^{\Lambda}, \Lambda \in \mathscr{S}\right\}\cup\{\emptyset\}.\]
Observe that $\mathscr{C}$ is a $\pi$-system that generates the Borel $\sigma$-algebra of $X$.
If we show that 
\begin{equation}\label{recdlr2}
\varphi_{\ast}\mu(C) = \int_{C}e^{c_{f}(x,\varphi^{-1}(x))}\,d\mu(x)
\end{equation}
holds for every $C \in \mathscr{C}$, then the proof will be complete. 
For each $\Lambda \in \mathscr{S}$ and each $\zeta \in \mathcal{A}^{\Lambda}$,
there is a positive integer $n$ such that $\Lambda \subseteq \Lambda_n$ and $\varphi \in \mathcal{F}_n(X)$. Using the identity 
$[\zeta] = \bigcup\limits_{\substack{\omega \in \mathcal{A}^{\Lambda_n} \\ \omega_{\Lambda} = \zeta}}[\omega]$ and the previous step, we obtain
\begin{eqnarray*}
\varphi_{\ast}\mu([\zeta]) = \sum\limits_{\substack{\omega \in \mathcal{A}^{\Lambda_{n}} \\ \omega_{\Lambda} = \zeta}}\varphi_{\ast}\mu([\omega]) 
=\sum\limits_{\substack{\omega \in \mathcal{A}^{\Lambda_{n}} \\ \omega_{\Lambda} = \zeta}}\int_{[\omega]}e^{c_{f}(x,\varphi^{-1}(x))}\,d\mu(x) 
=\int_{[\zeta]}e^{c_{f}(x,\varphi^{-1}(x))}\,d\mu(x).
\end{eqnarray*}
\end{proof}

The next result follows immediately from Theorems \ref{DLR} and \ref{DLR1} and provides us with a characterization for Gibbs measures on subshifts on finite type in terms of the DLR equations.	

\begin{corollary}\label{DLR2}
Let $X \subseteq \mathcal{A}^{\zd}$ be a subshift of finite type, let $f$ be a function in $SV_d(X)$, and let  $\gamma = (\gamma_{\Lambda})_{\Lambda \in \mathscr{S}}$ be the corresponding specification given in Definition \ref{especificacao15}. Then, $\mu$ is a Gibbs measure for $f$ if
and only if $\mu$ is a Borel probability measure on $X$ that satisfies
\begin{equation}
\mu(A|\mathcal{B}_{\Lambda^c}) = \gamma_{\Lambda}(A|\,\cdot\,) \qquad \mu\text{-a.e.}
\end{equation}
for each $\Lambda \in \mathscr{S}$ and each $A \in \mathcal{B}$.
\end{corollary}

\subsection{Gibbs measures in statistical mechanics}
In the general framework of classical equilibrium statistical mechanics, the main objects of interest are the DLR states, also referred to as infinite-volume Gibbs states in the physics literature, that are essentially probability measures admitted by a Gibssian specification \cite{georgii:11}. The main result of this subsection is Theorem \ref{thm:Gibbsian}, which states that, under certain conditions, if $f$ derives from an interaction potential $\Phi$, whose meaning will be precised later, the specification given in Definition \ref{especificacao15} assumes the familiar form of a Gibbs distribution. In other words, we obtain some sufficient conditions so that for each finite volume $\Lambda \subseteq \zd$, we have a Hamiltonian $H_{\Lambda}^{\Phi}$, and the kernels $\gamma_{\Lambda}$ from the specification from Definition \ref{especificacao15} can we written as the usual form the statistical mechanics literature:

\begin{equation}
    \gamma_{\Lambda}(A|\,x) = \frac{\sum\limits_{\omega \in \mathcal{A}^{\Lambda}}e^{-H^{\Phi}_{\Lambda}(\omega x_{\Lambda^{c}})}\mathbbm{1}_{\{\omega x_{\Lambda^c}\in A\}}}{\sum\limits_{\eta \in \mathcal{A}^{\Lambda}}e^{-H^{\Phi}_{\Lambda}(\eta x_{\Lambda^{c}})}\mathbbm{1}_{\{\eta x_{\Lambda^c}\in X\}}}.
\end{equation}

Recall that the set of all nonempty finite subsets of $\zd$ is denoted by $\mathscr{S}$. 
Let $\mathscr{S}_{0}$ be an infinite subset of $\mathscr{S}$, and let $\Psi : \mathscr{S}_{0} \rightarrow \mathbb{R}$ be an arbitrary function. 
We will say that the infinite sum 
\[\sum\limits_{\Lambda \in \mathscr{S}_{0}} \Psi_{\Lambda}\]
exists and is equal to a real number $s$, if the net 
\[\left(\sum\limits_{\substack{\Lambda \in \mathscr{S}_{0} \\ \Lambda \subseteq \Delta}} \Psi_{\Lambda}\right)_{\Delta \in \mathscr{S}}\]
converges to $s$. In this case, we write
\[\sum\limits_{\Lambda \in \mathscr{S}_{0}} \Psi_{\Lambda} = s.\]
	
\begin{remark}\label{posisummmmm}
If $\Psi_{\Lambda} \geq 0$ for each $\Lambda$ in $\mathscr{S}_{0}$, then the sum $\sum\limits_{\Lambda \in \mathscr{S}_{0}} \Psi_{\Lambda}$ converges if and only if 
\[\sup\left\{\sum\limits_{\substack{\Lambda \in \mathscr{S}_{0} \\ \Lambda \subseteq \Delta}}\Psi_{\Lambda} : \Delta \in \mathscr{S}\right\}\]
is finite. 	
In either case, we have 
\[\sum\limits_{\Lambda \in \mathscr{S}_{0}} \Psi_{\Lambda} = \sup\left\{\sum\limits_{\substack{\Lambda \in \mathscr{S}_{0} \\ \Lambda \subseteq \Delta}}\Psi_{\Lambda} : \Delta \in \mathscr{S}\right\}.\]
\end{remark}

Given a family  $(\Phi_{\Lambda})_{\Lambda \in \mathscr{S}}$ of real-valued functions defined on $X$ such that $\sum\limits_{\Lambda \in \mathscr{S}_{0}} \Phi_{\Lambda}(x)$ exists at each point
$x$ in $X$, then we will denote by $\sum\limits_{\Lambda \in \mathscr{S}_{0}} \Phi_{\Lambda}$ the function which associates to each point $x$ in $X$ the sum $\sum\limits_{\Lambda \in \mathscr{S}_{0}} \Phi_{\Lambda}(x)$.

\begin{lemma}\label{lemmamalucooo}
Let $(\Phi_{\Lambda})_{\Lambda \in \mathscr{S}}$ be a family of real-valued bounded functions defined on $X$, and let $\mathscr{S}_{0}$ be an infinite subset of $\mathscr{S}$. Suppose that $\sum\limits_{\Lambda \in \mathscr{S}_{0}} \|\Phi_{\Lambda}\|_{\infty}$ converges. Then, it follows that
\begin{enumerate}[label=(\alph*),ref=\alph*]
\item the net $\left(\sum\limits_{\substack{\Lambda \in \mathscr{S}_{0} \\ \Lambda \subseteq \Delta}} \Phi_{\Lambda}\right)_{\Delta \in \mathscr{S}}$ converges uniformly to $\sum\limits_{\Lambda \in \mathscr{S}_{0}} \Phi_{\Lambda}$, and

%\item for every real number $c$, we have $\sum\limits_{\Lambda \in \mathscr{S}_{0}} c \;\Phi_{\Lambda} = c  \sum\limits_{\Lambda \in \mathscr{S}_{0}} \Phi_{\Lambda}$, and
 
\item\label{naointeressa} if $(c_{n})_{n \in \mathbb{N}}$ is a sequence of functions $c_{n} : \mathscr{S}_{0} \rightarrow \mathbb{R}$ which converges pointwise to a function
$c : \mathscr{S}_{0} \rightarrow \mathbb{R}$ and  
\[C := \sup\limits_{n \in \mathbb{N}}\sup\limits_{\Lambda \in \mathscr{S}_{0}}|c_{n}(\Lambda)| < + \infty,\] 
then 
\[\lim\limits_{n \to \infty} \left\|\sum\limits_{\Lambda \in \mathscr{S}_{0}} c_{n}(\Lambda) \Phi_{\Lambda} - \sum\limits_{\Lambda \in \mathscr{S}_{0}} c(\Lambda) \Phi_{\Lambda} \right\|_{\infty} =0.\]
\end{enumerate}	
\end{lemma}
\begin{proof}
For each positive number $\epsilon$ there is a set $\Delta_{0} \in \mathscr{S}$ such that 
\begin{equation}
\sum\limits_{\substack{\Lambda \in \mathscr{S}_{0},\, \Lambda \cap \Delta_{0}^{c} \neq \emptyset \\ \Lambda \subseteq \Delta}}\|\Phi_{\Lambda}\|_{\infty} = \sum\limits_{\substack{\Lambda \in \mathscr{S}_{0} \\ \Lambda \subseteq \Delta}}\|\Phi_{\Lambda}\|_{\infty}
- \sum\limits_{\substack{\Lambda \in \mathscr{S}_{0} \\ \Lambda \subseteq \Delta_{0}}}\|\Phi_{\Lambda}\|_{\infty} < \frac{\epsilon}{2}	
\end{equation}	
holds whenever $\Delta$ belongs to $\mathscr{S}$ and satisfies $\Delta_{0} \subseteq \Delta$. It follows that for every $\Delta$ and $\Delta'$ in $\mathscr{S}$ such that
$\Delta_{0} \subseteq \Delta$ and $\Delta_{0} \subseteq \Delta'$, we have
\begin{eqnarray*}
\left\|\sum\limits_{\substack{\Lambda \in \mathscr{S}_{0} \\ \Lambda \subseteq \Delta}}\Phi_{\Lambda}- \sum\limits_{\substack{\Lambda \in \mathscr{S}_{0} \\ \Lambda \subseteq \Delta'}}\Phi_{\Lambda}\right\|_{\infty} 	
&=& \left\|\sum\limits_{\substack{\Lambda \in \mathscr{S}_{0},\,\Lambda \cap \Delta_{0}^{c} \neq \emptyset \\ \Lambda \subseteq \Delta}}\Phi_{\Lambda} - \sum\limits_{\substack{\Lambda \in \mathscr{S}_{0},\,\Lambda \cap \Delta_{0}^{c} \neq \emptyset \\ \Lambda \subseteq \Delta'}}\Phi_{\Lambda}\right\|_{\infty}\\
&\leq& \sum\limits_{\substack{\Lambda \in \mathscr{S}_{0},\,\Lambda \cap \Delta_{0}^{c} \neq \emptyset \\ \Lambda \subseteq \Delta}}\|\Phi_{\Lambda}\|_{\infty} + \sum\limits_{\substack{\Lambda \in \mathscr{S}_{0},\,\Lambda \cap \Delta_{0}^{c} \neq \emptyset \\ \Lambda \subseteq \Delta'}}\|\Phi_{\Lambda}\|_{\infty} 
< \epsilon.
\end{eqnarray*}
We conclude that $\left(\sum\limits_{\substack{\Lambda \in \mathscr{S}_{0}\\ \Lambda \subseteq \Delta}}\Phi_{\Lambda}\right)_{\Delta \in \mathscr{S}}$ is a Cauchy net on the space of all real-valued bounded functions 
on $X$, thus part (a) follows.

For part (b), observe that since $|c_{n}(\Lambda)| \leq C$ holds for each $\Lambda$ and each $n$, it follows that $|c(\Lambda)| \leq C$ holds for each $\Lambda$. Thus
$\sum\limits_{\Lambda \in \mathscr{S}_{0}} \|c(\Lambda)\Phi_{\Lambda}\|_{\infty}$ converges, as well as each
sum $\sum\limits_{\Lambda \in \mathscr{S}_{0}} \|c_{n}(\Lambda)\Phi_{\Lambda}\|_{\infty}$. 

For each positive integer $n$ and every $\Delta$ and $\Delta_{0}$ in $\mathscr{S}$ such that $\Delta_{0}\subseteq \Delta$, we have
\begin{eqnarray*}
\Bigg\|\sum\limits_{\Lambda \in \mathscr{S}_{0}} c_{n}(\Lambda)\Phi_{\Lambda}&-&\sum\limits_{\Lambda \in \mathscr{S}_{0}}c(\Lambda)\Phi_{\Lambda}\Bigg\|_{\infty} \\
&\leq& \left\|\sum\limits_{\Lambda \in \mathscr{S}_{0}}c_{n}(\Lambda)\Phi_{\Lambda} - \sum\limits_{\substack{\Lambda \in \mathscr{S}_{0} \\ \Lambda \subseteq \Delta}}c_{n}(\Lambda)\Phi_{\Lambda}\right\|_{\infty} + \left\|\sum\limits_{\substack{\Lambda \in \mathscr{S}_{0} \\ \Lambda \subseteq \Delta}}c_{n}(\Lambda)\Phi_{\Lambda} - \sum\limits_{\substack{\Lambda \in \mathscr{S}_{0}\\ \Lambda \subseteq \Delta_{0}}}c_{n}(\Lambda)\Phi_{\Lambda}\right\|_{\infty} \\	
&&+ \left\|\sum\limits_{\substack{\Lambda \in \mathscr{S}_{0} \\ \Lambda \subseteq \Delta_{0}}}c_{n}(\Lambda)\Phi_{\Lambda} - \sum\limits_{\substack{\Lambda \in \mathscr{S}_{0} \\ \Lambda \subseteq \Delta_{0}}}c(\Lambda)\Phi_{\Lambda}\right\|_{\infty} + \left\|\sum\limits_{\substack{\Lambda \in \mathscr{S}_{0} \\ \Lambda \subseteq \Delta_{0}}}c(\Lambda)\Phi_{\Lambda} - \sum\limits_{\Lambda \in \mathscr{S}_{0}}c(\Lambda)\Phi_{\Lambda}\right\|_{\infty} \\
&\leq& \left\|\sum\limits_{\Lambda \in \mathscr{S}_{0}}c_{n}(\Lambda)\Phi_{\Lambda} - \sum\limits_{\substack{\Lambda \in \mathscr{S}_{0} \\ \Lambda \subseteq \Delta}}c_{n}(\Lambda)\Phi_{\Lambda}\right\|_{\infty} + C \sum\limits_{\substack{\Lambda \in \mathscr{S}_{0},\, \Lambda \cap \Delta_{0}^{c} \neq \emptyset \\ \Lambda \subseteq \Delta}}\|\Phi_{\Lambda}\|_{\infty} \\	
&&+ \max\limits_{\substack{\Lambda \in \mathscr{S}_{0} \\ \Lambda \subseteq \Delta_{0}}}|c_{n}(\Lambda) - c(\Lambda)| \cdot \sum\limits_{\substack{\Lambda \in \mathscr{S}_{0} \\ \Lambda \subseteq \Delta_{0}}} \|\Phi_{\Lambda}\|_{\infty} + \left\|\sum\limits_{\substack{\Lambda \in \mathscr{S}_{0} \\ \Lambda \subseteq \Delta_{0}}}c(\Lambda)\Phi_{\Lambda} - \sum\limits_{\Lambda \in \mathscr{S}_{0}}c(\Lambda)\Phi_{\Lambda}\right\|_{\infty}. 
\end{eqnarray*}
According to part (a), for each positive number $\epsilon$ there is an element $\Delta_{0}$ of $\mathscr{S}$ such that 
\[\left\|\sum\limits_{\substack{\Lambda \in \mathscr{S}_{0} \\ \Lambda \subseteq \Delta_{0}}}c(\Lambda)\Phi_{\Lambda} - \sum\limits_{\Lambda \in \mathscr{S}_{0}}c(\Lambda)\Phi_{\Lambda}\right\|_{\infty} < \frac{\epsilon}{4}\]
and
\[C \sum\limits_{\substack{\Lambda \in \mathscr{S}_{0},\, \Lambda \cap \Delta_{0}^{c} \neq \emptyset \\ \Lambda \subseteq \Delta}}\|\Phi_{\Lambda}\|_{\infty} < \frac{\epsilon}{4}\]
holds for each $\Delta$ in $\mathscr{S}$ satisfying $\Delta_{0} \subseteq \Delta$. And also, we can find a positive integer $n_{0}$ such that 
\[\max\limits_{\substack{\Lambda \in \mathscr{S}_{0} \\ \Lambda \subseteq \Delta_{0}}}|c_{n}(\Lambda) - c(\Lambda)| \cdot \sum\limits_{\substack{\Lambda \in \mathscr{S}_{0} \\ \Lambda \subseteq \Delta_{0}}} \|\Phi_{\Lambda}\|_{\infty} < \frac{\epsilon}{4}\]
holds whenever $n \geq n_{0}$.
We conclude that for every $n \geq n_{0}$, if we let $\Delta$ be an element of $\mathscr{S}$ such that $\Delta_{0} \subseteq \Delta$ and
\[\left\|\sum\limits_{\Lambda \in \mathscr{S}_{0}}c_{n}(\Lambda)\Phi_{\Lambda} - \sum\limits_{\substack{\Lambda \in \mathscr{S}_{0} \\ \Lambda \subseteq \Delta}}c_{n}(\Lambda)\Phi_{\Lambda}\right\|_{\infty} < \frac{\epsilon}{4},\]
we have
\[\left\|\sum\limits_{\Lambda \in \mathscr{S}_{0}} c_{n}(\Lambda)\Phi_{\Lambda}-\sum\limits_{\Lambda \in \mathscr{S}_{0}}c(A)\Phi_{\Lambda}\right\|_{\infty} < \epsilon.\]
\end{proof}

\begin{definition}
An \emph{interaction potential} is a family $\Phi = (\Phi_{\Lambda})_{\Lambda \in \mathscr{S}}$ of functions $\Phi_{\Lambda} : X \rightarrow \mathbb{R}$ such that
\begin{enumerate}[label=(\alph*),ref=\alph*]
\item for each $\Lambda \in \mathscr{S}$, the function $\Phi_{\Lambda}$ is $\mathcal{B}_{\Lambda}$-measurable, and	

\item for all $\Lambda \in \mathscr{S}$ and $x \in X$, the sum
\begin{equation}
H^{\Phi}_{\Lambda}(x) := \sum\limits_{\Delta \in \mathscr{S},\,\Delta \cap \Lambda \neq \emptyset}\Phi_{\Delta}(x)	
\end{equation}
converges.
\end{enumerate}
The quantity $H^{\Phi}_{\Lambda}(x)$ is called the \emph{energy} of $x$ in $\Lambda$ for the interaction potential $\Phi$, and the \emph{Hamiltonian} in $\Lambda$ for $\Phi$ is the function 
$H^{\Phi}_{\Lambda}$ which associates to each $x$ in $X$ the energy $H^{\Phi}_{\Lambda}(x)$. 
\end{definition}
\begin{remark}\label{depdepdep}
Given an arbitrary subset $\Lambda$ of $\zd$ and a $\mathcal{B}_{\Lambda}$-measurable function $f :X \rightarrow \mathbb{R}$,
the equality $f(x) =  f(y)$ holds whenever $x$ and $y$ are elements of $X$ such that $x_{\Lambda} = y_{\Lambda}$. The reader can easily verify that 
it suffices to prove this result for characteristic functions. Observe that 
\[\{B \subseteq X : \text{$\mathbbm{1}_{B}(x) = \mathbbm{1}_{B}(y)$ holds whenever $x_{\Lambda} = y_{\Lambda}$}\}\]
is a $\sigma$-algebra of subsets of $X$ which contains the collection
\[\mathscr{C} = \{\pi_{i}^{-1}(C) : i \in \Lambda, A \subseteq \mathcal{A}\}.\]
Since $\mathcal{B}_{\Lambda}$ is generated by $\mathscr{C}$, the result follows.
\end{remark}

\begin{example}\label{interactionising}
Let $X$ be the full shift $\{-1,+1\}^{\zd}$. Given two parameters $J$ and $h$ in $\mathbb{R}$, let us consider
\begin{equation}
\Phi^{J,h}_{\Lambda}(x) = 
\begin{cases}
-J x_{i}x_{j} &\text{if}\; \Lambda=\{i,j\}\; \text{and}\; i \sim j,\\
-h x_{i} &\text{if}\; \Lambda=\{i\},\\
0 &\text{otherwise.}	
\end{cases}	
\end{equation}	
The equation above defines an interaction potential $\Phi^{J,h} = (\Phi^{J,h}_{\Lambda})_{\Lambda \in \mathscr{S}}$ is called the Ising potential  
with coupling constant $J$ and external field $h$.
\end{example}

\begin{definition}
An interaction potential $\Phi = (\Phi_{\Lambda})_{\Lambda \in \mathscr{S}}$ is said to be 
\begin{enumerate}[label=(\alph*),ref=\alph*]
\item translation invariant if the relation
\begin{equation}
\Phi_{\Lambda} \circ \sigma^{i} = \Phi_{\Lambda + i}	
\end{equation}
holds for each $\Lambda \in \mathscr{S}$ and each $i \in \zd$, and 

\item absolutely summable if each $\Phi_{\Lambda}$ is bounded and satisfies 
\begin{equation}
\sum\limits_{\Lambda \in \mathscr{S},\, i \in \Lambda} \|\Phi_{\Lambda}\|_{\infty} < +\infty
\end{equation}
for each $i \in \zd$.
\end{enumerate}
\end{definition}

\begin{remark}
Observe that if $\Phi$ is absolutely summable, then the sum
$\sum\limits_{\Delta \in \mathscr{S},\,\Delta \cap \Lambda \neq \emptyset}\|\Phi_{\Delta}\|_{\infty}$
converges for each $\Lambda$ in $\mathscr{S}$. In fact, we have
\[\sum\limits_{\substack{\Delta \in \mathscr{S},\,\Delta \cap \Lambda \neq \emptyset \\ \Delta \subseteq \Delta'}}\|\Phi_{\Delta}\|_{\infty} \leq 
\sum_{i \in \Lambda}\sum\limits_{\substack{\Delta \in \mathscr{S},\, i \in \Delta \\ \Delta \subseteq \Delta'}}\|\Phi_{\Delta}\|_{\infty} 
\leq \sum_{i \in \Lambda}\sum\limits_{\Delta \in \mathscr{S},\, i \in \Delta}\|\Phi_{\Delta}\|_{\infty}\]
for each $\Delta'$ in $\mathscr{S}$. Thus, our assertion follows from Remark \ref{posisummmmm}.
Furthermore, one can easily verify that the potential $\Phi^{J,h}$ given in Example \ref{interactionising} is translation invariant and absolutely summable.
\end{remark}

In the following, given an absolutely summable potential $\Phi$, we will let $A_{\Phi}$ be a real-valued function defined on $X$ given by
\begin{equation}
A_{\Phi}(x) = -\sum\limits_{\Lambda \in \mathscr{S},\,\mathbf{0} \in \Lambda}\frac{1}{|\Lambda|}\,\Phi_{\Lambda}.	
\end{equation}
Observe that $A_{\Phi}$ is well defined since $\sum\limits_{\Lambda \in \mathscr{S},\,\mathbf{0} \in \Lambda}\frac{1}{|\Lambda|}\|\Phi_{\Lambda}\|_{\infty} \leq
\sum\limits_{\Lambda \in \mathscr{S},\,\mathbf{0} \in \Lambda}\|\Phi_{\Lambda}\|_{\infty} < +\infty$.

\begin{example}
Let $\Phi^{J,h}$ be the Ising potential defined in Example \ref{interactionising}. Then, the function $A_{\Phi^{J,h}}$ is given by
\begin{equation}
A_{\Phi^{J,h}}(x) = \frac{J}{2} \sum\limits_{j\, \sim\, \bf{0}}x_{\bf{0}}x_j  + h x_{\bf{0}}	
\end{equation}	
for each $x$ in $\{-1,+1\}^{\zd}$. 
%Observe that $A_{\Phi^{J,h}}$ coincides with $f^{J,h}$ (see Example \ref{ising2}).
\end{example}

\begin{theorem}\label{thm:Gibbsian}
Let $X \subseteq \mathcal{A}^{\zd}$ be a subshift, and let $\Phi$ be a translation invariant and absolutely summable potential. If we suppose that
the function $f = A_{\Phi}$ belongs to $SV_{d}(X)$, then the corresponding family $\gamma = (\gamma_{\Lambda})_{\Lambda \in \mathscr{S}}$ defined in Definition \ref{especificacao15}
is a Gibbsian specification, that is,
\begin{equation}
\gamma_{\Lambda}(A|x) = \frac{1}{Z^{\Phi}_{\Lambda}(x)} \int_{\mathcal{A}^{\Lambda}}e^{-H^{\Phi}_{\Lambda}(\zeta x_{\Lambda^{c}})}\mathbbm{1}_{\{\zeta x_{\Lambda^{c}} \in A\}} \lambda^{\Lambda}(d\zeta)	
\end{equation} 	
where $\lambda$ is the counting measure on $(\mathcal{A},\mathcal{P}(\mathcal{A}))$, and 
\begin{equation}
Z^{\Phi}_{\Lambda}(x) = \int_{\mathcal{A}^{\Lambda}}e^{-H^{\Phi}_{\Lambda}(\zeta x_{\Lambda^{c}})}\mathbbm{1}_{\{\zeta x_{\Lambda^{c}} \in X\}} \lambda^{\Lambda}(d\zeta).		
\end{equation}
\end{theorem}
\begin{proof}
Let $\Lambda$ be an element of $\mathscr{S}$, let $A$ be a Borel subset of $X$, and let $x$ be a point in $X$. For each positive integer $n$, we have
\begin{align*}
f_{n} &= -\sum\limits_{i \in \Lambda_{n}}\sum\limits_{\Delta \in \mathscr{S},\,\mathbf{0} \in \Delta} \frac{1}{|\Delta|}\,\Phi_{\Delta}\circ \sigma^{i}(x)
= -\sum\limits_{i \in \Lambda_{n}}\sum\limits_{\Delta \in \mathscr{S},\,\mathbf{0} \in \Delta} \frac{1}{|\Delta + i|}\,\Phi_{\Delta + i}(x)\\
&= -\sum\limits_{i \in \Lambda_{n}}\sum\limits_{\Delta \in \mathscr{S},\,i \in \Delta} \frac{1}{|\Delta|}\,\Phi_{\Delta}(x)
= -\sum\limits_{\Delta \in \mathscr{S}} \frac{|\Delta \cap \Lambda_{n}|}{|\Delta|}\,\Phi_{\Delta}(x) \\
&= -\sum\limits_{\Delta \in \mathscr{S},\,\Delta \cap \Lambda \neq \emptyset} \frac{|\Delta \cap \Lambda_{n}|}{|\Delta|}\,\Phi_{\Delta}(x) \; 
-\sum\limits_{\Delta \in \mathscr{S},\,\Delta \cap \Lambda = \emptyset} \frac{|\Delta \cap \Lambda_{n}|}{|\Delta|}\,\Phi_{\Delta}(x).
\end{align*}	
Using Remark \ref{depdepdep}, we obtain
\begin{eqnarray*}
&&\frac{\sum\limits_{\omega \in \mathcal{A}^{\Lambda}}e^{f_n(\omega x_{\Lambda^c})}\mathbbm{1}_{\{\omega x_{\Lambda^c}\in A\}}}{\sum\limits_{\eta\, \in \mathcal{A}^{\Lambda}}e^{f_n(\eta x_{\Lambda^c})}\mathbbm{1}_{\{\eta x_{\Lambda^c}\in X\}}} \\
&& \qquad\quad= \frac{\sum\limits_{\omega \in \mathcal{A}^{\Lambda}}\exp{\left(-\sum\limits_{\Delta \in \mathscr{S},\,\Delta \cap \Lambda \neq \emptyset} \frac{|\Delta \cap \Lambda_{n}|}{|\Delta|}\,\Phi_{\Delta}(\omega x_{\Lambda^c}) - 
\sum\limits_{\Delta \in \mathscr{S},\,\Delta \cap \Lambda = \emptyset} \frac{|\Delta \cap \Lambda_{n}|}{|\Delta|}\,\Phi_{\Delta}(\omega x_{\Lambda^c})\right)}\mathbbm{1}_{\{\omega x_{\Lambda^c}\in A\}}}
{\sum\limits_{\eta\, \in \mathcal{A}^{\Lambda}}\exp{\left(-\sum\limits_{\Delta \in \mathscr{S},\,\Delta \cap \Lambda \neq \emptyset} \frac{|\Delta \cap \Lambda_{n}|}{|\Delta|}\,\Phi_{\Delta}(\eta x_{\Lambda^c}) -	 
\sum\limits_{\Delta \in \mathscr{S},\,\Delta \cap \Lambda = \emptyset} \frac{|\Delta \cap \Lambda_{n}|}{|\Delta|}\,\Phi_{\Delta}(\eta x_{\Lambda^c})\right)}\mathbbm{1}_{\{\eta x_{\Lambda^c}\in X\}}}\\
&& \qquad\quad= \frac{\sum\limits_{\omega \in \mathcal{A}^{\Lambda}}\exp{\left(-\sum\limits_{\Delta \in \mathscr{S},\,\Delta \cap \Lambda \neq \emptyset} \frac{|\Delta \cap \Lambda_{n}|}{|\Delta|}\,\Phi_{\Delta}(\omega x_{\Lambda^c}) - 
\sum\limits_{\Delta \in \mathscr{S},\,\Delta \cap \Lambda = \emptyset} \frac{|\Delta \cap \Lambda_{n}|}{|\Delta|}\,\Phi_{\Delta}(x)\right)}\mathbbm{1}_{\{\omega x_{\Lambda^c}\in A\}}}
{\sum\limits_{\eta\, \in \mathcal{A}^{\Lambda}}\exp{\left(-\sum\limits_{\Delta \in \mathscr{S},\,\Delta \cap \Lambda \neq \emptyset} \frac{|\Delta \cap \Lambda_{n}|}{|\Delta|}\,\Phi_{\Delta}(\eta x_{\Lambda^c}) -	 
\sum\limits_{\Delta \in \mathscr{S},\,\Delta \cap \Lambda = \emptyset} \frac{|\Delta \cap \Lambda_{n}|}{|\Delta|}\,\Phi_{\Delta}(x)\right)}\mathbbm{1}_{\{\eta x_{\Lambda^c}\in X\}}}\\
&& \qquad\quad= \frac{\sum\limits_{\omega \in \mathcal{A}^{\Lambda}}\exp{\left(-\sum\limits_{\Delta \in \mathscr{S},\,\Delta \cap \Lambda \neq \emptyset} \frac{|\Delta \cap \Lambda_{n}|}{|\Delta|}\,\Phi_{\Delta}(\omega x_{\Lambda^c})\right)}\mathbbm{1}_{\{\omega x_{\Lambda^c}\in A\}}}
{\sum\limits_{\eta\, \in \mathcal{A}^{\Lambda}}\exp{\left(-\sum\limits_{\Delta \in \mathscr{S},\,\Delta \cap \Lambda \neq \emptyset} \frac{|\Delta \cap \Lambda_{n}|}{|\Delta|}\,\Phi_{\Delta}(\eta x_{\Lambda^c})\right)}\mathbbm{1}_{\{\eta x_{\Lambda^c}\in X\}}}, 	 	
\end{eqnarray*} 
and according to Lemma \ref{lemmamalucooo}(\ref{naointeressa}), we conclude that
\begin{eqnarray*}
\gamma_{\Lambda}(A|\,x) &=& \lim\limits_{n \to \infty}\frac{\sum\limits_{\omega \in \mathcal{A}^{\Lambda}}\exp{\left(-\sum\limits_{\Delta \in \mathscr{S},\,\Delta \cap \Lambda \neq \emptyset} \frac{|\Delta \cap \Lambda_{n}|}{|\Delta|}\,\Phi_{\Delta}(\omega x_{\Lambda^c})\right)}\mathbbm{1}_{\{\omega x_{\Lambda^c}\in A\}}}
{\sum\limits_{\eta\, \in \mathcal{A}^{\Lambda}}\exp{\left(-\sum\limits_{\Delta \in \mathscr{S},\,\Delta \cap \Lambda \neq \emptyset} \frac{|\Delta \cap \Lambda_{n}|}{|\Delta|}\,\Phi_{\Delta}(\eta x_{\Lambda^c})\right)}\mathbbm{1}_{\{\eta x_{\Lambda^c}\in X\}}} \\
&=& \frac{\sum\limits_{\omega \in \mathcal{A}^{\Lambda}}e^{-H^{\Phi}_{\Lambda}(\omega x_{\Lambda^{c}})}\mathbbm{1}_{\{\omega x_{\Lambda^c}\in A\}}}{\sum\limits_{\eta \in \mathcal{A}^{\Lambda}}e^{-H^{\Phi}_{\Lambda}(\eta x_{\Lambda^{c}})}\mathbbm{1}_{\{\eta x_{\Lambda^c}\in X\}}}.
\end{eqnarray*}
\end{proof}

\subsection{KMS-states} \label{subsec:KMS} We briefly present the notion of the Kubo-Martin-Schwinger (KMS) states \cite{HaagHugWin1967,Kubo1957,MartinSchwinger1959} and present the cocycle dynamics. Most of the general aspects here presented can be found in detail in \cite{Bratteli1987vol1,Bratteli1981vol2}. 

Basically, the notion of KMS state is defined under two objects: a dynamics and a set of norm-dense entire analytic elements of the C$^*$-algebra. The dynamics describe the system, and the analytic elements allow us to extend analytically the dynamics to the complex plane, and the density of these elements can be used to evaluate the extended dynamics on the rest of the C$^*$-algebra by approximations.

Given a C$^*$-algebra $A$, an \emph{one-parameter group of $\ast$-automorphisms} is a family $\tau = \{ \tau_t \}_{t \in \mathbb{R}}$ where each $\tau_t$ is a $*$-automorphism on $A$, and $\tau$ is additive with respect to the real parameter in the sense that $\tau_{s+t} = \tau_s \circ \tau_t$ for $s,t \in \mathbb{R}$, and $\tau_0$ is the identity. The pair $(A,\tau)$ is said to be a \emph{C$^*$-dynamical system} when $\tau$ is strongly continuous, i.e., for each $a \in A$, the mapping $t \mapsto \tau_t(a)$ is norm continuous, and in this case $\tau$ is said to be the \emph{dynamics on $A$}.

\begin{remark}\label{remark:dynamicsextension} It is sufficient to define the dynamics $\tau$ on a dense $*$-subalgebra, since in this case, it admits a unique extension to a dynamics on the whole C$^*$-algebra. %Proof of this fact can be found in Lemma 5.1.12 of \cite{Lima2019}.%
In addition, every $*$-automorphism on a C$^*$-algebra is an isometry and then $\| \tau_t \| = 1$ for every $t \in \mathbb{R}$. We also observe that the strong continuity of $\tau$ is equivalent to the $\sigma(A,A')$-continuity for the same map above (see Theorem 3.10 of \cite{Brezis2011}), and hence $\tau$ is a one-parameter $\sigma(A,A')$-continuous group of isometries, as in Definition 2.5.17 of \cite{Bratteli1987vol1}.
\end{remark}

Given a C$^*$-dynamical system $(A,\tau)$, an element $a\in A$ is said to be \emph{entire analytic} for $\tau$ when there exists a function $f: \mathbb{C} \to A$, satisfying the following:
\begin{itemize}
    \item[$(a)$] $f(t) = \tau_t(a)$ for every $t \in \mathbb{R}$;
    \item[$(b)$] the function $\eta \circ f$ is analytic on the whole $\mathbb{C}$ for every $\eta \in A'$.
\end{itemize}

We denote by $A_\tau$ the set of the entire analytic elements of the C$^*$-dynamical system $(A,\tau)$. $A_\tau$ is always non-empty and it is a norm dense $*$-subalgebra of $A$ (see Proposition 2.5.22 and Corollary 2.5.23 of \cite{Bratteli1987vol1}). Moreover, $A_\tau$ is $\tau$-invariant, i.e. $\tau_t(A_\tau) \subseteq A_\tau$ for every $t \in \mathbb{R}$ (see Lemma 5.1.22 of \cite{Lima2019}). In addition, on $A_\tau$ the rule of the composition of $*$-automorphisms is extended to the whole complex plane, that is, for every $z,w \in \mathbb{C}$ we have $\tau_{z+w}(a) = \tau_z \circ \tau_w(a)$ (see Lemma 5.1.21 of \cite{Lima2019}).

\begin{definition} A state on a C$^*$-algebra $A$ is a positive linear functional $\omega:A \to \mathbb{C}$ s.t. $\|\omega\| = 1$. 
\end{definition}

\begin{definition}[KMS states] %Frausino D14 pg 14 (22 do pdf)
For a given $(A,\tau)$ C$^*$-dynamical system and $\beta \in \mathbb{R}$, a state $\omega$ on $A$ is said to be a $\tau$-$\text{KMS}_\beta$-state when
\begin{align}\label{eq:KMS_condition}
\omega(a\tau_{i\beta}(b)) = \omega(ba)
\end{align}
for every $a, b$ in a $\tau$-invariant norm-dense $\ast$-subalgebra $A_0 \subseteq A_\tau$.
\end{definition}

We simply say $\omega$ is a $\text{KMS}_\beta$-state, that is, we omit $\tau$, when the dynamics is understood. Here, we are interested in the KMS states of groupoid C$^*$-algebras via continuous 1-cocycle $c:\mathcal{G} \to \mathbb{R}$ dynamics, and to connect it to the DLR measures in the case of the groupoid of the Gibbs relation. For a locally compact Hausdorff second countable \'etale groupoid $\mathcal{G}$, we define the one-parameter group of $*$-automorphisms $\tau$ by
\begin{align}\label{eq:groupoid_dynamics_C_star}
   \tau_t(f):= e^{it c} f, \quad f \in C_c(\mathcal{G}),
\end{align}
and since $C_c(\mathcal{G})$ is norm dense in $C^*(\mathcal{G})$, it is extended uniquely to the whole C$^*$-algebra. The next Theorem will connect the quasi-invariant measures to the KMS-states. It is straightforward to prove that $(C^*(\mathcal{G}),\tau)$ is a C$^*$-dynamical system.

\begin{remark} The dynamics here demand the cocycle to be continuous in order to be well-defined. However, the continuity of $c$ is not required for the conformal measures in Section \ref{rn}, that is, this restriction concerns only the KMS states. In particular, for the models given by a potential, in the regularity here studied, this is not restrictive since $d$-summable potentials give continuous associated cocycles.    
\end{remark}

The next theorem establishes a connection between the quasi-invariant measures and the KMS-states.

\begin{theorem}[Theorem 3.3.12 of \cite{renault2009}]\label{thm:Renault_KMS} Consider a locally compact \'{e}tale groupoid $\mathcal{G}$, a real-valued continuous 1-cocycle $c:\mathcal{G} \to \mathbb{R}$, and the one-parameter group of $*$-automorphisms $\tau$ in \eqref{eq:groupoid_dynamics_C_star}. Then, for every $\beta \in \mathbb{R}$,
\begin{itemize}
    \item[$(a)$] A probability measure $\mu$ on $\mathcal{G}^{(0)}$ is quasi-invariant under $\mathcal{G}$ with Radon-Nikodym derivative
    \begin{equation}\label{eq:RN_KMS}
        \frac{d\nu_r}{d\nu_s} = e^{-\beta c},
    \end{equation}
    if and only if the state is given by
    \begin{equation}\label{eq:KMS_quasi_inv}
        \varphi_\mu(f) := \int_{\mathcal{G}^{(0)}} f d\mu, \quad f \in C_c(\mathcal{G}) 
    \end{equation}
    is a KMS$_\beta$ on $C^*_r(\mathcal{G})$ (or $C^*(\mathcal{G})$);
    \item[$(b)$] If $c^{-1}(0)$ is principal, every KMS$_\beta$ state of $C^*(\mathcal{G})$ is in the form of \eqref{eq:KMS_quasi_inv} for some quasi-invariant probability measure $\mu$ on $\mathcal{G}^{(0)}$ with Radon-Nikodym derivative given by \eqref{eq:RN_KMS}.
\end{itemize}
\end{theorem}

An important comment must be made at this point: for a groupoid of an AP relation, the subgroupoid $c^{-1}(0)$ is always principal for every cocycle. Hence, the KMS states for the cocycle dynamics are precisely those that arise from the quasi-invariant measures with the Radon-Nikodym derivative \eqref{eq:RN_KMS}, namely the states in \eqref{eq:KMS_quasi_inv}. By applying Theorem \ref{thm:Renault_KMS} for the 1-cocycle generated by a potential $f \in SV_d(X)$ dynamics on C$^*(\gr)$, the following result holds.

\begin{corollary} Let $X \subseteq \mathcal{A}^{\zd}$ be a subshift and $f \in SV_d(X)$ be a potential on $X$. A Borel measure $\mu$ on $X$ is a Gibbs measure if and only if the state $\varphi_\mu$ on $C^*(\gr)$, given by
\begin{equation*}
    \varphi_\mu(f) := \int_X f d\mu, \quad f \in C_c(\gr), 
\end{equation*}
where $\gr$ is the Gibbs relation on $X$, is a KMS$_1$ state for the 1-cocycle dynamics, where the cocycle $c_f: \gr \to \mathbb{R}$ is given by
\begin{equation*}
    c_{f}(x,y) = \sum\limits_{k \in \zd}f(\sigma^k y) - f(\sigma^k x).
\end{equation*}
\end{corollary}

For the groupoid $\gr^0$, an analogous result also holds: since $\gr^0 \subseteq \gr$ and both groupoids have the same set of units, we also have that $c^{-1}(0)$ is also principal as a subgroupoid of $\gr^0$. We present such a result next.

\begin{corollary} Let $X \subseteq \mathcal{A}^{\zd}$ be a subshift and $f \in SV_d(X)$ be a potential on $X$. A Borel measure $\mu$ on $X$ is a topological Gibbs measure if and only if the state $\varphi_\mu$ on $C^*(\gr^0)$, given by
\begin{equation*}
    \varphi_\mu(f) := \int_X f d\mu, \quad f \in C_c(\gr^0), 
\end{equation*}
where $\gr^0$ is the topological Gibbs relation on $X$, is a KMS$_1$ state for the 1-cocycle dynamics, where the cocycle $c_f: \gr^0 \to \mathbb{R}$ is given by
\begin{equation*}
    c_{f}(x,y) = \sum\limits_{k \in \zd}f(\sigma^k y) - f(\sigma^k x).
\end{equation*}
\end{corollary}

\subsection{Gibbs specification as an operator and fixed point measure} In the context of (generalized) Markov shifts, one of the notions for DLR measure %(Theorem 6.16 of \cite{BEFR2018})
is that the Borel measure is a fixed point for the dual maps (in an integral sense) on the sequence of maps acting on the positive measurable functions on the shift space. Here, we briefly explain this notion and describe some results, which can be found in detail in \cite{BEFR2018}, while the notion of generalized Markov shift and its properties can be found in \cite{EL1999,BEFR2022}. This subsection has two purposes:
\begin{itemize}
    \item[$(A)$] to point out which sequence of maps is the analogous object for such a fixed point notion that works for subshifts of $\mathcal{A}^{\zd}$ in order to characterize the DLR measures;
    \item[$(B)$] to present sufficient conditions in the space on which the sequence of maps can be considered as a sequence of bounded operators, that is, when the fixed points are eigenmeasures with associate eigenvalue $1$.
\end{itemize}
We prove that for subshifts cases the SFT condition is an answer for $(B)$. For Markov shifts, a sufficient condition for these maps is compactness, which is essentially the same condition, because in this case, SFT is equivalent to compactness. We recall that for locally compact Markov shifts, the notions of standard and generalized Markov shifts coincide.

Example 2.3 of \cite{renault:05} presents the case of stationary equivalence relations as an example of AP equivalence relation, which is suitable to study locally compact one-sided Markov shift spaces $\Sigma_A$ endowed with the usual shift dynamics $\sigma$ (see \cite{sarig:09}). The AP relation in this case corresponds to the stationary sequence
\begin{equation*}
    \Sigma_A \stackrel{\sigma}{\to} \Sigma_A \stackrel{\sigma}{\to} \Sigma_A \stackrel{\sigma}{\to} \Sigma_A \stackrel{\sigma}{\to} \cdots,
\end{equation*}
where the corresponding AP relation is given by
\begin{equation*}
    R = \{(x,y) \in \Sigma_A \times \Sigma_A : \exists n \in \mathbb{N}_0, \sigma^n x = \sigma^n y\}.
\end{equation*}
For the non-locally compact case, we use the generalized countable Markov shift space (GCMS), that includes the standard shift space as a dense subset, and it is always locally compact (and compact in many cases). The complement of $\Sigma_A$ is a set that corresponds to finite words. In order to study its generalized thermodynamic formalism in the same generality as for models in standard Markov shifts, the shift map is not defined in the whole space, but in an open subset of shiftable words that contains $\Sigma_A$. In \cite{BEFR2022}, the generalized thermodynamic formalism was introduced and studied for the structure of the generalized Renault-Deaconu groupoid. In \cite{BEFR2018}, inspired by this generalized space, it was studied the Radon-Nikodym problem for the case where the dynamics map was not defined in the whole space. So in there, the notion \emph{generalized approximately proper equivalence relation} (GAP) was constructed. A GAP is a pair $(\{U_n\}_{n\in \mathbb{N}_0},\{R_n\}_{n\in \mathbb{N}_0})$ where the $U_n$'s are open sets, with $U_0 = X$, such that  $U_n \subseteq U_{n+1}$ for every $n\in \mathbb{N}_0$, and the $R_n$'s are proper equivalence relations, with $R_0 = \diag(X \times X)$ and $R_n \cap (U_n \times U_m) \subseteq R_m$ whenever $n \leq m$. For the GCMS case, the sets $U_n$ correspond to the open sets $\Dom (\sigma^n)$. These GAP relations also include AP relations, and among the results, we highlight the construction of a notion for the DLR measure via a sequence $\{Q_n\}_{n \in \mathbb{N}}$  of $\sigma$-additive maps %(see Definition A.2 of \cite{BEFR2018})
$Q_n:\mathcal{M}^+(X) \to \mathcal{M}^+(X)$, where $X$ is a locally compact second-countable metrizable space, and $\mathcal{M}^+(X)$ is the set of all non-negative Borel measurable functions on $X$. Each map $Q_n$ is defined by
\begin{equation*}
    Q_n(f)(x) = \begin{cases}
                \sum_{y \in R_n(x)} f(y) \rho_n(y) \zeta_n(y)^{-1} &\text{if } x \in U_n;\\
                0 &\text{otherwise}, 
             \end{cases}
\end{equation*}
where $R_n(x):= \{y \in X: (x,y) \in R_n\}$, $\rho_n : U_n \to \mathbb{R}$ is a positive function that is given by $\rho_n(x) = e^{h_n(x)}$, with $h_n(x)$ the Birkhoff sum for a potential $h$ in the context of Markov shifts, and $\zeta_n$ is the $n$-partition function given by $\zeta_n(x) = \sum_{y \in R_n(x)}\rho_n(y)$. Associated to $\rho_n$ there is a cocycle $c_n(x,y) = h_n(x) - h_n(y)$ on $R_n$. These cocycles are compatible in the sense of the existence of a unique cocycle on $R = \bigcup_{n \in \mathbb{N}_0} R_n$ such that $c\vert_{R_n} = c_n$%%(Proposition 6.4 of \cite{BEFR2018})
, and a Borel measure $\mu$ on $X$ is $e^c$-conformal if and only if $Q_n^*(\mu) = \mu$ for every $n \in \mathbb{N}$%(Theorem 6.16 of \cite{BEFR2018})
. These functions $Q_n$ correspond to the probability kernels for DLR measures in the compact $\Sigma_A$ (see Definition 6 of \cite{Cio:14}).

In general, the maps $Q_n$, $n\in \mathbb{N}$, cannot be defined as operators on $C(X)$, where $X$ is a locally compact Hausdorff second-countable metrizable space, even in the case of GCMS. In fact, for instance, the generalized renewal shift $X_A$ is compact, and hence $1 \in C(X_A)$. On the other hand, we have that $U_n = \Dom \sigma^n$, $n \in \mathbb{N}_0$, and each $U_n$ is open but not closed. For the potential $f \equiv 0$, we have
\begin{equation*}
    Q_n(1) = \mathbbm{1}_{U_n} \notin C(X_A),
\end{equation*}
and therefore $Q_n$ is not an operator. However, for compact $\Sigma_A$, the maps $Q_n$ can be defined as bounded operators on $C(\Sigma_A)$, as we prove next.

\begin{lemma}\label{lemma:Z_n_and_E_ng} Let $\Sigma_A$ be transitive and compact, $f:\Sigma_A \to \mathbb{R}$ a continuous potential, and $g \in C(\Sigma_A)$. Define the functions $Z_n: \Sigma_A \to \mathbb{R}$, $E_n^g: \Sigma_A \to \mathbb{R}$ by
\begin{equation*}
    Z_n(x) := \sum_{y \in R_n(x)} e^{f(y)} \quad \text{and} \quad E_n^g(x) :=  \sum_{y \in R_n(x)} e^{f(y)} g(y).
\end{equation*}
Then, $Z_n$ and $E_n^g$ are continuous functions, and $Z_n(x) > 0$ for every $x \in \Sigma_A$.
\end{lemma}

\begin{proof} First observe that both functions $Z_n$ and $E_n^g$ are well-defined since $R_n(x)$ is finite for every $x \in \Sigma_A$ because the alphabet is finite, and transitivity implies that $Z_n(x) > 0$ for all $x$. We prove that $Z_n$ is continuous, and the proof of the continuity of $E_n^g$ is similar. Let $(x^n)_{n \in \mathbb{N}}$ be a sequence in $\Sigma_A$ that converges to $x \in \Sigma_A$. Then, there exists $M \in \mathbb{N}$ such that, for every $m > M$ we have that $x_k^m = x_k$ for every $0 \leq k \leq n$, and in this case we have
\begin{equation*}
    R_n(x) = \left\{y \in \Sigma_A: y = w \sigma^n x, |w| = n \right\} \text{ and } R_n(x^m) = \left\{y \in \Sigma_A: y = w \sigma^n x^m, |w| = n \right\}.
\end{equation*}
By setting $B_n = \{w \in \mathcal{A}^n: w \text{ is admissible and } A_{w_{n-1},x_n} = 1\}$, we have
\begin{align*}
    |Z_n(x^m) - Z_n(x)| &= \left|\sum_{y \in R_n(x^m)} e^{f(y)} - \sum_{y \in R_n(x)} e^{f(y)}\right| = \left|\sum_{w \in B_n} e^{f(w\sigma^n x^m)} - e^{f(w\sigma^n x)}\right| \stackrel{m \to \infty}{\longrightarrow} 0.
\end{align*}
And therefore $Z_n$ is continuous.
\end{proof}

\begin{proposition} Let $\Sigma_A$ be transitive and compact, $f:\Sigma_A \to \mathbb{R}$ a continuous potential. Then $Q_n$ is a bounded operator on $C(\Sigma_A)$.
\end{proposition}

\begin{proof} Let $g \in C(\Sigma_A)$, then $Q_n(g)(x) = Z_n(x)^{-1} E_n^g(x)$, $x \in \Sigma_A$, and by Lemma \ref{lemma:Z_n_and_E_ng}, $Q_n(g)$ is a product of two continuous functions, and therefore it is continuous, i.e., $Q_n$ is well-defined, and its linearity is straightforward. Then $Q_n$ is an operator on $C(\Sigma_A)$. Now, we prove that $Q_n$ is bounded. In fact, for every $x \in \Sigma_A$ we have
\begin{align*}
    \left|Q_n(g)(x)\right| &= \left|Z_n(x)^{-1} \sum_{y \in R_n(x)} e^{f(y)} g(y)\right| \leq \left|Z_n(x)^{-1}\right| \left|Z_n(x)\right| \|g\| = \|g\|, 
\end{align*}
and therefore $\|Q_n\| \leq 1$, that is, $Q_n$ is bounded.
\end{proof}

Now, observe that $\Sigma_A$ is compact if and only if it is an SFT (see Section 1.3 of \cite{sarig:09}).

In parallel, in the context of subshifts of $\mathcal{A}^{\zd}$, the analogous object that plays the role of the kernel maps $Q_n$ is the family of specifications defined as $\sigma$-additive maps, as we see next. Consider a subshift $X \subseteq \mathcal{A}^{\zd}$ together with its Gibbs relation. In addition, let $f \in SV_d(X)$ be a potential. For each $\Lambda \in \mathscr{S}$, we define the maps $Q_\Lambda: \mathcal{M}^+(X) \to \mathcal{M}^+(X)$, given by
\begin{equation}\label{eq:kernel}
    Q_\Lambda(g)(x) := \sum\limits_{\omega \in \mathcal{A}^{\Lambda}} \gamma_{\Lambda}([\omega]|x) \mathbbm{1}_{\{\omega x_{\Lambda^c} \in X\}} g(\omega x_{\Lambda^c}), \quad g \in \mathcal{M}^+(X).
\end{equation}
Observe that, for every $A \in \mathcal{B}$ and $x \in X$, we have
\begin{equation*}
    Q_\Lambda(\mathbbm{1}_A)(x) = \gamma_\Lambda(A|\,x).
\end{equation*}

%This section is dedicated to connecting the notion of DLR measures via family $\{Q_n\}_{n \in \mathbb{N}}$ to the notions of Gibbsianess for subshifts of $\mathcal{A}^{\zd}$ presented so far in this paper. In our context, as we will discuss later, it is not possible to extend the results of \cite{BEFR2018} for the generality of potentials in $SV_d(X)$. On the other hand, the equivalence with such a construction particularly holds for local potentials.

\begin{remark} Since defined $Q_\Lambda$ via identity \eqref{lalalalallalalala}, which is a consequence of Lemma \ref{xenaaprincesaguerreira}, it follows that \eqref{gamma} can also be extended to $g \in \mathcal{M}^+(X)$ as
\begin{equation*}
    Q_\Lambda(g)(x) = \lim\limits_{n \to \infty}\frac{\sum\limits_{\omega \in \mathcal{A}^{\Lambda}}e^{f_n(\omega x_{\Lambda^c})}\mathbbm{1}_{\{\omega x_{\Lambda^c}\in X\}} g(\omega x_{\Lambda^c})}{\sum\limits_{\eta\, \in \mathcal{A}^{\Lambda}}e^{f_n(\eta x_{\Lambda^c})}\mathbbm{1}_{\{\eta x_{\Lambda^c}\in X\}}} 
\end{equation*} 
for each $x \in X$. In addition, for each $\Lambda \in \mathscr{S}$, $Q_{\Lambda}$ is $\sigma$-additive.
\end{remark}

Proposition A.2 of \cite{BEFR2018} ensures that every $Q_\Lambda$ has a unique dual map $Q_\Lambda^*$, in the sense that, for each $\Lambda \in \mathscr{S}$, and for any positive $\sigma$-finite Borel measure $\mu$, there exists a unique $\sigma$-finite positive Borel measure $Q_\Lambda^*(\mu)$ given by
\begin{equation*}
    \int g d Q_\Lambda^*(\mu) = \int Q_\Lambda(g) d \mu, \quad g \in \mathcal{M}^+(X).
\end{equation*}

On the other hand, Remark (1.24) of \cite{georgii:11} provides us with an equivalence between Gibbs measures and fixed point measures for the operators $Q_\Lambda$. We state it next using our notations. Let us recall that a family $\mathscr{S}_0 \subseteq \mathscr{S}$ is said to be cofinal when, for every $\Lambda \in \mathscr{S}$, there exists $\Delta \in \mathscr{S}_0$ such that $\Lambda \subseteq \Delta$.

\begin{remark}[Consequence of Remark (1.24) from \cite{georgii:11}]\label{remark:DLR_fixed_point_Georgii} Let $\mu$ be a Borel probability measure on $X$ and $\gamma = (\gamma_\Lambda)_{\Lambda \in \mathscr{S}}$ be a specification for the potential $f \in SV_d(X)$. The following conditions are equivalent:
\begin{itemize}
    \item[$(a)$] $\mu$ is a DLR measure;
    \item[$(b)$] $Q_\Lambda^*(\mu) = \mu$ for every $\Lambda \in \mathscr{S}$;
    \item[$(c)$] $Q_{\Lambda_{n}}^{*}(\mu) = \mu$ holds for every $n \in \mathbb{N}$;
    \item[$(d)$] There exists a cofinal family $\mathscr{S}_0 \subseteq \mathscr{S}$ such that $Q_\Lambda^*(\mu) = \mu$ for every $\Lambda \in \mathscr{S}_0$.
\end{itemize} 
\end{remark}

In the context of GAPs, the equivalence between $(a)$ and $(c)$ of the remark above was proved in Theorem 6.16 of \cite{BEFR2018} for the sequence of $Q_n$. For subshifts, it is straightforward that the fixed point property in Remark \ref{remark:DLR_fixed_point_Georgii} holds for every $\Lambda$ if and only if it holds for the family $\{\Lambda_n\}_{n \in \mathbb{N}}$. Next, we show that $X$ being an SFT is a sufficient condition that grants that the maps $Q_\Lambda$, $\Lambda \subseteq \zd$ finite set, can be defined as bounded operators on $C(X)$, and consequently, the operators $Q_{\Lambda}^*$ are adjoint operators of $Q_{\Lambda}^*$, in the functional analysis sense, and the DLR measures are eigenmeasures for the eigenvalue $1$.

Let us recall that in the case where $X \subseteq \mathcal{A}^{\zd}$ is an SFT, there exists a finite collection of patterns $\mathcal{F}$, which can be chosen to be defined on the same finite set $\Delta \subseteq \zd$%(see Remark 2.18 of \cite{Kimura})
, such that
\begin{equation*}
    X = \left\{x \in \mathcal{A}^{\zd} : \left(\sigma^ix\right)_\Delta \notin \mathcal{F}, \text{ for all $i \in \zd$}\right\}.
\end{equation*}

%\begin{lemma} Let $\Lambda \subseteq \zd$ be finite, and define

%\end{lemma}

\begin{lemma}\label{lemma:sequence_and_limit_belonging_SFT_equivalence} Let $X \subseteq \mathcal{A}^{\zd}$ be an SFT, $\Lambda \subseteq \zd$ finite, and $\omega \in \mathcal{A}^\Lambda$. Then, for any sequence $(x^{(n)})_{n \in \mathbb{N}}$ in $X$ that converges to $x \in X$, there exists $n_0 \in \mathbb{N}$ such that
\begin{equation}\label{eq:sequence_and_limit_belonging_SFT_equivalence}
    \text{$\omega x^{(n)}_{\Lambda^c} \in X \iff \omega x_{\Lambda^c} \in X \quad$  whenever $n \geq n_0$.}
\end{equation}
\end{lemma}  

\begin{proof} Define
\begin{equation*}
    \widetilde{\Lambda} = \Lambda \cup \left(\bigcup_{\substack{i \in \zd \\ (i + \Delta)\cap\Lambda \neq \emptyset}}(i+\Delta)\right),
\end{equation*}
and observe that $\widetilde{\Lambda}$ is finite. Since $(x^{(n)})_{n \in \mathbb{N}}$ converges to $x$, take $n_0 \in \mathbb{N}$ such that, for every $n \geq n_0$, we have $x^{(n)}_{\widetilde{\Lambda}} = x_{\widetilde{\Lambda}}$. We claim that, for every $n \geq n_0$, given $i \in \zd$, we have
\begin{equation}\label{eq:raluca}
    \sigma^i\left(\omega x^{(n)}_{\Lambda^c}\right)_{\Delta} = \begin{cases}
                                                    \sigma^i\left(\omega x_{\Lambda^c}\right)_{\Delta} \quad &\text{if } (i+\Delta) \cap \Lambda \neq \emptyset;\\
                                                    \sigma^i(x^{(n)})_{\Delta} \quad& \text{otherwise,}
                                                   \end{cases}
\end{equation}
and 
\begin{equation}\label{eq:raluca2}
    \sigma^i\left(\omega x_{\Lambda^c}\right)_{\Delta} = \begin{cases}
                                                    \sigma^i\left(\omega x^{(n)}_{\Lambda^c}\right)_{\Delta} \quad &\text{if } (i+\Delta) \cap \Lambda \neq \emptyset;\\
                                                    \sigma^i(x)_{\Delta} \quad& \text{otherwise.}
                                                   \end{cases}
\end{equation}
In fact, let us show that (\ref{eq:raluca}) holds. Let $n \geq n_0$. Given $i \in \zd$ such that $(i+\Delta) \cap \Lambda \neq \emptyset$, for each $j \in \Delta$ we have
\begin{equation*}
    \sigma^i\left(\omega x^{(n)}_{\Lambda^c}\right)_j = \left(\omega x^{(n)}_{\Lambda^c}\right)_{i+j} = \begin{cases}
                        \omega_{i+j} \quad & \text{if } i+j \in \Lambda,\\
                        x^{(n)}_{i+j} = x_{i+j} \quad & \text{otherwise,}
                    \end{cases}
\end{equation*}
in other words,
\begin{equation*}
    \sigma^i\left(\omega x^{(n)}_{\Lambda^c}\right)_j = \left(\omega x_{\Lambda^c}\right)_{i+j} = \sigma^i\left(\omega x_{\Lambda^c}\right)_j
\end{equation*}
holds for each $j \in \Delta$. On the other hand, if $(i+\Delta) \cap \Lambda = \emptyset$, it is straightforward to 
verify that $\sigma^i\left(\omega x^{(n)}_{\Lambda^c}\right)_{\Delta} = \sigma^i(x^{(n)})_{\Delta}$. The proof of (\ref{eq:raluca2}) is analogous; therefore, our claim is proven. We conclude that for each $n \geq n_0$ and $i \in \zd$, we have $\sigma^i\left(\omega x^{(n)}_{\Lambda^c}\right)_{\Delta} \notin \mathcal{F}$ if and only if $\sigma^i\left(\omega x_{\Lambda^c}\right)_{\Delta} \notin \mathcal{F}$.
\end{proof}

\begin{theorem} Let $X \subseteq \mathcal{A}^{\zd}$ be an SFT and $\Lambda \subseteq \zd$ finite. Then $Q_\Lambda$ is a bounded operator on $C(X)$.    
\end{theorem}

\begin{proof} Let $(x^{(n)})_{n \in \mathbb{N}}$ be a sequence in $X$ that converges to $x \in X$. It follows from Lemma \ref{lemma:sequence_and_limit_belonging_SFT_equivalence} that there exists $n_0 \in \mathbb{N}$ such that for each $\omega \in \mathcal{A}^{\Lambda}$,
\begin{equation*}
    \mathbbm{1}_{\{\omega x_{\Lambda^{c}}^{(n)} \in X\}} = \mathbbm{1}_{\{\omega x_{\Lambda^{c}} \in X\}}
\end{equation*}
holds for every $n \geq n_0$, and then,
\begin{align*}
    \gamma_\Lambda([\omega]|x^{(n)}) &= \frac{\mathbbm{1}_{\{\omega x^{(n)}_{\Lambda^c} \in X\}}}{\sum_{\eta \in \mathcal{A}^\Lambda} e^{c_f\left(\omega x^{(n)}_{\Lambda^c},\eta x^{(n)}_{\Lambda^c}\right)}\mathbbm{1}_{\{\eta x^{(n)}_{\Lambda^c}\in X\}}} = \frac{\mathbbm{1}_{\{\omega x_{\Lambda^c} \in X\}}}{\sum_{\eta \in \mathcal{A}^\Lambda} e^{c_f\left(\omega x^{(n)}_{\Lambda^c},\eta x^{(n)}_{\Lambda^c}\right)}\mathbbm{1}_{\{\eta x_{\Lambda^c}\in X\}}}.
\end{align*}
It is straightforward to show that Lemma \ref{lemma:sequence_and_limit_belonging_SFT_equivalence} implies that the map $\varphi_{\omega,\eta}$ defined by equation (\ref{eq:geradores}) belongs to $\mathcal{F}(X)$; moreover, Lemma 5.19. from \cite{Kimura} implies that the map $x \mapsto e^{c_f(x,\varphi_{\omega,\eta}(x))}$ is continuous. It follows that
\begin{equation*}
    \lim_{n \to \infty} \gamma_\Lambda([\omega]|x^{(n)}) = \gamma_\Lambda([\omega]|x),
\end{equation*}
thus, $\gamma_\Lambda([\omega]|\cdot)$ is continuous. Given a function $g \in C(X)$, the continuity of $Q_\Lambda(g)$ follows from equation (\ref{eq:kernel}). The linearity of $Q_\Lambda$ is straightforward. Now we prove that $Q_\Lambda$ is bounded. In fact,
\begin{align*}
    \left\|Q_\Lambda(g)\right\| &= \left\| \sum\limits_{\omega \in \mathcal{A}^{\Lambda}} \gamma_{\Lambda}([\omega]|x) \mathbbm{1}_{\{\omega x_{\Lambda^c} \in X\}} g(\omega x_{\Lambda^c}) \right\| \leq \|g\| \sum\limits_{\omega \in \mathcal{A}^{\Lambda}} \left\|\gamma_{\Lambda}([\omega]|\cdot)\right\| <\infty,
\end{align*}
where the last inequality comes from the fact that $\gamma_{\Lambda}([\omega]|\cdot)$ is a continuous and $X$ is compact.
%We recall that, for any $\omega \in \mathcal{A}^\Lambda$, if it does not exist $x \in X$ such that $\omega x_{\Lambda^c} \in X$, then we have $\gamma_{\Lambda}([\omega]|x) = 0$.
%\begin{align*}
%    \left\|\gamma_{\Lambda}([\omega]|\cdot)\right\| = \sup_{x \in X} \left|\sum_{\eta \in \mathcal{A}^\Lambda} e^{c_f\left(\omega x_{\Lambda^c},\eta x_{\Lambda^c}\right)}\mathbbm{1}_{\{\eta x_{\Lambda^c}\in X\}}\right|^{-1} = \left(\inf_{x \in X} \left|\sum_{\eta \in \mathcal{A}^\Lambda} e^{c_f\left(\omega x_{\Lambda^c},\eta x_{\Lambda^c}\right)}\mathbbm{1}_{\{\eta x_{\Lambda^c}\in X\}}\right|\right)^{-1} 
%\end{align*}
\end{proof}

\subsection{Capocaccia's definition}

In this subsection, we prove that in the case where the subshift $X$ is a STF, Definition \ref{gibbsm} is equivalent to the Gibbs state definition due to Capoccia \cite{cap:76}; therefore, items (a) through (f) from Theorem \ref{main_theorem} are equivalent to the item (g). Let us present another definition of Gibbs states, introduced by D. Capocaccia, that provides us with an understanding of such objects in the general context of compact metrizable spaces where $\zd$ acts by an expansive group of homeomorphisms. In the particular case where $X$ is a subshift and $T$ is the shift action of $\zd$, we show that this definition is deeply connected with the other notions described in section \ref{sec:medgibbs}. 

Let us start by considering a very general setting, where we assume that $X$ is a compact metrizable space and $T$ is an expansive continuous action of $\zd$ on $X$. 
Recall that the Gibbs relation of $(X,T)$ introduced in Example \ref{gr} is defined by

\[\gr(X,T) = \left\{(x,y) \in X \times X : \lim\limits_{\|k\| \to \infty} \rho(T^k x,T^k y) = 0 \right\},\]
where $\rho$ is a metric on $X$ that induces its topology, and the definition of $\gr$ does not depend on the choice of the metric $\rho$.
Translating Capocaccia's terminology into our setting, we say that two points $x$ and $y$ in $X$ are \emph{conjugate} if the pair $(x,y)$ belongs to $\gr$. 
And, if $O$ is an open subset of $X$, then a mapping $\varphi:O \rightarrow X$ is said to be \emph{conjugating} if $\rho(T^kx,T^k\varphi(x))$ tends uniformly to zero on $O$ as $\|k\|$ approaches infinity.
Note that the notion of a conjugating mapping also does not depend on the choice of the metric $\rho$.

Furthermore, Capocaccia assumed that the following condition was satisfied:
for every pair of conjugate points $x,y \in X$ there is an open subset $O$ of $X$ containing $x$ and a conjugating mapping
$\varphi: O \rightarrow X$ that is continuous at $x$ and satisfies $\varphi(x) = y$.
It was proved in \cite{cap:76} that this assumption implies that
for every such mapping $\varphi$ there is an open set $\widetilde{O} \subseteq O$
containing $x$ such that $\varphi$ is a homeomorphism of $\widetilde{O}$ onto $\varphi(\widetilde{O})$. 
Moreover, if $\varphi'$ is a mapping that shares the same
properties as $\varphi$, then $\varphi$ and $\varphi'$ agree on some neighborhood of $x$.

\begin{example}
Note that the assumption above is satisfied in the case where $X$ is a subshift of finite type and $T: i \mapsto \sigma^i$ is the $\zd$-shift action. Indeed, as we commented in Example \ref{gr}, $\gr(X,T)$ coincides with the Gibbs relation $\gr$ from definition \ref{gibbsrel_shift}. Since we have $\gr = \gr^{0}$, then for each pair $(x,y)$ in $\gr$ there is	
an element $\varphi$ of $\mathcal{F}(X)$ such that $y = \varphi(x)$. It is straightforward to show that $\varphi$ is a conjugating mapping.
\end{example}

By using the same notation as we adopted in Section \ref{rn}, for each Borel set $B$ of $X$ we will denote the restriction of a Borel measure $\mu$ on $X$ to the $\sigma$-algebra of Borel subsets of $B$ by $\mu_B$.

\begin{definition}[Capocaccia's definition for a Gibbs state]\label{def:Gibbs_state_Capocaccia}
We say that a family $\mathscr{I} = (R_{(O,\varphi)})$ is a family of multipliers if
\begin{itemize}
\item[(a)] $\mathscr{I}$ is indexed by all pairs $(O,\varphi)$, where $O$ is an open subset of $X$ and $\varphi$ is a conjugating homeomorphism defined on $O$, and 
$R_{(O,\varphi)}$ is a positive continuous function on $O$,
\item[(b)] if $O' \subseteq O$ and $\varphi' = \varphi\restriction_{O'}$, then $R_{(O',\varphi')} = R_{(O,\varphi)}\restriction_{O'}$, and
\item[(c)] if $O \subseteq O'$ and $\varphi'(O) \subseteq O''$, then
\begin{equation*}
R_{(O,\varphi''\circ \varphi'\restriction_{O})} = R_{(O',\varphi')}\restriction_{O}\cdot R_{(O'',\varphi'')}\circ \varphi'\restriction_{O}.
\end{equation*}
\end{itemize}
A Borel probability measure $\mu$ on $X$ is said to be a Gibbs state for the family of multipliers $\mathscr{I}$ if the condition 
\begin{equation}\label{gibbscap}
\varphi_{\ast}\big(R_{(O,\varphi)}\,d\mu_{O}\big) = \mu_{\varphi(O)}
\end{equation}
holds for every pair $(O,\varphi)$.
\end{definition}

\begin{remark}
Observe that expression (\ref{gibbscap}) is well-defined. In fact,
according to Theorem $8.3.7$ from \cite{cohn:13}, the image $\varphi(O)$ of a one-to-one measurable function $\varphi$ from a Borel 
subset $O$ of a Polish space $X$ into another Polish space $Y$, is Borel set of $Y$.  
Thus, for each pair $(O,\varphi)$, since $O$ is a Borel subset of $X$, it follows that the same holds for its image $\varphi(O)$. Furthermore,
it is straightforward to verify that $\mu$ is a Gibbs state for the family $\mathscr{I}$ if and only if for each pair $(O,\varphi)$ the equation
\begin{equation}\label{gibbscapp}
\frac{d(\mu_{\varphi(O)}\circ \varphi)}{d\mu_{O}} = R_{(O,\varphi)}
\end{equation}
holds $\mu$-almost everywhere on $O$.
\end{remark}

\begin{theorem}
Let $X$ be an SFT, and let $f$ be a function in $SV_d(X)$. Then, a Borel probability measure $\mu$ on $X$ is a Gibbs measure for
$f$ if and only if $\mu$ is a Gibbs state for the family of multipliers $\mathscr{I}$ defined by letting
\[R_{(O,\varphi)}(x) = e^{c_f(x,\varphi(x))}\quad \text{at each point $x \in O$,}\]
for all pairs $(O,\varphi)$.
\end{theorem}

\begin{proof}
It is straightforward to check that $\mathscr{I}=(R_{(O,\varphi)})$ is a family of multipliers.
Let $\varphi$ be a conjugating homeomorphism defined on $O$. Since 
$\varphi : O \rightarrow \varphi(O)$ is a Borel isomorphism such that $\mathsf{gr}(\varphi) \subseteq \gr$, it follows from 
Proposition \ref{boraut} that
\[\frac{d\mu_{\varphi(O)}\circ \varphi}{d\mu_O}(x)= D_{\mu,\gr}(\varphi(x),x) = e^{c_f(x,\varphi(x))}\]
holds for $\mu$-almost every point $x$ in $O$.  

Conversely, since each element $\varphi$ of $\mathcal{F}(X)$ is a conjugating homeomorphism, then the equation
\[\frac{d\mu\circ\varphi}{d\mu} = e^{c_{f}(x,\varphi(x))}\]
holds for $\mu$-almost every $x$ in $X$.
It follows from Remark \ref{capeeeta} that $\mu$ is a topological Gibbs measure for $f$, and using the fact that $X$ is a SFT, 
we conclude that $\mu$ is a Gibbs measure for $f$.
\end{proof}

\section{Concluding remarks}

This paper proves the equivalence between different notions of Gibbsianess. We also discuss other notions of Gibbs measures that are also equivalent to some of those proven here, therefore, equivalent to the rest of the notions in most cases, and we provide a list of equivalences combining our results with the parts of the literature that were not connected yet. Recently, L. Borsato and S. MacDonald in \cite{Merda} proved the equivalence between Gibbs measures (Definition \ref{gibbsm}) and DLR measures. For one side, the proof that every Gibbs measure is a DLR measure (Theorem 6 of \cite{Merda}) is essentially the same as the proof of Theorem \ref{DLR} and relies on the cocycle formula to set the Gibbs specification for the context of subshifts over a finite alphabet on countable groups. The proof shown in this paper for a subshift of $\mathcal{A}^{\zd}$ is also present in \cite{Kimura}. However, we highlight their new result, which is the converse of the aforementioned implication where the condition of $X$ being a SFT is dropped. Also recently, using a different proof, C.-E. Pfister in \cite{Carlinhos2022} proved the equivalence between the notions of Gibbs measures and DLR measures. The results of his work are general in two different directions: one concerning a countable group $G$ acting on a compact ultrametric space endowed with a continuous potential whose cocycle is well-defined, which in particular includes $G = \zd$, $X \subseteq \mathcal{A}^{\zd}$, $|\mathcal{A}| < \infty$, and a $d$-summable potential; and another that removes the group structure of the countable set $G$ and indexes a different but finite alphabet for each $g \in G$, where the potential corresponds to a family of continuous functions indexed by $G$ with an analogous cocycle regularity. We do not touch the equivalence between invariant Gibbs measures and equilibrium measures. The reader can find this equivalence in different settings on \cite{BBBB2023, keller:98, haydn:87, HaydnRuelle1992, muir, Muir_paper, ruelle:04}.

Our results open possibilities for further research directions. We highlight the quantum approach of non-locally compact shift spaces, as we explain as follows. In \cite{muir,Muir_paper}, S. Muir introduced a notion of Gibbs measures for $\mathbb{N}^{\zd}$ via Radon-Nikodym associated with the group action of permutations and exponential summable potential. In this setting, his definition agrees with the notion of DLR measure for the infinite alphabet case, and possibly it admits a groupoid approach. The generalization of the thermodynamic formalism (Gibbs, conformal, equilibrium measures, pressure etc) from $\mathbb{N}^{\zd}$ to $\mathbb{N}^{G}$, where $G$ is a countable amenable group, was recently obtained in \cite{BBBB2023}. However, the lack of local compactness of the space turns into a challenge to treat the problem via groupoid theory: the notion of AP equivalence relation \cite{renault:05} requires that the shift space be at least locally compact. Further, there is also the concern about constructing an adequate C$^*$-algebra that is the operator algebra background for these shift spaces. Perhaps, a similar path to the one done in \cite{BEFR2022} could compactify the space and generate bigger equivalence relation than the Gibbs for the shift, and then it would be possible to construct the respective C$^*$-algebra via AP relations or its generalized version \cite{BEFR2018}. Questions concerning the existence of new measures in this possible compactification and phase transition phenomena arise via \cite{BEFR2022}.

\section*{Acknowledgements}
RB is supported by CNPq grants 312294/2018-2 and 408851/2018-0, by FAPESP grant 16/25053-8, and by the University Center of Excellence \textquotedblleft Dynamics, Mathematical Analysis and Artificial Intelligence", at the Nicolaus Copernicus University; BHF-K is supported by Japan Science and Technology Agency grant PJ22180021;  TR is supported by NCN (National Science Center, Poland), Grant 2019/35/D/ST1/01375. We thank O. Karpel and A. Sakai for their hospitality concerning the visit of BHF-K at AGH University of Science and Technology.
\section{Appendix}

\bibliographystyle{alpha}
\bibliography{bibliografia}

\begin{thebibliography}{aHLRS15}

\bibitem[AD97]{ADelaroche:97}
C.~Anantharaman-Delaroche.
\newblock Purely infinite {$C^*$}-algebras arising from dynamical systems.
\newblock {\em Bull. de la Soc. Math. de France}, 125(2):199--225, 1997.

\bibitem[ADR00]{renaultamenable}
C.~Anantharaman-Delaroche and J.~Renault.
\newblock {\em Amenable Groupoids}.
\newblock L'Enseignment math{\'e}matique. L'Enseignement Math{\'e}matique,
  2000.

\bibitem[aHKR15]{aHKR:15}
A.~an~Huef, S.~Kang, and I.~Raeburn.
\newblock Spatial realisations of {KMS} states on the {$C^*$}-algebras of
  higher-rank graphs.
\newblock {\em J. Math. Anal. Appl.}, 427(2):977--1003, 2015.

\bibitem[aHKR17]{aHKR:17}
A.~an~Huef, S.~Kang, and I.~Raeburn.
\newblock K{MS} states on the operator algebras of reducible higher-rank
  graphs.
\newblock {\em Integral Equ. Oper. Theory}, 88(1):91--126, 2017.

\bibitem[aHLRS13]{aHLRS:13}
A.~an~Huef, M.~Laca, I.~Raeburn, and A.~Sims.
\newblock K{MS} states on the {$C^*$}-algebras of finite graphs.
\newblock {\em J. Math. Anal. Appl.}, 405(2):388--399, 2013.

\bibitem[aHLRS14]{aHLRS:14}
A.~an~Huef, M.~Laca, I.~Raeburn, and A.~Sims.
\newblock K{MS} states on {$C^*$}-algebras associated to higher-rank graphs.
\newblock {\em J. Funct. Anal.}, 266(1):265--283, 2014.

\bibitem[aHLRS15]{aHLRS:15}
A.~an~Huef, M.~Laca, I.~Raeburn, and A.~Sims.
\newblock K{MS} states on the {$C^*$}-algebra of a higher-rank graph and
  periodicity in the path space.
\newblock {\em J. Funct. Anal.}, 268(7):1840--1875, 2015.

\bibitem[AN07]{aaronson:07}
J.~Aaronson and H.~Nakada.
\newblock Exchangeable, {G}ibbs and equilibrium measures for {M}arkov
  subshifts.
\newblock {\em Ergod. Theory Dyn. Syst.}, 27(2):321--339, 2007.

\bibitem[BBBB23]{BBBB2023}
E.~R. Beltr\'an, R.~Bissacot, L.~Borsato, and R.~Briceno.
\newblock Thermodynamic formalism for amenable groups and countable state
  spaces. {A}r{X}iv:2302.05557, 2023.

\bibitem[BEFR18]{BEFR2018}
R.~Bissacot, R.~Exel, R.~Frausino, and T.~Raszeja.
\newblock Quasi-invariant measures for generalized approximately proper
  equivalence relations. {A}r{X}iv:1809.02461, 2018.

\bibitem[BEFR22]{BEFR2022}
R.~Bissacot, R.~Exel, R.~Frausino, and T.~Raszeja.
\newblock Thermodynamic formalism for generalized {M}arkov shifts on infinitely
  many states. {A}r{X}iv:1808.00765, 2022.

\bibitem[BK49]{Bogo:49}
N.~N. Bogolyubov and B.~I. Khatset.
\newblock On some mathematical problems of the theory of statistical
  equilibrium.
\newblock {\em Dokl. Akad. Nauk 66}, 321-324. MR 11-40, 1949.

\bibitem[BM21]{Merda}
L.~Borsato and S.~MacDonald.
\newblock Conformal measures and the {D}obrušin-{L}anford-{R}uelle equations.
\newblock {\em Proc. Am. Math. Soc.}, 149:1, 03 2021.

\bibitem[Bow08]{bowen:08}
R.~Bowen.
\newblock {\em Equilibrium States and the Ergodic Theory of {A}nosov
  Diffeomorphisms}.
\newblock Springer-Verlag, Berlin, second revised edition, 2008.

\bibitem[BPK69]{Bogo:69}
N.~N. Bogolyubov, O.~Ya. Petrina, and B.~I. Khatset.
\newblock A mathematical description of the equilibrium state of classical
  systems on the basis of a formalism of a canonical ensemble.
\newblock {\em Teoret. Mat. Fiz.}, 1:251--274, 1969.

\bibitem[BR81]{Bratteli1981vol2}
O.~Bratteli and D.~W. Robinson.
\newblock {\em Operator algebras and quantum statistical mechanics. Vol. 2:
  Equilibrium states. Models in quantum statistical mechanics}.
\newblock Springer Nature, 1981.

\bibitem[BR86]{Bratteli1987vol1}
O.~Bratteli and D.~W. Robinson.
\newblock {\em Operator Algebras and Quantum Statistical Mechanics. Vol. 1:
  C$^{\ast}$- and W$^{\ast}$-algebras, symmetry groups, decomposition of
  states}.
\newblock Springer Nature, 1986.

\bibitem[Bre11]{Brezis2011}
H.~Brezis.
\newblock {\em Functional analysis, {S}obolev spaces and partial differential
  equations}.
\newblock Springer, New York London, 2011.

\bibitem[CaHLS23]{Rafa:23}
L.~O. Clark, A.~an~Huef, R.~P. Lima, and C.~Sehnem.
\newblock Equivalence of two definitions of af groupoid.
\newblock {\em Manuscript in preparation}, 2023.

\bibitem[Cap76]{cap:76}
D.~Capocaccia.
\newblock {A} definition of {G}ibbs state for a compact set with
  $\mathbb{Z}^\nu$ action.
\newblock {\em Comm. Math.Phys.}, 48:85--88, 1976.

\bibitem[Chr23]{christensen2023structure}
J.~Christensen.
\newblock The structure of {K}{M}{S} weights on {\'e}tale groupoid
  {C}*-algebras.
\newblock {\em J. of Noncommut. Geom.}, 17(2):663–691, 2023.

\bibitem[CK80]{CK1980}
J.~Cuntz and Wolfgang K.
\newblock A class of {C}*-algebras and topological {Markov} chains.
\newblock {\em Invent. {M}ath.}, 56(3):251--268, Oct 1980.

\bibitem[CL17]{Cio:14}
L.~Cioletti and A.~O. Lopes.
\newblock Interactions, specifications, dlr probabilities and the ruelle
  operator in the one-dimensional lattice.
\newblock {\em Discrete Contin. Dyn. Syst.}, 37(12):6139--6152, 2017.

\bibitem[Coh13]{cohn:13}
D.~Cohn.
\newblock {\em Measure Theory}.
\newblock Birkh\"auser, second edition, 2013.

\bibitem[CT21]{ChristensenThomsen:21}
J.~Christensen and K.~Thomsen.
\newblock K{MS} states on crossed products by abelian groups.
\newblock {\em Math. Scand.}, 127(3):488--508, 2021.

\bibitem[CT22]{ChristensenThomsen:22}
J.~Christensen and K.~Thomsen.
\newblock K{MS} states on the crossed product {$C^*$}-algebra of a
  homeomorphism.
\newblock {\em Ergod. Theory Dyn. Syst.}, 42(4):1373--1414, 2022.

\bibitem[Cun77]{Cuntz1977}
J.~Cuntz.
\newblock Simple {C}*-algebra generated by isometries.
\newblock {\em Commun. {M}ath. {P}hys.}, 57(2):173--185, June 1977.

\bibitem[CV22]{ChristensenVaes:22}
J.~Christensen and S.~Vaes.
\newblock K{MS} spectra for group actions on compact spaces.
\newblock {\em Comm. Math. Phys.}, 390(3):1341--1367, 2022.

\bibitem[Dea95]{Deaconu:95}
V.~Deaconu.
\newblock Groupoids associated with endomorphisms.
\newblock {\em Trans. Amer. Math. Soc.}, 347(5):1779--1786, 1995.

\bibitem[Dob68]{dob:68}
R.~L. Dobrushin.
\newblock Gibbsian random fields for lattice systems with pairwise
  interactions.
\newblock {\em Funkcional. Anal. i Prilozen.}, 2(4):31--43, 1968.

\bibitem[DU91]{DU91}
M.~Denker and M.~Urbański.
\newblock On the existence of conformal measures.
\newblock {\em Trans. {A}m. {M}ath. {S}oc.}, 328(2):563--587, 1991.

\bibitem[EFW84]{EFW:1984}
M.~Enomoto, M.~Fujii, and Y.~Watatani.
\newblock K{MS} states for gauge action on {$O\sb{A}$}.
\newblock {\em Math. Japon.}, 29(4):607--619, 1984.

\bibitem[EL99]{EL1999}
R.~Exel and M.~Laca.
\newblock Cuntz-krieger algebras for infinite matrices.
\newblock {\em J. Reine Angew. Math.}, 512:119--172, 1999.

\bibitem[EL03]{EL2003}
R.~Exel and M.~Laca.
\newblock Partial dynamical systems and the {KMS} condition.
\newblock {\em Commun. Math. Phys.}, 232:223--277, 2003.

\bibitem[EP22]{exelpitts:22}
R.~Exel and D.~R. Pitts.
\newblock {\em Characterizing Groupoid C$^{\ast}$-algebras of Non-Hausdorff
  {\'E}tale Groupoids}.
\newblock Lecture Notes in Mathematics. Springer International Publishing,
  2022.

\bibitem[Eva80]{Evans:80}
D.~E. Evans.
\newblock On {$O\sb{n}$}.
\newblock {\em Publ. Res. Inst. Math. Sci.}, 16(3):915--927, 1980.

\bibitem[Exe04]{Exel:04}
R.~Exel.
\newblock K{MS} states for generalized gauge actions on {C}untz-{K}rieger
  algebras (an application of the {R}uelle-{P}erron-{F}robenius theorem).
\newblock {\em Bull. Braz. Math. Soc. (N.S.)}, 35(1):1--12, 2004.

\bibitem[FGLP21]{FGLP:21}
C.~Farsi, E.~Gillaspy, N.~S. Larsen, and J.~A. Packer.
\newblock Generalized gauge actions on {$k$}-graph {$C^*$}-algebras: {KMS}
  states and {H}ausdorff structure.
\newblock {\em Indiana Univ. Math. J.}, 70(2):669--709, 2021.

\bibitem[FHKP22]{FHKP:22}
C.~Farsi, L.~Huang, A.~Kumjian, and J.~Packer.
\newblock Cocycles on groupoids arising from {$\mathbb{N}^k$}-actions.
\newblock {\em Ergod. Theory Dyn. Syst.}, 42(11):3325--3356, 2022.

\bibitem[FKPS19]{Farsi:19}
C.~Farsi, A.~Kumjian, D.~Pask, and A.~Sims.
\newblock Ample groupoids: Equivalence, homology, and {M}atui's {HK}
  conjecture.
\newblock {\em M\"{u}nster J. of Math.}, 12:411--451, 2019.

\bibitem[FM77]{feldman:77}
J.~Feldman and C.~C. Moore.
\newblock Ergodic equivalence relations, cohomology, and von {N}eumann
  algebras. {I}.
\newblock {\em Trans. Amer. Math. Soc.}, 234(2):289--324, 1977.

\bibitem[FV17]{velenik2017}
S.~Friedli and Y.~Velenik.
\newblock {\em Statistical Mechanics of Lattice Systems: A Concrete
  Mathematical Introduction}.
\newblock Cambridge University Press, 2017.

\bibitem[Geo11]{georgii:11}
H.~Georgii.
\newblock {\em {G}ibbs Measures and Phase Transitions}.
\newblock De Gruyter Studies in Mathematics 9. De Gruyter, second edition,
  2011.

\bibitem[Hay87]{haydn:87}
N.~T.~A. Haydn.
\newblock On {G}ibbs and equilibrium states.
\newblock {\em Ergod. Theory Dyn. Syst.}, 7:119--132, 1987.

\bibitem[HHW67]{HaagHugWin1967}
R.~Haag, N.~M. Hugenholtz, and M.~Winnink.
\newblock On the equilibrium states in quantum statistical mechanics.
\newblock {\em Commun. Math. Phys.}, 5(3):215--236, 1967.

\bibitem[HR92]{HaydnRuelle1992}
N.~T.~A. Haydn and D.~Ruelle.
\newblock Equivalence of {G}ibbs and equilibrium states for homeomorphisms
  satisfying expansiveness and specification.
\newblock {\em Commun. Math. Phys.}, 148(1):155--167, August 1992.

\bibitem[IR69]{lanfordruelle:69}
O.~E.~Lanford III and D.~Ruelle.
\newblock Observables at {I}nfinity and {S}tates with {S}hort {R}ange
  {C}orrelations in {S}tatistical {M}echanics.
\newblock {\em Comm. Math.Phys.}, 13:194--215, 1969.

\bibitem[Kel98]{keller:98}
G.~Keller.
\newblock {\em Equilibrium States in Ergodic Theory}.
\newblock London Mathematical Society Student Texts $42$. Cambridge University
  Press, 1998.

\bibitem[Kim15]{Kimura}
B.H.F. Kimura.
\newblock Gibbs measures on subshifts.
\newblock Master's thesis, Instituto de Matemática e Estatística, University
  of São Paulo, São Paulo, 2015.

\bibitem[KP00]{KumjianPask:00}
A.~Kumjian and D.~Pask.
\newblock Higher rank graph {$C^\ast$}-algebras.
\newblock {\em New York J. Math.}, 6:1--20, 2000.

\bibitem[KPRR97]{KPRR1997}
A.~Kumjian, D.~Pask, I.~Raeburn, and J.~Renault.
\newblock Graphs, groupoids, and {C}untz-{K}rieger algebras.
\newblock {\em J. Funct. Anal.}, 144(2):505--541, 1997.

\bibitem[KSS07]{KessStadStrat2007}
M.~Kesseb\"{o}hmer, M.~Stadlbauer, and B.~O. Stratmann.
\newblock Lyapunov spectra for {KMS} states on cuntz-krieger algebras.
\newblock {\em Math. Zeitschrift}, 256(4):871--893, March 2007.

\bibitem[Kub57]{Kubo1957}
R.~Kubo.
\newblock Statistical-mechanical theory of irreversible processes. {I}. general
  theory and simple applications to magnetic and conduction problems.
\newblock {\em J. Phys. Soc. Japan}, 12(6):570--586, 1957.

\bibitem[Lim19]{Lima2019}
R.~P. Lima.
\newblock {Characterization of extremal KMS states on groupoid
  C$^{\ast}$-algebras}.
\newblock Master's thesis, Instituto de Matemática e Estatística, University
  of São Paulo, São Paulo, 2019.

\bibitem[LY19]{LiYang:19}
H.~Li and D.~Yang.
\newblock K{MS} states of self-similar {$k$}-graph {$\rm C^*$}-algebras.
\newblock {\em J. Funct. Anal.}, 276(12):3795--3831, 2019.

\bibitem[Mey13]{meyerovitch:13}
T.~Meyerovitch.
\newblock {G}ibbs and equilibrium measures for some families of subshifts.
\newblock {\em Ergod. Theory Dyn. Syst.}, 33:934--953, 2013.

\bibitem[Min67]{minlos:67}
R.~A. Minlos.
\newblock The {G}ibbs limit distribution.
\newblock {\em Funktsional. Anal, i Prilozhen.}, 1:40--54, 1967.

\bibitem[MS59]{MartinSchwinger1959}
P.~C. Martin and J.~Schwinger.
\newblock Theory of many-particle systems. i.
\newblock {\em Phys. Rev.}, 115(6):1342, 1959.

\bibitem[Mui11a]{muir}
S.~Muir.
\newblock {\em {G}ibbs/equilibrium measures for functions of multidimensional
  shifts with countable alphabets}.
\newblock dissertation, University of North Texas, 2011.

\bibitem[Mui11b]{Muir_paper}
S.~Muir.
\newblock A new characterization of {G}ibbs measures on {$\mathbb{N}^{\zd}$}.
\newblock {\em Nonlinearity}, 24(10):2933--2952, sep 2011.

\bibitem[Mur90]{Murphy:90}
G.~J. Murphy.
\newblock {\em C$^{\ast}$-Algebras and Operator Theory}.
\newblock Elsevier Science, 1990.

\bibitem[Nes13]{Neshveyev:2013}
S.~Neshveyev.
\newblock K{MS} states on the {$C^\ast$}-algebras of non-principal groupoids.
\newblock {\em J. Oper. Theory}, 70(2):513--530, 2013.

\bibitem[Ny08]{ny:08}
A.~Le Ny.
\newblock {\em Introduction to (generalized) Gibbs Measures}, volume~15 of {\em
  Ens\'aios Matem\'aticos}.
\newblock SBM, 2008.

\bibitem[OP78]{OlesenPedersen:1978}
D.~Olesen and G.~K. Pedersen.
\newblock Some {$C\sp{\ast} $}-dynamical systems with a single {KMS} state.
\newblock {\em Math. Scand.}, 42(1):111--118, 1978.

\bibitem[Pfi22]{Carlinhos2022}
C.~E. Pfister.
\newblock Gibbs measures on compact ultrametric spaces. {A}r{X}iv:2202.06802,
  2022.

\bibitem[PS97]{PetSch:97}
K.~Petersen and K.~Schmidt.
\newblock Symmetric {G}ibbs measures.
\newblock {\em Trans. Amer. Math. Soc.}, 349:2775--2811, 1997.

\bibitem[Ren80]{Renault:80}
J.~Renault.
\newblock {\em A groupoid approach to C$^{\ast}$-algebras}, volume 793 of {\em
  Lecture Notes in Mathematics}.
\newblock Springer-Verlag, Berlin, Germany, 1980 edition, apr 1980.

\bibitem[Ren00]{Renault1999}
J.~Renault.
\newblock Cuntz-like {Algebras}.
\newblock In {\em Operator {Theoretical} {Methods} ({Timisoara, 1998})}, pages
  371--386. Theta {Foundation}, 2000.

\bibitem[Ren05]{renault:05}
J.~Renault.
\newblock The {R}adon{\textendash}{N}ikodym problem for approximately proper
  equivalence relations.
\newblock {\em Ergod. Theory Dyn. Syst.}, 25(05):1643, August 2005.

\bibitem[Ren09]{renault2009}
J.~Renault.
\newblock {\em C$^{\ast}$-algebras and Dynamical Systems}.
\newblock Publica{\c{c}}{\~o}es matem{\'a}ticas. IMPA, 2009.

\bibitem[RFM11]{Fer}
S.~Gallo R.~Fernandéz and G.~Maillard.
\newblock Regular {$G$}-measures are not always {G}ibbsian.
\newblock {\em Elect. Comm. in Probab.}, 16, 2011.

\bibitem[Rue04]{ruelle:04}
D.~Ruelle.
\newblock {\em {T}hermodynamic {F}ormalism}.
\newblock Cambridge University Press, Cambridge, second edition, 2004.

\bibitem[Sar09]{sarig:09}
O.~Sarig.
\newblock {L}ecture {N}otes on {T}hermodynamic {F}ormalism for {T}opological
  {M}arkov {S}hifts, 2009.

\bibitem[Sar15]{sarig:15}
O.~M. Sarig.
\newblock Thermodynamic formalism for countable {M}arkov shifts.
\newblock {\em Proc. of Symposia in Pure Math.}, 89:81--117, 2015.

\bibitem[Sch97]{schmidt:97}
K.~Schmidt.
\newblock Invariant cocycles, random tilings and the super-{K} and strong
  {M}arkov properties.
\newblock {\em Trans. Amer. Math. Soc.}, 349:2812--2825, 1997.

\bibitem[Sin72]{sinai:72}
Ya.~G. Sinai.
\newblock Gibbs measures in ergodic theory.
\newblock {\em Uspehi Mat. Nauk}, 27, 1972.

\bibitem[Sri98]{sri:98}
S.~M. Srivastava.
\newblock {\em A Course on Borel Sets}.
\newblock Springer-Verlag, New York, 1998.

\bibitem[SSWP20]{SSW2020}
A.~Sims, G.~Szab{\'o}, D.~Williams, and F.~Perera.
\newblock {\em Operator Algebras and Dynamics: Groupoids, Crossed Products, and
  Rokhlin Dimension}.
\newblock Advanced Courses in Mathematics - CRM Barcelona. Springer
  International Publishing, 2020.

\bibitem[Tho12]{Thomsen:12}
K.~Thomsen.
\newblock K{MS} states and conformal measures.
\newblock {\em Comm. Math. Phys.}, 316(3):615--640, 2012.

\bibitem[Tho16]{Thomsen2016}
K.~Thomsen.
\newblock Phase transition in {$O_2$}.
\newblock {\em Comm. Math. Phys.}, 349(2):481--492, September 2016.

\end{thebibliography}

\end{document}